\documentclass[11pt]{article}
\usepackage[margin=1in]{geometry}

\usepackage{mdframed}
\usepackage{subfig}
\usepackage{cite}
\usepackage{stfloats}
\usepackage{relsize}

\usepackage{setspace}

\usepackage{amsopn}
\usepackage{float}
\restylefloat{table}
\usepackage{amsthm}

\usepackage[font={small,it}]{caption}
\usepackage{algpseudocode,algorithm,algorithmicx}
\usepackage{amsxtra}
\usepackage{amsmath}
\usepackage{graphicx}
\usepackage{latexsym}
\usepackage{amsfonts}
\usepackage{amssymb}
\usepackage{mathrsfs}
\usepackage{bm}
\usepackage{yhmath}
\usepackage{verbatim}
\usepackage{float}
\usepackage{xr-hyper}
\usepackage{multirow}
\usepackage{subfloat}

\usepackage{epstopdf}

\usepackage{authblk}

\newtheorem{theorem}{Theorem}[section]

\newtheorem{lemma}{Lemma}[section]

\theoremstyle{definition}
\newtheorem{definition}{Definition}[section]

\newtheorem*{remark}{Remark}

\theoremstyle{condition}
\newtheorem{condition}{Condition}[section]

\DeclareMathAlphabet{\bi}{OML}{cmm}{b}{it}
\DeclareMathAlphabet{\bcal}{OMS}{cmsy}{b}{n}
\DeclareMathAlphabet{\brmn}{OT1}{cmr}{bx}{n}

\DeclareMathSymbol{\R}{\mathalpha}{AMSb}{"52}

\newcommand{\bdelta}{\boldsymbol{\delta}}

\newcommand{\bnu}{\boldsymbol{\nu}}
\newcommand{\brho}{\boldsymbol{\rho}}

\newcommand{\Lc}{\mathcal{L}}

\newcommand{\Rb}{\mathbb{R}}

\renewcommand{\o}{\omega}
\usepackage{amsfonts}

\graphicspath{{./images/}}

\DeclareMathAlphabet{\bi}{OML}{cmm}{b}{it}
\DeclareMathAlphabet{\bcal}{OMS}{cmsy}{b}{n}
\DeclareMathAlphabet{\brmn}{OT1}{cmr}{bx}{n}

\def \x{\mathbf{x}}
\def \v{\mathbf{v}}
\def \y{\mathbf{y}}

\def \e{\mathbf{e}}
\def \d{\mathbf{d}}
\def \a{\mathbf{a}}

\title{A Generalization of Wirtinger Flow for Exact Interferometric Inversion}

\author{Bariscan Yonel, Birsen Yazici}
\affil{Electrical \& Computer Systems Engineering, Rensselaer Polytechnic Institute, Troy, NY}
\providecommand{\keywords}[1]{\textbf{\textit{Key words---}} #1}

\begin{document}
\date{\vspace{-3ex}}
\maketitle

\begin{abstract}
\small{Interferometric inversion 
involves recovery of a signal from \emph{cross}-correlations of its linear transformations. 
A close relative of interferometric inversion is the generalized phase retrieval problem, which 
consists of recovering a signal from the \emph{auto}-correlations of its linear transformations. 
Recently, significant advancements have been made in phase retrieval methods 
despite the ill-posed, and non-convex nature of the problem. 
One such method is Wirtinger Flow (WF) \cite{candes2015phase}, a non-convex optimization framework that provides high probability guarantees of exact recovery under certain measurement models, such as coded diffraction patterns, and Gaussian sampling vectors. 
In this paper, we develop a generalization of WF for interferometric inversion, which we refer to as Generalized Wirtinger Flow (GWF). 
GWF theory extends the probabilistic exact recovery results in \cite{candes2015phase} to arbitrary measurement models characterized in the equivalent lifted problem, hence covers a larger class of measurement models. 
GWF framework unifies the theory of low rank matrix recovery (LRMR) and the non-convex optimization approach of WF, thereby establishes theoretical advantages of the non-convex approach over LRMR.  
We show that the conditions for exact recovery via WF can be derived through a low rank matrix recovery formulation. 
We identify a new sufficient condition on the lifted forward model that directly implies exact recovery conditions of standard WF.
This condition is less stringent than those of LRMR, which is the state of the art approach for exact interferometric inversion.
We next establish our sufficient condition for the cross-correlations of linear measurements collected by complex Gaussian sampling vectors. 
In the particular case of the Gaussian model, we show that the exact recovery conditions of standard WF imply our sufficient condition, and that the regularity condition of WF is redundant for the interferometric inversion problem. 
Finally, we demonstrate the effectiveness of GWF numerically in a deterministic multi-static radar imaging scenario.} 


\end{abstract}



\section{Introduction} \label{sec:Intro}

\emph{Interferometric inversion} involves the recovery of a signal of interest from the cross-correlations of its linear measurements, each collected by a different sensing process.
Let $\mathbf{L}^m_i, \mathbf{L}^m_j \in \mathbb{C}^N$ denote the $m^{th}$ sampling vectors of the $i^{th}$ and $j^{th}$ sensing processes and $\brho_t \in \mathbb{C}^N$ be the ground truth/signal of interest.
We define
\begin{equation}
f^m_i = \langle \mathbf{L}^m_i , \brho_t \rangle, \quad f^m_j = \langle \mathbf{L}^m_j , \brho_t \rangle, \quad m = 1, \cdots, M,
\end{equation}
as the linear measurements and describe the cross-correlated measurements as
\begin{equation}\label{eq:interferom}
d^m_{ij}  = f^m_i  \overline{f^m_j} =  (\mathbf{L}^m_i )^H \brho_t \brho_t^H \mathbf{L}^m_j \quad m = 1, \cdots M,
\end{equation}
where $\overline{( \cdot )}$ denotes complex conjugation. Thus, interferometric inversion involves recovery of $\brho_t \in \mathbb{C}^N$  from $d^m_{ij} \in \mathbb{C}, \ m=1,...,M$ using the model in \eqref{eq:interferom}.

Interferometric inversion problem arises in many applications in different disciplines. 
These include radar and sonar interferometry \cite{goldstein1987interferometric, saebo2010seafloor, bamler1998synthetic}, passive imaging in acoustic, electromagnetic and geophysical applications \cite{Yarman08, Wang11, Wang10, yarman2008bistatic, Yarman10, LWang12_b, Mason2015, son2015passive, Wang13_3, Wang14,  LWang12, wang12, son2017passive, Wacks14, Wacks14_2, ammari2013passive}, interferometric microscopy \cite{ralston2007interferometric}, beamforming and sensor localization in large area networks \cite{stoica2007probing}, \cite{patwari2005locating, gezici2005localization, guivant2001optimization} among others. 
Additionally, cross-correlations were shown to provide robustness to statistical fluctuations in scattering media or incoherent sources in wave-based imaging \cite{garnier2005imaging, lobkis2001emergence}, and with respect to phase errors in the correlated linear transformations \cite{blomgren2002super, gough2004displaced, flax1988phase, mason2016robustness}. 
Therefore, in applications such as passive imaging \cite{Yarman08, Wang11, Wang10, yarman2008bistatic, Yarman10, LWang12_b, Mason2015, son2015passive, Wang13_3, Wang14,  LWang12, wang12, son2017passive, Wacks14, Wacks14_2, ammari2013passive} and interferometry \cite{goldstein1987interferometric, saebo2010seafloor, bamler1998synthetic}, cross-correlations are formed as a part of the inference process after acquiring linear measurements by sensors that are configured differently in space, time or frequency. 
Additionally, the cross-correlated measurement model arises naturally from the underlying physical sensing processes in certain applications such as optical and radio astronomy \cite{jain2003cross, copeland2006dynamics}, or quantum optical imaging \cite{schotland2010quantum}.

A special case of the interferometric inversion problem is when $i = j$ in \eqref{eq:interferom}, in which the model becomes the auto-correlations of linear measurements collected by a single sensing process.
In this case, the interferometric inversion problem reduces to the well-known \emph{phase retrieval} problem.
Notably, both problems are non-convex due to the quadratic equality constraints enforced by the correlated measurement model.
In recent years, several phase retrieval methods with exact recovery guarantees have been developed despite the non-convex nature of the problem.
These methods are characterized by either one or both of the following two principles: convexification of the solution set, which includes \emph{lifting} based approaches \cite{Candes13a, Candes13b, Waldspurger2015}, or a provably accurate initialization, followed by an algorithmic map that refines the initial estimate which is most prominently established by Wirtinger Flow (WF) \cite{candes2015phase} and its variants \cite{chen2015solving, chen2017, Zhang2017a, zhang2016reshaped, Zhang2017b, wang2018solving, bendory2018non, soltanolkotabi2019structured}. 

In methods that deploy lifting, such as PhaseLift \cite{Candes13a, Candes13b}, signal recovery from quadratic measurements is reformulated as a low rank matrix recovery (LRMR) problem.
While the LRMR approach offers convergence guarantees via convexification, it has limited practical applicability in typical sensing and imaging problems since lifting increases the dimension of the inverse problem by an order of magnitude and requires demanding memory storage in implementation.
The WF framework, on the other hand, avoids lifting hence offers considerable advantages in computational complexity and memory requirements.
Despite solving the non-convex problem directly, convergence to a global solution at a geometric rate is guaranteed by WF for coded diffraction patterns and Gaussian sampling vectors \cite{candes2015phase}, and more recently for short time Fourier transforms \cite{bendory2018non}.  
These advantages promote WF as a theoretical framework suitable for large scale imaging problems.

Conventionally, interferometric inversion in imaging applications has been approached by Fourier based techniques, such as time or frequency difference of arrival (TDOA/FDOA) backprojection \cite{Yarman08, Yarman10, Wang11, wang12, Wang13_3, wacks2018doppler, Wacks14, Wacks14_2}.
While these methods are practical and computationally efficient, their applicability is limited to scenes composed of well-separated point targets due to underlying assumptions. 
As an alternative, LRMR theory has been explored for interferometric inversion \cite{Mason2015, Demanet13}.
In particular, \cite{Demanet13} assesses the robustness of solutions to several convex semi-definite programs (SDPs) in a deterministic framework after lifting the interferometric measurements. 
Notably, these SDP solvers are inspired by the PhaseLift method \cite{Candes13a, Candes13b, Demanet14}, hence suffer from the same drawbacks in practice. 
In \cite{Mason2015}, an iterative optimization approach to LRMR was developed for interferometric passive imaging to circumvent the poor scaling properties of SDP approaches. 
The methodology in \cite{Mason2015} is based on Uzawa's method for matrix completion \cite{cai2010singular, recht2010guaranteed}. 
While this method is computationally more efficient than the SDP solvers, it still operates on the lifted domain, hence requires significant memory and computational resources. 
Additionally, these convexified lifting based solvers require stringent theoretical conditions on the measurement model, 
which poses a major theoretical barrier for interferometric inversion problems with deterministic forward models. 



In this paper, motivated by its advantages over lifting based methods, 
we develop a generalization of WF applicable to interferometric inversion problems with exact recovery guarantees.  We refer to this method as the Generalized Wirtinger Flow (GWF). 
Beyond the immediate extension of algorithmic principles of WF to the interferometric inversion, our mathematical framework differs from that of WF in two significant ways. 
\emph{i}) The GWF theory is established in the lifted domain. 
This lifting-based perspective allows us to develop a novel approach that bridges the theory between LRMR and the non-convex methodology of WF, and culminates into derivation of a new sufficient condition for exact recovery
in the context of interferometric inversion.   
Namely, GWF guarantees exact recovery to a general class of problems that are characterized over the equivalent lifted domain by the restricted isometry property (RIP) on the set of rank-1, positive semi-definite (PSD) matrices. 
Specifically, we provide an upper bound for the \emph{restricted isometry constant} (RIC) over the set of rank-1, PSD matrices so that RIP directly implies convergence to a global solution. 
\emph{ii}) Given the sufficient condition, we derive exact recovery guarantees for an arbitrary mapping over the lifted domain using solely geometric arguments in a deterministic framework, rather than model specific probabilistic arguments of standard WF. 
Notably, our analysis proves that the regularity condition, which forms the basis of convergence in standard WF, is redundant for the case of interferometric inversion. 


In addition, developing the GWF framework through the equivalent lifted problem allows us to identify the key theoretical advantages of the non-convex optimization approach over lifting based convex solvers beyond the immediate gains in computation and applicability. 
Our exact recovery condition for GWF proves to be a less stringent restricted isometry property than that of the sufficient conditions of LRMR methods.  
As a result of its fundamental differences from standard WF, and less stringent exact recovery requirements than those of LRMR methods, GWF offers a promising step towards exact recovery guarantees for interferometric inversion problems which are governed by deterministic measurement models. 
One such application is array imaging, for which RIP was shown to hold over rank-1 matrices for sufficiently high central frequencies, albeit asymptotically in the number of receivers \cite{Chai11}. 
Accordingly, we consider GWF as an exact interferometric imaging method alternative to LRMR for multi-static radar imaging, and design imaging system parameters for the lifted forward model to satisfy of our sufficient condition in the non-asymptotic regime \cite{son2019exact}.



In our key results, 
we first prove that the RIP over rank-1, PSD matrices on the lifted forward model implies an accurate initialization by the spectral method, and ensures that the regularity condition is satisfied in the $\epsilon$-neighborhood defined by the initialization if RIC is less than $0.214$.
Next, although not applicable to the case of auto-correlations in \cite{candes2015phase}, we show that this sufficient condition is satisfied for the cross-correlation of linear measurements collected by i.i.d. Gaussian sampling vectors. 
Essentially, our results establish that the probabilistic arguments used for the Gaussian sampling model in \cite{candes2015phase} are fully captured by the single event that the RIP is satisfied over rank-1, PSD matrices in the case $i \neq j$ in \eqref{eq:interferom}. 
To validate our analysis we conduct numerical experiments for the Gaussian sampling model by counting empirical probability of exact recovery.
We then demonstrate the effectiveness of GWF in a realistic passive multi-static radar imaging scenario. 
Our preliminary numerical simulations confirm our theoretical results and show that the interferometric inversion is solved in an exact manner by GWF.


The rest of our paper is organized as follows. 
We first introduce the problem formulation for interferometric inversion on the signal domain, and discuss lifting based prior art in Section \ref{sec:Sect2}. 
We then formulate the GWF algorithm in Section \ref{sec:Alg}, present key definitions and terminology followed by the main theorem statements in Section \ref{sec:Results}. 
The proofs of the main theorems are provided in Section \ref{sec:Proofs}. 
The numerical simulations for Gaussian sampling model and interferometric multi-static radar imaging are presented in Section \ref{sec:NumSim}.
Section \ref{sec:Conc} concludes our paper. Appendices \ref{sec:AppDerivs}, \ref{sec:AppThm1}, and \ref{sec:AppThm2} include proofs of lemmas used in Section \ref{sec:Results}.

\section{Problem Formulation and Prior Art}\label{sec:Sect2}
In this section, we introduce the non-convex formulation of the interferometric inversion problem and its key challenges. Next we discuss lifting based, convex formulations which address these key challenges in solving quadratic systems of equations via low rank matrix recovery theory. 
\subsection{The Non-Convex Objective Function}\label{sec:NonConv}

To address the interferometric inversion problem, we define the following objective function, and set up the corresponding optimization problem:
\begin{eqnarray}
& \mathcal{J}(\brho)  :=  \frac{1}{2M}  \sum_{m = 1}^M | \left(\mathbf{L}_i^m\right)^H \brho \brho^H \mathbf{L}_j^m - d_{ij}^m |^2, \label{eq:GWF_opt} \\
& \ \ \ \hat{\brho}  = \underset{\brho}{\text{argmin}} \ \mathcal{J}(\brho). \label{eq:GWF_argmin}
\end{eqnarray}
Let $\y = [ y_{ij}^1,  y_{ij}^2, \cdots  y_{ij}^M]^T$ and $\mathcal{L}: \mathbb{C}^N \rightarrow \mathbb{C}^M$ be the cross-correlated measurement map defined as
\begin{equation}\label{eq:MeasMap}
\mathcal{L}(\brho) = \y, \ \text{where} \ \brho \in \mathbb{C}^N, \  \ y_{ij}^m = \left(\mathbf{L}_i^m\right)^H \brho \brho^H \mathbf{L}_j^m.
\end{equation}

The objective function in \eqref{eq:GWF_opt} is the $\ell_2$ mismatch in the range of $\mathcal{L}$, i.e., the space of cross-correlated measurements, which is solved over the signal domain $\mathbb{C}^N$ in \eqref{eq:GWF_argmin}. 
The difficulty in \eqref{eq:GWF_argmin} is that the objective function $\mathcal{J}$ is non-convex over the variable $\brho$ due to the invariance of cross-correlated measurements to global phase factors. 
Essentially \eqref{eq:GWF_argmin} has a non-convex solution set with infinitely many elements, which casts interferometric inversion a challenging, ill-posed problem.  
\begin{definition}[Global Solution Set]\label{eq:SolutionSet}
We say that the points
\begin{displaymath}
\mathit{P} := \{ e^{\mathrm{i} \phi} \brho_t : \phi \in [0, 2\pi) \},
\end{displaymath}
form the global solution set for the interferometric inversion from the cross-correlated measurements \eqref{eq:interferom}.
\end{definition}

More generally, 
for any $\brho \in \mathbb{C}^N$, let $\mathcal{E}_{\brho} = \{ \mathbf{z} \in \mathbb{C}^N: \mathcal{L}(\mathbf{z}) = \mathcal{L}(\brho) \}$ be the equivalence class of $\brho$ under $\mathcal{L}$. 
We then define the following collection of signals as the \emph{equivalence set} of $\brho$. 

\begin{definition}[Equivalence Set]\label{def:EqSet}
Let $\brho \in \mathbb{C}^N$ and 
\begin{equation}
\mathit{P}_{\brho} := \{ e^{\mathrm{i} \phi} \brho, \phi \in [0, 2\pi) \},
\end{equation}
We refer to $\mathit{P}_{\brho}$ as the equivalence set of $\brho$.
\end{definition}

\begin{remark}
Note that $\mathit{P}_{\brho} \subset \mathcal{E}_{\brho}$; and $\mathit{P}_{\brho_t}$ is identical to the global solution set in Definition \eqref{eq:SolutionSet}.
\end{remark}

Alleviating the non-injectivity 
of the measurement map is a key step in formulating methods that guarantee exact recovery in phase retrieval literature \cite{bandeira2014saving}, and offers us a blueprint in addressing \eqref{eq:GWF_argmin}. 
A key observation is that one can consider \eqref{eq:MeasMap} as a mapping from a rank-1, positive semi-definite matrix $\brho \brho^H \in \mathbb{C}^{N \times N}$ instead of a quadratic map from the signal domain in $\mathbb{C}^N$, and attempt to recover $\brho_t \brho_t^H$.
This approach is known as the \emph{lifting} technique which is the main premise of LRMR based phase retrieval \cite{Candes13a, Candes13b, Waldspurger2015, candes2015phase2} and interferometric inversion methods \cite{Mason2015, Demanet13, Demanet14}. 

\subsection{Low Rank Matrix Recovery via Lifting}\label{sec:Lifting}
We adopt the concepts of the LRMR approach to the interferometric inversion problem based on PhaseLift \cite{Mason2015}, \cite{Candes13b}, and introduce the following definitions.
\begin{definition}{\emph{Lifting}.}
Each correlated measurement in \eqref{eq:interferom} can be written in the form of an inner product of two rank-1 operators, $\tilde{\brho}_t = \brho_t \brho_t^H$ and $\overline{\mathbf{F}}^m = \mathbf{L}_j^m (\mathbf{L}_i^m)^H$ such that\footnote{$(\bar{\cdot})$ denotes element-wise complex conjugation}
\begin{equation}\label{eq:PhaseLift}
d_{ij}^m = \langle {\mathbf{F}}^m , \tilde{\brho}_t \rangle_F \quad m = 1, . . ., M
\end{equation}
where $\langle \cdot, \cdot \rangle_F$ is the Frobenius inner product. We refer to the procedure of transforming interferometric inversion over $\mathbb{C}^N$ to the recovery of the rank-1 unknown $\tilde{\brho}_t$ in $\mathbb{C}^{N \times N}$ as {lifting}.
\end{definition}

The lifting technique introduces a new linear measurement map 
which we define as follows:
\begin{definition}{\emph{Lifted Forward Model}.}
Let $\d = [d_{ij}^1, d_{ij}^2, \cdots d_{ij}^M ] \in \mathbb{C}^M$ denote the vector obtained by stacking the cross-correlated measurements in \eqref{eq:interferom}. Then using \eqref{eq:PhaseLift}, we define $\mathcal{F}: \mathbb{C}^{N \times N} \rightarrow \mathbb{C}^M$ as follows
\begin{equation}\label{eq:LiftedLinMod}
\mathbf{d} = \mathcal{F} (\tilde{\brho}_t)
\end{equation}
and refer to $\mathcal{F}$ as the {lifted forward model}/map.
\label{def:Def2.5}
\end{definition}
\begin{remark}
The map $\mathcal{F}$ can be interpreted as a $M \times N^2$ matrix with $\bar{\mathbf{F}}^m$ as its rows in the vectorized problem in which $N \times N$ variable $\tilde{\brho}_t$ is concatenated into a vector.
\end{remark}

We refer to the problem of recovering of $\tilde{\brho}_t$ from $\d$ using the model \eqref{eq:LiftedLinMod} as the lifted problem, or interferometric inversion in the lifted domain.

The main advantage of lifting is that for all $\brho \in \mathbb{C}^N$, each non-convex equivalence set $\mathit{P}_{\brho}$ is now mapped to a set with a single element $\tilde{\brho} = \brho \brho^H$.
Using the definition of the lifted forward model, quadratic equality constraints reduce to affine equality constraints to define a convex manifold in $\mathbb{C}^{N \times N}$.
Using the \emph{rank-1 positive semi-definite (PSD)} structure of the unknown $\tilde{\brho}_t$, interferometric inversion in the lifted domain can be formulated as the following optimization problem:
\begin{equation}\label{eq:PLift_orig}
\text{find:} \ \mathbf{X}  \quad \text{s.t. } \ \mathbf{d} = \mathcal{F} ( \mathbf{X} ), \ \mathbf{X} \succeq 0, \  \text{rank}(\mathbf{X}) = 1.
\end{equation}
Here, we refer to $\mathbf{X} \in \mathbb{C}^{N \times N}$ as the lifted variable.
Due to the fact that there surely exists a rank-1 solution from \eqref{eq:LiftedLinMod}, \eqref{eq:PLift_orig} is equivalent to:
\begin{equation}\label{eq:PLift}
\text{minimize:} \ \text{rank}(\mathbf{X}) \quad \text{s.t. } \ \mathbf{d} = \mathcal{F} ( \mathbf{X} ), \ \mathbf{X} \succeq 0,
\end{equation}
which is known to be an NP-hard problem \cite{cai2010singular, recht2010guaranteed}.
Given its rank-minimization form, \eqref{eq:PLift} is approached by low rank matrix recovery theory analogous to compressive sensing \cite{donoho2006compressed, candes2006stable}. 
Most prominently, the non-convex rank term of the objective function is relaxed by a convex surrogate, which under the positive semi-definite constraint, corresponds to the trace norm. This results in the following formulation \cite{Chai11, Candes13a, Candes13b}:
\begin{equation}\label{eq:PLiftConv}
\text{minimize:} \ \text{tr}(\mathbf{X}) \quad \text{s.t. } \ \mathbf{d} = \mathcal{F} ( \mathbf{X} ), \ \mathbf{X} \succeq 0,
\end{equation}
which can be solved in polynomial time via semi-definite programming.  
To obtain a robust program, \eqref{eq:PLiftConv} is commonly \emph{perturbed} using the least squares criterion, which under the additive i.i.d. noise assumption, can be written as
\begin{equation}\label{eq:PLiftConv_2}
{\text{minimize:}} \  \text{tr}(\mathbf{X}), \quad \text{s.t. }  \frac{1}{2} \| \mathcal{F}(\mathbf{X}) - \d \|_2^2 \leq \mathcal{O}(\sigma^2), \ \mathbf{X} \succeq 0,
\end{equation}
where $\sigma^2$ corresponds to the variance in the case of Gaussian noise, and the order in the threshold can be tuned to obtain a desired lower bound on the log-likelihood of observing $\d$. 
Alternatively, one can equivalently formulate the Lagrangian of \eqref{eq:PLiftConv_2} with a proper choice of the regularization parameter $\lambda$ in relation to the threshold, as follows: 
\begin{equation}\label{eq:PerturbedLift}
\underset{\mathbf{X} \succeq 0 }{\text{minimize:}} \ \frac{1}{2} \| \mathcal{F}(\mathbf{X}) - \d \|_2^2 + \lambda \text{tr}(\mathbf{X}).
\end{equation}
\eqref{eq:PerturbedLift} can be solved by Uzawa's method \cite{cai2010singular, recht2010guaranteed} which is analogous to the singular value thresholding algorithm with a PSD constraint \cite{Mason2015} with the following iterations:
\begin{eqnarray}
\mathbf{X}_k & = & \mathcal{P}_{+} \circ \mathcal{P}_{\tau_k} ( \mathcal{F}^H \bnu_{k-1} ) \label{eq:SVT} \\
\bnu_k  & = &   \bnu_{k-1} + \mu_k (\mathbf{d} - \mathcal{F} (\mathbf{X}_k) ) \label{eq:SVT2}.
\end{eqnarray}
In \eqref{eq:SVT} and \eqref{eq:SVT2}, $\bnu = [\nu^1, \cdots \nu^M]^T$ denote the Lagrange multipliers initialized as $\bnu_0 = 0$, $\mu_k$ is the step size, subscript $k$ denotes the iteration number. 
$\mathcal{P}_{\tau}$ is the shrinkage operator acting on the singular values of its argument with threshold $\tau_k = \mu_k \lambda$, which through the parameter $\lambda$ of trace regularization, enforces the low-rank constraint.
 $\mathcal{P}_+$ is the projection operator onto the PSD cone. 

As $M \ll N^2$ in typical estimation problems, the lifted forward model has a non-trivial null space.
LRMR theory encapsulates identifying necessary and sufficient conditions on $\mathcal{F}$ in order to guarantee exact recovery despite having an under-determined system of equations in \eqref{eq:LiftedLinMod}.
The key conditions on $\mathcal{F}$ are primarily characterized by its null space \cite{recht2008necessary, recht2011null, oymak2011simplified} or restricted isometry properties \cite{recht2010guaranteed, cai2013sharp, bhojanapalli2016global} on low rank matrices.
Methods such as PhaseLift \cite{Candes13a, Candes13b} and PhaseCut \cite{Waldspurger2015} assert conditions on the mapping $\mathcal{F}$ such that there exists no feasible element in the PSD cone with a smaller trace than the true solution $\tilde{\brho_t}$, from which exact recovery results to a unique minimizer are directly implied by the standard arguments of semi-definite programming \cite{Candes13b}.
Iterative optimization via Uzawa's method on the other hand requires RIP over rank-5 matrices with a sufficiently small RIC ($\leq {1}/{10}$) for the convex problem in \eqref{eq:PerturbedLift} to have the identical solution to the original non-convex problem in \eqref{eq:PLift}. 
Notably, in this scenario the PSD constraint can simply be dropped due to the guaranteed uniqueness of the solution to the original problem, and the convexity of the problem in \eqref{eq:PerturbedLift}. 
Furthermore, under the properties on $\mathcal{F}$ asserted by PhaseLift, it is observed in \cite{Demanet14} that the lifted problem can be robustly solved as a convex feasibility problem by Douglas-Rachford splitting by eliminating the trace minimization step completely.
These indicate a level of redundancy between trace regularization and the PSD constraint, which arises from the convexification of the problem. 

Altogether, lifting based approaches provide a profound perspective to the interferometric inversion problem. 
Our observation is that the key principles of lifting based methods in establishing exact recovery guarantees are reciprocated in the non-convex framework of WF.
In fact, WF corresponds to solving a perturbed non-convex feasibility problem over the lifted domain, and in this sense, is reminiscent of the optimizationless PhaseLift method of \cite{Demanet14}, and Uzawa's iterations in \cite{Mason2015}.
To observe this, we introduce the GWF iterations for interferometric inversion, and develop the method as a solver in the lifted problem framework.
The basis of our extension from WF to GWF framework is the identification of conditions on the lifted forward model for the exactness of a solution in the lifted domain.

It is worth noting that low rank models and LRMR theory cover a rich area of research with wide range of applications in fields such as control theory and machine learning. 
Several of the LRMR methods discussed in this subsection were first proposed for matrix completion, and relate strongly to the problem of low rank matrix factorization.  
Non-convex optimization theory for low rank matrix factorization has undergone notable developments in recent years, with methods that offer exact recovery results for quadratic or bilinear equations of rank-$r$ matrices in the random Gaussian model \cite{zheng2015convergent}, or if the measurement map satisfies RIP over rank-$6r$ with a RIC $\leq 1/10$ \cite{tu2015low}.  
For further discussion on advances in non-convex low rank matrix recovery, we refer the reader to \cite{wang2016unified, chi2018nonconvex}.  




\label{sec:Obj}


\section{Generalized Wirtinger Flow Framework for Exact Interferometric Inversion}\label{sec:Alg}
In this section, we present the algorithmic principles of the GWF framework for exact interferometric inversion. We specifically identify and present the theoretical advantages obtained from viewing the algorithm in the lifted domain, despite operating on the signal domain. 
We first present the GWF iterations, then proceed with generalizing the spectral method for initialization using our lifted formulation of GWF. 
We finally provide an algorithm summary with specifications of computational complexity of each step. 

\subsection{GWF Iterations}
In presenting the GWF iterations, we begin by extending the iterative scheme of WF to the case of interferometric measurement model in \eqref{eq:interferom}, with $i \neq j$. 
We next introduce an equivalent, novel interpretation of these iterations in the lifted domain, which yields the basis of our GWF framework for exact recovery with arbitrary lifted forward models. 

\subsubsection{Extending WF}
In contrast to the lifting based approaches, the non-convex form of the problem in \eqref{eq:GWF_argmin} is preserved in WF. 
Given an \emph{accurate} initial estimate $\brho_0$, WF involves using the following updates to refine the current estimate $\brho_{k}$:
\begin{equation}\label{eq:GWF_Updates}
\brho_{k+1} = \brho_{k} - \frac{\mu_{k + 1}}{\| \brho_0 \|^2} \nabla \mathcal{J}(\brho_{k}).
\end{equation}
Notably, $\mathcal{J}$ 
is a real-valued function of a complex variable $\brho \in \mathbb{C}^N$, and therefore, non-holomorphic.
Hence, the gradient over $\mathcal{J}$ is defined by the means of Wirtinger derivatives 
\begin{eqnarray}\label{eq:compGrad}
\nabla \mathcal{J} &=& \left(\frac{\partial \mathcal{J}}{\partial \brho}\right)^H = \left(\frac{\partial \mathcal{J}}{\partial \overline{\brho}}\right)^T, \ \text{where} \\
\frac{\partial}{\partial \brho} &= &\frac{1}{2}( \frac{\partial}{\partial \brho_R} - \mathrm{i} \frac{\partial}{\partial \brho_I} ), \quad \frac{\partial}{\partial \overline{\brho}} = \frac{1}{2}( \frac{\partial}{\partial \brho_R} + \mathrm{i} \frac{\partial}{\partial \brho_I} ),
\end{eqnarray}
and $\brho = \brho_R + \mathrm{i} \brho_I$, with $\brho_R, \brho_I \in \mathbb{R}^N$.
Thus, the iterations in \eqref{eq:GWF_Updates} correspond to that of the steepest descent method \cite{candes2015phase}, where $\mu_{k+1}$ is the step size.
For interferometric inversion by solving \eqref{eq:GWF_argmin}, $\nabla \mathcal{J}$ evaluated at $\brho_k$ is given by
\begin{equation}\label{eq:UpdateTerm}
\nabla \mathcal{J}( \brho_k ) = \frac{1}{2M} \sum_{ m = 1}^M \left[ \overline{e_{ij}^m} (\mathbf{L}_j^m \left(\mathbf{L}_i^m \right)^H \brho_{k}) +  e_{ij}^m (\mathbf{L}_i^m \left(\mathbf{L}_j^m \right)^H \brho_{k}) \right],
\end{equation}
where $e_{ij}^m = (\left(\mathbf{L}_i^m \right)^H \brho_{k} \brho_{k}^H \mathbf{L}_j^m - d_{ij}^m)$ is the mismatch between the synthesized and cross-correlated measurements.
Note that in the case of phase retrieval we have $i = j$. Letting $\mathbf{L}^m = \mathbf{L}_i^m = \mathbf{L}_j^m$ and $ d^m = | \langle \mathbf{L}^m ,  \brho_t \rangle |^2$, for $m = 1 \cdots M$, \eqref{eq:UpdateTerm} reduces to the standard WF iterations of \cite{candes2015phase}.
See Appendix \ref{sec:App0} for the derivation of $\nabla \mathcal{J}$ for $i \neq j$. 



\subsubsection{Interpretation of GWF Updates in the Lifted Domain}
Our formulation of GWF framework resides in the lifted domain, and reveals an illuminating interpretation of the updates in \eqref{eq:UpdateTerm}. 
Moving the common term of the current iterate $\brho_k$ outside the summation, \eqref{eq:UpdateTerm} can be expressed as
\begin{equation}\label{eq:GradTerm}
\nabla \mathcal{J}( \brho_k ) = \frac{1}{M} \left[ \frac{1}{2} \left(\sum_{ m = 1}^M \overline{e_{ij}^m} (\mathbf{L}_j^m \left(\mathbf{L}_i^m \right)^H) + \sum_{ m = 1}^M e_{ij}^m (\mathbf{L}_i^m \left(\mathbf{L}_j^m \right)^H) \right) \right] \brho_k.
\end{equation}
From the definition of the lifted forward model in \eqref{eq:LiftedLinMod}, the second term inside the brackets in \eqref{eq:GradTerm} becomes the \emph{backprojection} of the measurement error onto 
the adjoint space of $\mathcal{F}$. 
\begin{definition}{\emph{Backprojection.}} Let $\y = [y^1, y^2, \cdots y^M ] \in \mathbb{C}^M$. For the lifted forward model $\mathcal{F}: \mathbb{C}^{N \times N} \rightarrow \mathbb{C}^M$ in Definition \ref{def:Def2.5}, we define the adjoint operator $\mathcal{F}^H: \mathbb{C}^M \rightarrow \mathbb{C}^{N \times N}$ as, 
\begin{equation}\label{eq:BackProject}
\mathcal{F}^H (\y ) = \sum_{ m = 1}^M y^m  (\mathbf{L}_i^m \left(\mathbf{L}_j^m \right)^H),
\end{equation}
and refer to \eqref{eq:BackProject} as the backprojection of $\y$. 
\end{definition}
Since \eqref{eq:GradTerm} consists of the average of terms that are Hermitian transposes of each other, the update term in \eqref{eq:GWF_Updates} corresponds to backprojecting the mismatch between synthesized and true measurements, and then projecting it onto the set of symmetric matrices $S$, i.e.,
\begin{equation}\label{eq:UpdtTerm_Lift}
\nabla \mathcal{J} ( \brho_k ) = \frac{1}{M} \mathcal{P}_{S} \left( \mathcal{F}^{H} ( {\mathbf{e}} ) \right) \brho_k,
\end{equation}
where $\mathbf{e} = [ e_{ij}^1, e_{ij}^2, \cdots e_{ij}^M ] \in \mathbb{C}^M$ is the measurement mismatch vector and $\mathcal{P}_{S}( \cdot )$ denotes the projection operator onto $S$. 

The representation in \eqref{eq:UpdtTerm_Lift} provides a novel perspective in interpreting GWF as a solver of the lifted problem.
Consider the structure enforced by 
\eqref{eq:GWF_Updates} for the non-relaxed form of the low rank recovery problem in \eqref{eq:PLift}.
In the GWF updates, the rank of the lifted variable $\mathbf{X}$ in \eqref{eq:PLift} is merely fixed at one, and the rank minimization problem is converted to its original non-convex feasibility problem in \eqref{eq:PLift_orig}. 
Knowing that \eqref{eq:GWF_opt} corresponds to the $\ell_2$ mismatch in the space of cross-correlated measurements, the GWF solver of \eqref{eq:GWF_opt} is equivalently the solver of the following lifted problem:
\begin{equation}\label{eq:perturbRank1}
\text{minimize: } \ \frac{1}{2M} \| \mathcal{F} ( \mathbf{X} ) - \mathbf{d} \|_2^2 \quad \text{s.t. } \ \text{rank} (\mathbf{X} ) = 1 \ \text{and} \ \mathbf{X} \succeq 0.
\end{equation}


Observe that the rank-1, PSD constraint precisely corresponds to minimization over the set of elements of $\mathcal{X} = \{ \brho \brho^H, \brho \in \mathbb{C}^N \}$ 
and that \eqref{eq:perturbRank1} can be equivalently cast as:
\begin{equation}\label{eq:WFform}
\text{minimize: } \ \frac{1}{2M} \| \mathcal{F} ( \mathbf{X} ) - \mathbf{d} \|^2_2 \quad \text{s.t.} \quad \mathbf{X} = \brho \brho^H.
\end{equation}
The updates to solve \eqref{eq:WFform} then can be performed on the leading eigenspace of the lifted variable $\mathbf{X}$ directly by means of the \emph{Jacobian} $\frac{\partial \bar{\mathbf{X}}}{\partial \bar{\brho}}$.
Since $\mathbf{X} = \mathbf{X}^H$, this yields\footnote{We use the property of Wirtinger derivatives in writing the update \eqref{eq:alternativeInt}, such that the derivative of a real valued function of a complex variable has the property that $\overline{(\frac{\partial \mathcal{J}}{\partial \mathbf{X}} )} = {(\frac{\partial\mathcal{J}}{\partial \overline{\mathbf{X}}} )}$}
\begin{equation}\label{eq:alternativeInt}
\brho_{k+1} = \brho_k - \mu_{k+1} \left( \frac{\partial \overline{\mathbf{X}}}{\partial \bar{\brho}} (\frac{\partial \mathcal{J}}{\partial \overline{\mathbf{X}}} + \frac{\partial \mathcal{J}}{\partial \mathbf{X}}) \right)^T, \text{where} 
\end{equation}
\begin{equation} \label{eq:alternativeInt2}
\left(\frac{\partial \overline{\mathbf{X}}}{\partial \overline{\brho}} \frac{\partial \mathcal{J}}{\partial \overline{\mathbf{X}}} \right)^T =  \frac{1}{2M} ( \mathcal{F}^H \mathcal{F} ( \mathbf{X} ) - \mathcal{F}^H \mathbf{d}) \brho, \ \left(\frac{\partial \overline{\mathbf{X}}}{\partial \overline{\brho}} \frac{\partial \mathcal{J}}{\partial {\mathbf{X}}} \right)^T =  \frac{1}{2M} ( \mathcal{F}^H \mathcal{F} ( \mathbf{X} ) - \mathcal{F}^H \mathbf{d})^H \brho.
\end{equation}
Initializing the algorithm with a proper estimate in the constraint set $\mathcal{X}$, i.e. $\mathbf{X}_0 = \brho_0 \brho_0^H$, and substituting $\mathbf{X} = \brho \brho^H$, we precisely obtain $\nabla \mathcal{J}(\brho_k)$ derived in \eqref{eq:UpdtTerm_Lift} when \eqref{eq:alternativeInt} is evaluated at $\brho = \brho_k$. 


\label{sec:Upd}

\subsection{The Distance Metric and Equivalence of Convergence in Signal and Lifted Domains}

While the view of GWF in the lifted problem and the formulation in \eqref{eq:WFform} are illuminating, the algorithmic map of GWF operates exclusively on the signal domain in $\mathbb{C}^{N}$.
The duality between the lifted domain and the signal domain is established by the distance metric of WF framework which is defined as follows \cite{candes2015phase}:
\begin{definition}\label{def:DistMet}
Let $\brho_t \in \mathbb{C}^N$ be an element of the global solution set in \eqref{eq:SolutionSet}. The distance of an element $\brho \in \mathbb{C}^N$ to $\brho_t$ is defined as \cite{candes2015phase}
\begin{equation}\label{eq:distance}
\mathrm{dist}(\brho, {\brho}_t) = \| \brho - e^{\mathrm{i}\Phi(\brho)} {\brho}_t \|, \ \text{where} \ \ \Phi(\brho) : = \underset{\phi \in [0,2\pi]}{\mathrm{arg min}} \| \brho - {\brho}_t e^{\mathrm{i}\phi} \|.
\end{equation}
\end{definition}
For the purpose of convergence, the metric implies that we are primarily interested in the distance of an estimate $\brho$ to \emph{any} of the elements in the non-convex solution set $\mathit{P}$.
In technical terms, invoking Definition \ref{def:EqSet}, \eqref{eq:distance} is a measure of distance between the equivalence sets $\mathit{P}_{\brho}$ and $\mathit{P}$. 
By Definition \ref{def:DistMet}, the ambiguity due to the invariance of the cross-correlation map, $\mathcal{L}$, to the global phase factors is evaded on $\mathbb{C}^N$, without lifting the problem.
Observe that the phase ambiguity is indeed removed analytically, since the $\ell_2$ norm is minimized when $\mathrm{Re} (\langle \brho , e^{\mathrm{i} \Phi(\brho) } \brho_t  \rangle ) = | \langle \brho , e^{j \Phi(\brho) } \brho_t  \rangle | =  | \langle \brho , \brho_t  \rangle | $, which is achieved at $e^{\mathrm{i} \Phi(\brho) }  = \frac{ \langle \brho_t , \brho \rangle}{| \langle \brho , \brho_t  \rangle | }$.
Hence, the squared distance becomes 
\begin{equation}\label{eq:distmetric}
\mathrm{dist}^2 (\brho, \brho_t) =  \| \brho \|^2 + \| \brho_t \|^2 - 2 | \langle \brho , \brho_t  \rangle |,
\end{equation}
and is independent of any global phase factor on $\brho$ or $\brho_t$.
Geometrically, the metric suggests that the cosine of the angle $\theta$ between the elements $\brho$ and $\brho_t$ is given by
\begin{equation}\label{eq:CosTh}
\cos \theta = \frac{ | \langle \brho , \brho_t  \rangle | }{\| \brho \| \| \brho_t \|},
\end{equation}
which can be viewed as the angle between the subspaces spanned by the elements $\brho \brho^H$ and $\brho_t \brho_t^H$ in the lifted domain $\mathbb{C}^{N \times N}$. This can be seen by evaluating the error between the lifted terms as follows:
\begin{equation}\label{eq:LiftDst}
\| \brho \brho^H - \brho_t \brho_t^H \|_F^2 = \| \brho \brho^H \|_F^2 +  \| \brho_t \brho_t^H \|_F^2 - 2 \mathrm{Re} (\langle \brho \brho^H ,  \brho_t \brho_t^H \rangle_F ).
\end{equation}
For two rank-1 arguments, the Frobenius inner product in \eqref{eq:LiftDst} reduces to $| \langle \brho , \brho_t \rangle |^2$, and the cosine of the angle between the elements becomes equal to $\cos^2 (\theta)$. Since \eqref{eq:CosTh} is non-negative, the relationship between the two angles is one-to-one.
Therefore, WF distance metric can be interpreted as a measure of distance between the lifted variables, and the convergence with respect to the metric \eqref{eq:distance} in signal domain is equivalent to convergence with respect to \eqref{eq:LiftDst} in the lifted space.

\label{sec:DMet}

\subsection{Theoretical Advantages of GWF via the Lifted Perspective}

The lifted perspective of GWF reveals its connection to convex low rank matrix recovery methods, such as those proposed in \cite{Demanet14} and \cite{Mason2015}, as the objective in \eqref{eq:perturbRank1} simply corresponds to solving the quadratic data-fit over a more exclusive search space.  
In the scope of lifting based approaches, we discussed Uzawa's method as a solver for the trace regularized problem in \eqref{eq:PerturbedLift}.
Consider solving \eqref{eq:perturbRank1} by Uzawa's method described in \eqref{eq:SVT}-\eqref{eq:SVT2}, which essentially reduces to projected gradient descent: a gradient step over the smooth $\ell_2$ mismatch term is followed by a projection onto the intersection of the PSD cone and the set of rank-1 matrices.
While in general projections to the intersection of two sets is an optimization problem on its own, and the rank-1 constraint constitutes a non-convex manifold, there exists a simple projection onto the set $\mathcal{X} = \{\mathbf{X} \in \mathbb{C}^{N \times N} :  \text{rank} (\mathbf{X} ) = 1 \cap \ \mathbf{X} \succeq 0\}$ such that 
\begin{equation}\label{eq:Uzaw_proj}
\mathcal{P}_{\mathcal{X}} = \mathcal{P}_{r = 1} \circ \mathcal{P}_{+},
\end{equation}
which is the successive operation of the projection onto the PSD cone, followed by a rank-1 approximation. 

Despite that a solver can be formulated, recovery guarantees of Uzawa's method do not cover this case due to the non-convexity of the projected set. 
Simply, a gradient step over the convex PSD cone is not guaranteed to improve the rank-1 approximation of an estimate, which is after all the original motivation behind pursuing the convex solvers. 
Essentially, GWF framework circumvents convexification and presents an alternative update scheme to Uzawa's for minimizing the objective in \eqref{eq:perturbRank1}.
This alternative update form stems from the fact PSD rank-1 constraint of the lifted unknown can be enforced by a variable transformation in the update equation, rather than by projections given in \eqref{eq:Uzaw_proj}.

Clearly, the immediate advantage of the GWF formulation in \eqref{eq:WFform} over the convex relaxations is the dimensionality reduction of the search space, attributed to iterating in the signal domain.
There is yet another significant advantage to the GWF formulation relating to exact recovery guarantees.
By \eqref{eq:WFform}, an iterative scheme for the unrelaxed, non-convex form of the lifted problem is formulated, which enforces the rank-1, PSD structure on the iterates.
This allows the constraint set to be considerably \emph{smaller} than that of the trace relaxation or the convex feasibility problems.
Formally, the problem \eqref{eq:WFform} has a unique solution if the equivalence set of any $\brho \in \mathbb{C}^N$ is its equivalence class under the correlation map $\mathcal{L}$ as defined in Definition \ref{def:EqSet}.
\begin{condition}{\emph{Uniqueness Condition}.}\label{con:Cond1}
There exists a unique solution $\brho_t \brho_t^H \in \mathcal{X}$ for the problem in \eqref{eq:WFform} if
\begin{equation}
\mathit{P}_{\brho} = \mathcal{E}_{\brho}, \quad \forall \brho \in \mathbb{C}^N.
\end{equation}
\end{condition}
In other words, there should exist no element $H$ in the null space of $\mathcal{F}$, such that $\brho_t \brho_t^H + H$ is a rank-1, PSD matrix.
Therefore, for exact interferometric inversion by GWF, the null space condition of the lifted forward model has to hold over a much less restrictive set than any of the approaches discussed in Section \ref{sec:Lifting}.

\subsection{The Spectral Method for GWF Initialization}\label{sec:Spectral}

Having to solve a non-convex problem, exact recovery guarantees of the WF framework depend on the accuracy of the initial estimate $\brho_0$.
The initial estimate of the standard WF algorithm is computed by the spectral method which corresponds to the leading eigenvector of the following positive semi-definite matrix:
\begin{equation}\label{eq:spectral}
\mathbf{Y} = \frac{1}{{M}}\sum_{m = 1}^{{M}} d^{m} \mathbf{L}^m (\mathbf{L}^m)^H,
\end{equation}
where $\mathbf{L}^m = \mathbf{L}_i^m = \mathbf{L}_j^m$, and $d^m = | \langle \mathbf{L}^m, \brho_t \rangle|^2$ for $i = j$.
The leading eigenvector is scaled by the square root of the corresponding largest eigenvalue $\lambda_0$ of $\mathbf{Y}$.
In \cite{candes2015phase}, the spectral method is described from a stochastic perspective.
By the strong law of large numbers, under the assumption that we have $\mathbf{L}^m \sim \mathcal{N}(0, \frac{1}{2} \mathbf{I}) + \mathrm{i} \mathcal{N}(0, \frac{1}{2} \mathbf{I} ), m = 1, \cdots, M$, the spectral matrix $\mathbf{Y}$ becomes equal to
\begin{equation}
\mathbb{E} \big[ \frac{1}{{M}}\sum_{m = 1}^{{M}} d^{m} \mathbf{L}^m (\mathbf{L}^m)^H \big] = \| \brho_t \|^2 \mathbf{I} +  \brho_t \brho_t^H,
\end{equation}
as $M \rightarrow \infty$, which has the true solution $\brho_t$ as its leading eigenvector.
The concentration of the spectral matrix around its expectation is used to show that the leading eigenvector of $\mathbf{Y}$ is \emph{sufficiently} accurate, such that the sequence of iterates $\{ \brho_k \}$ of \eqref{eq:GWF_Updates} converges to an element in the global solution set $P$.

In developing the GWF framework, we view the spectral method as a procedure in the lifted domain.
In fact, we observe that the spectral matrix of phase retrieval in \eqref{eq:spectral} is the backprojection estimate of the lifted unknown $\tilde{\brho}_t ={\brho}_t{\brho}_t^H$. 
Having different measurement vectors $\mathbf{L}_i^m$ and $\mathbf{L}_j^m$ in the cross-correlated measurement case, using the definition of the backprojection operator in \eqref{eq:BackProject},  we extend \eqref{eq:spectral} and redefine $\mathbf{Y}$ as follows:
\begin{equation} \label{eq:CrossCorrSpectral}
\mathbf{Y} = \frac{1}{M} \mathcal{F}^H (\d) =  \frac{1}{M} \sum_{m = 1}^M d_{ij}^m \mathbf{L}_i^m \left( \mathbf{L}_j^m \right)^H.
\end{equation}

As noted, the true solution $\tilde{\brho}_t = \brho_t \brho_t^H$ of the lifted problem lies in the positive semi-definite (PSD) cone.
In standard WF for phase retrieval, the spectral matrix $\mathbf{Y}$ is formed by summation of positive semi-definite outer products that are scaled by $\mathbb{R}^{+}$ valued measurements $\{d^m\}_{m=1}^M$, as auto-correlations are by definition squared magnitudes.
Hence, the WF spectral method generates an estimate of the lifted unknown within the constraint set by default.
This obviously is not the case for the backprojection estimate \eqref{eq:CrossCorrSpectral} with the cross-correlated measurement model.
Therefore, the extension of the spectral method to cross-correlations includes a projection step onto the PSD cone.
Since the PSD cone is convex, its projection operator is non-expansive and yields a closer estimate to $\brho_t \brho_t^H$ than that of \eqref{eq:CrossCorrSpectral}.
The GWF spectral matrix then becomes
\begin{equation} \label{eq:spectralInitGWF}
\hat{\mathbf{X}} := \frac{1}{2M }\sum_{m = 1}^M d_{ij}^m \mathbf{L}_i^m \left( \mathbf{L}_j^m \right)^H + \overline{d_{ij}^m} \mathbf{L}_j^m \left( \mathbf{L}_i^m \right)^H.
\end{equation}
We discard the positive semi-definitivity in \eqref{eq:spectralInitGWF}, and only project onto the set of symmetric matrices, which is also convex.
This is simply because only the leading eigenvector will be kept from the generated lifted estimate, which is unaffected by the projection onto the PSD cone\footnote{Unless the leading eigenvalue is negative, a scenario that is excluded due to the conditions for exact recovery in Section \ref{sec:Results}.}.
Letting $ \lambda_0, \mathbf{v}_0$ denote the leading eigenvalue-eigenvector pair of $\hat{\mathbf{X}}$, the GWF initial estimate $\brho_0$ is given by\footnote{Note that \eqref{eq:spectralInitGWF} yields a symmetric matrix, hence have eigenvalues $\lambda_i \in \mathbb{R}$.}
\begin{equation}\label{eq:GWF_initp}
\brho_0 = \sqrt{\lambda_0} \mathbf{v}_0.
\end{equation}

Using the representation in the lifted problem in \eqref{eq:CrossCorrSpectral} and plugging in \eqref{eq:LiftedLinMod} for the measurements, the GWF spectral matrix $\hat{\mathbf{X}}$ can be written as
\begin{equation}\label{eq:spectral_init_Fin}
\hat{\mathbf{X}} := \frac{1}{M} \mathcal{P}_{S} \left( \mathcal{F}^{H} ( {\mathbf{d}} ) \right) = \frac{1}{M} \mathcal{P}_{S} \left( \mathcal{F}^H \mathcal{F} ( \brho_t \brho_t^H )  \right).
\end{equation}
Hence, the leading eigenvalue-eigenvector extraction corresponds to keeping the rank-1, PSD approximation, $\tilde{\brho}_0 := \brho_0 \brho_0^H$, of the backprojection estimate in the lifted domain. 
Furthermore, by definition in \eqref{eq:spectral_init_Fin}, the accuracy of the spectral estimate fully hinges on the properties of the normal operator of the lifted problem, i.e., $\mathcal{F}^H \mathcal{F}$ over the set of rank-1, PSD matrices, $\mathcal{X}$.




As a comparison, consider the LRMR approach for the interferometric inversion by the PSD constrained singular value thresholding (SVT) algorithm in \cite{Mason2015}.
It can be seen that, setting $\mu_k = 1/M$ in \eqref{eq:SVT2}, the first iterate generated by \eqref{eq:SVT} is $\hat{\mathbf{X}} = \mathcal{P}_{\tau} ( \frac{1}{M} \mathcal{F}^H  \mathbf{d})$, which, prior to the singular value thresholding, is identical to the backprojection estimate computed by the spectral method. 
Hence, the spectral method simply differs from the first Uzawa iteration by keeping the rank-1 approximation via the projection operator $\mathcal{P}_{\mathcal{X}}$ defined in \eqref{eq:Uzaw_proj}, instead of a low-rank approximation.
This precisely corresponds to the first Uzawa iteration to solve the non-convex rank-1 constrained problem in \eqref{eq:perturbRank1}. 


\label{sec:Spect}

\subsection{Algorithm Summary and Computational Complexity}

%

Compared to the lifting based LRMR approaches, GWF provides significant reductions in computational complexity, and memory requirements per iteration. 
As shown in Section \ref{sec:Spectral}, GWF uses the first iteration of the Uzawa's method to compute an initial estimate $\brho_0$, and replaces the following iterations over the lifted domain with iterations on the leading eigenspace. 
The GWF algorithm is summarized as follows:
\begin{itemize}
\item \textbf{Input:} Interferometric measurements $d_{ij}^m$ and measurement vectors $\mathbf{L}_i^m, \mathbf{L}_j^m$, $\forall i \neq j$, $m = 1, \cdots, M$. 

\item \textbf{Initialization:} Run Uzawa's method in \eqref{eq:SVT} initialized with $\y = 0$, $\mu_0 = 1/M$, and $\lambda = 0$, i.e., trace regularization free, for $1$ iteration, yielding 
$$
\hat{\mathbf{X}} = \frac{1}{M} \mathcal{P}_{S} \left(\mathcal{F}^H ( \d ) \right).
$$
Keep the rank-1 approximation $\lambda_0 \brho_0 \brho_0^H$. The initialization step consists of the outer product of the two measurement vectors for each of the $M$ samples, resulting in $\mathcal{O}(MN^2)$ multiplications, followed by an eigenvalue decomposition with $\mathcal{O}(N^3)$ complexity.

\item \textbf{Iterations:} Perform gradient descent updates as $\brho_{k+1} = \brho_{k} - \frac{\mu_{k + 1}}{\| \brho_0 \|^2} \nabla \mathcal{J}(\brho_{k})$, with
$$
\nabla \mathcal{J}( \brho_k ) = \frac{1}{M} \mathcal{P}_S \left( \mathcal{F}^H(\mathbf{e}_k) \right) \brho_{k},
$$
where $(\mathbf{e}_k)^m = (\left(\mathbf{L}_i^m \right)^H \brho_{k} \brho_{k}^H \mathbf{L}_j^m - d_{ij}^m)$. 

Each iteration requires the following operations:
\begin{enumerate}
\item Computing and storing the linear terms $(\mathbf{L}_{i,j}^m)^H \brho_k$, requiring $M$ number of $N$ multiplications for each, resulting in $\mathcal{O}(MN)$ multiplications. 
\item Computing the error by cross correlating linear terms, requiring $\mathcal{O}(M)$ multiplications.
\item Multiplication of the linear terms $(\mathbf{L}_{i,j}^m)^H \brho_k$ and the error $e_{ij}^m$ for each $m = 1, \cdots M$, requiring $\mathcal{O}(M)$ multiplications.
\item Multiplication of the result in $3$ with vectors $\{\mathbf{L}_i^m\}_{m=1}^M$ and $\{\mathbf{L}_j^m\}_{m=1}^M$, requiring $\mathcal{O}(MN)$ multiplications.
\end{enumerate}
These operations result in $\mathcal{O}(MN)$ multiplications for each iteration. 
\end{itemize}
%
%


\section{Theory of Exact Interferometric Inversion via GWF} \label{sec:Results}

In this section we present our exact recovery guarantees for interferometric inversion by GWF. 
Notably, we merge the exact recovery conditions of standard Wirtinger Flow that rely on statistical properties of the sampling vectors into a single condition on the lifted forward model in the context of interferometric inversion ($i \neq j$). 
In Theorem \ref{thm:Theorem1}, we state that if $\mathcal{F}$ satisfies the restricted isometry property on the set of rank-1, positive semidefinite matrices with a sufficiently small restricted isometry constant, GWF is guaranteed to recover the true solution upto a global phase factor, for any $\brho_t \in \mathbb{C}^N$.  
Following Theorem \ref{thm:Theorem1}, we establish the validity of our condition in Theorem \ref{thm:Theorem2} for the case of $\mathcal{O}(N \log N)$ measurements that are cross-correlations of i.i.d. complex Gaussian sampling vectors.

We begin by introducing the definitions of some concepts that appear in our theorem statements.

\subsection{The $\epsilon$-Neighborhood of $\mathit{P}$ and the Regularity Condition}

\begin{definition}{\emph{$\epsilon$-Neighborhood of $\mathit{P}$}.}
We denote the $\epsilon$-neighborhood of the global solution set $\mathit{P}$ in \eqref{eq:SolutionSet} by $\mathit{E}(\epsilon)$ and define it as follows \cite{candes2015phase}:
\begin{equation}
\mathit{E}(\epsilon) = \{ \brho \in \mathbb{C}^N : \mathrm{dist}( \brho, \mathit{P} ) \leq \epsilon \}.
\end{equation}
\end{definition}

The set $E(\epsilon)$ is determined by the distance of the spectral initialization to the global solution set, i.e., $\epsilon = \text{dist}(\brho_0, \brho_t)$. 
The main result of the standard WF framework is that, for Gaussian and coded diffraction pattern measurement models \cite{candes2015phase}, $\epsilon$ is sufficiently small so that the objective function $\mathcal{J}$ satisfies the following regularity condition:
\begin{condition}{\emph{Regularity Condition}.}\label{con:RegCon}
The objective function $\mathcal{J}$ satisfies the regularity condition if, for all $\brho \in \mathit{E}(\epsilon)$, the following holds
\begin{equation}\label{eq:RegCon}
\text{\emph{Re}} \left( \langle \nabla \mathcal{J}(\brho),  (\brho - {\brho}_t e^{\mathrm{i} \Phi(\brho)}) \rangle \right) \geq \frac{1}{\alpha} \text{\emph{dist}}^2 (\brho, {\brho}_t) + \frac{1}{\beta} \|  \nabla \mathcal{J}(\brho) \|^2
\end{equation}
for fixed $\alpha > 0$ and $\beta > 0$ such that $\alpha \beta > 4$. 
\end{condition}

The regularity condition guarantees that the iterations in \eqref{eq:UpdateTerm} are contractions with respect to the distance metric \eqref{eq:distance}, which ensures that all gradient descent iterates remain in $E(\epsilon)$, which is established by the following lemma from \cite{candes2015phase}:
\begin{lemma}\label{lem:Lemma1}
Assume that $\mathcal{J}$ obeys the regularity condition in \eqref{eq:RegCon} for some fixed $\alpha, \beta$ for all $\brho \in E(\epsilon)$. Having $\brho_0 \in E(\epsilon)$, and assuming $\mu \leq {2}/{\beta}$, consider the following update
\begin{equation}\label{eq:altUpdtWF}
\brho_{k+1} = \brho_k - \mu \nabla \mathcal{J}(\brho_k).
\end{equation}
Then, for all $k$ we have $\brho_k \in E(\epsilon)$ and
$$
\mathrm{dist}^2 ( \brho_k, \brho_t ) \leq  ( 1 - \frac{2 \mu}{\alpha} )^k \mathrm{dist}^2 ( \brho_0, \brho_t ).
$$
\end{lemma}
\begin{proof}
See \cite{candes2015phase}, Lemma 7.10. 
\end{proof}

Furthermore, from the definition of $\nabla \mathcal{J}$ in \eqref{eq:UpdateTerm}, the regularity condition implies that there exists no $\brho \in E(\epsilon)$ that belongs to the equivalence class of $\brho_t$ under the mapping $\mathcal{L}$.
Hence, the uniqueness condition for exact recovery is satisfied locally, and \ref{con:RegCon} is a sufficient condition by Lemma 7.1 of \cite{candes2015phase} such that the algorithm iterates $\{ \brho_k \}$ converge to $\mathit{P}$ at a geometric rate. 
The spectral initialization is said to be {sufficiently} accurate, if \eqref{eq:RegCon} holds $\forall \brho \in \mathit{E}(\epsilon)$.

\subsection{Sufficient Conditions For Exact Recovery}

The condition we assert on the lifted forward model $\mathcal{F}$ is the restricted isometry property on the set of rank-1, PSD matrices, $\mathcal{X}$. 

\begin{definition}{\emph{Restricted Isometry Property} (RIP).}\label{def:Defn3_1}
Let $\mathcal{A}: \mathbb{C}^{K \times N} \rightarrow \mathbb{C}^M$ denote a linear operator. Without loss of generality assume $K \leq N$. For every $ 1 \leq r \leq K$, the $r-\text{restricted}$ isometry constant (RIC) is defined as the smallest $\delta_r < 1$ such that
\begin{equation}\label{eq:RIP1}
(1 - \delta_r ) \| \mathbf{X} \|_F^2 \leq \| \mathcal{A}(\mathbf{X}) \|^2 \leq (1 + \delta_r) \| \mathbf{X} \|_F^2
\end{equation}
holds for all matrices $\mathbf{X}$ of rank at most $r$, where $\| \mathbf{X} \|_F = \sqrt{\text{Tr}(\mathbf{X}^H \mathbf{X})}$ denotes the Frobenius norm. 
\end{definition}

Suppose \eqref{eq:RIP1} holds for all  $\mathbf{X} \in \mathcal{D} \subset \mathbb{C}^{K \times N}$ that have rank-r with some constant $0 < \delta_{r} < 1$, then $\mathcal{A}$ is said to satisfy the RIP$_\mathcal{D}$ with RIC-$\delta_{r}$. 
For the interferometric inversion problem, having $K = N$, if there exists $\delta_1 < 1$ such that $\mathcal{F}$ in \eqref{eq:LiftedLinMod} satisfies the restricted isometry property on the PSD cone, we say that the lifted forward model satisfies RIP$_+$ with RIC-$\delta_{1}$, i.e., RIP over the set of rank-1, PSD matrices, $\mathcal{X}$.  
%


\eqref{eq:RIP1} quantifies how close $\mathcal{F}^H \mathcal{F}$ is to an identity over rank-1, PSD matrices through the following lemma:

\begin{lemma} \label{cor:Corr1}
Suppose $\mathcal{F}$ satisfies RIP on the set rank-1, PSD matrices of size $N \times N$, i.e., $\mathcal{X} = \{ \brho \brho^H : \brho \in \mathbb{C}^N \}$ with RIC-$\delta_1$. Then, for any $\mathbf{X} \in \mathcal{X}$ we have
\begin{equation}\label{eq:perturbmod}
(\mathcal{F}^H \mathcal{F} - \mathbf{I})(\mathbf{X}) = \bdelta (\mathbf{X})
\end{equation}
where $\bdelta : \mathcal{X} \rightarrow \mathbb{C}^M$ is a bounded operator such that
$$
\| \bdelta \|_{\mathcal{X}} :=  \underset{\mathbf{\brho \in \mathbb{C}^N \setminus \{0 \}}}{\text{max}} \frac{\| {\bdelta} (\brho \brho^H) \|}{\| \brho \brho^H \|} = \delta_1,
$$
where $\| \cdot \|$ denotes the spectral norm. 
\end{lemma}
\begin{proof}
See Appendix \ref{sec:App1}. 
\end{proof}


In the GWF framework, we identify RIP over rank-1, PSD matrices with RIC-$\delta_1 \leq 0.214$ as a sufficient condition for an arbitrary lifted forward model $\mathcal{F}$, which guarantees that the spectral initialization of GWF provides an initial estimate that is sufficiently accurate, i.e., GWF iterates are guaranteed to converge to a global solution in $\mathit{P}$ starting from $\brho_0$. 


\begin{theorem}{\emph{Exact Recovery by GWF}.}\label{thm:Theorem1}
Assume the lifted forward model $\mathcal{F}$\footnote{Upto a normalization factor, such as $\frac{1}{\sqrt{M}}$.} satisfies the RIP condition over rank-1, PSD matrices with the restricted isometry constant (RIC)-$\delta_1$. Then, the initial estimate $\brho_0$ obtained from the spectral method by \eqref{eq:spectralInitGWF} and \eqref{eq:GWF_initp} satisfies
\begin{equation}
\mathrm{dist}^2 (\brho_0, \brho_t ) \leq \epsilon^2 \| \brho_t \|^2,
\end{equation}
where $\epsilon$ is an $\mathcal{O}(1)$ constant with $\epsilon^2 = (2 + \delta_1)(1 - \sqrt{1 - \frac{\delta_1}{1- \delta_1}}) + \frac{\delta_1^2}{8}$, and the regularity condition in \eqref{eq:RegCon} surely holds for the objective function in \eqref{eq:GWF_opt} if $\delta_1 \leq 0.214$, with any $\alpha$, $\beta$ satisfying 
\begin{equation}\label{eq:suffcond_c}
\frac{1}{\alpha \| \brho_t \|^2} + \frac{c^2(\delta_1) \| \brho_t \|^2}{\beta} \leq h(\delta_1) := (1 - \delta_2)(1-\epsilon)(2-\epsilon),
\end{equation}
where $\delta_2 = \frac{\sqrt{2}(2 + \epsilon) \delta_1}{\sqrt{(1-\epsilon)(2-\epsilon})}$ and $c(\delta_1) = (2+\epsilon)(1+\epsilon)(1 + \delta_1)$. Thus, for the iterations of \eqref{eq:altUpdtWF} with the update term in \eqref{eq:UpdateTerm} and the fixed step size $\mu \leq {2}/{\beta}$, we have
\begin{equation}
\mathrm{dist}^2 ( \brho_k, \brho_t ) \leq \epsilon^2 ( 1 - \frac{2 \mu}{\alpha} )^k \| \brho_t \|^2. 
\end{equation}
\end{theorem}
\begin{proof}
See Section \ref{sec:Proof_Thm1}. 
\end{proof}

We refer to $\delta_2$ as the restricted isometry constant of the \emph{local} RIP-2 condition by Lemma \ref{lem:Lemma6} which is exclusively over elements of the form $\{ \brho \brho^H - \brho_t \brho_t^H : \brho \in E(\epsilon) \}$ , and $c$ as the local Lipschitz constant of $\nabla \mathcal{J}$ by Lemma \ref{lem:Lemma7}, both stated in Section \ref{sec:Proofs}.

\begin{remark}
We summarize the implications of Theorem \ref{thm:Theorem1} by the following remarks.
\begin{enumerate}
\item Theorem \ref{thm:Theorem1} establishes a regime in which the regularity condition holds by default as a result of the RIP over rank-1, PSD matrices, captured by \eqref{eq:suffcond_c}. 
This regime also defines the range of values the RIC-$\delta_1$ can attain, through the quantity $\delta_2$ which must satisfy $\delta_2 < 1$, as shown in Figure \ref{fig:Figure_1}. 
Notably, having $\delta_2 < 1$ guarantees the uniqueness of a solution locally in $E(\epsilon)$, and the restricted strong convexity of $\mathcal{J}$ since there exists a fixed $\alpha, \beta$ satisfying \eqref{eq:suffcond_c}. 
Numerically, plugging in the $\epsilon$ constant defined by $\delta_1$, this is satisfied for $\delta_1 \leq 0.214$.


\item Figure \ref{fig:Figure_1} demonstrates the values the constants $c$ and $h$ attain in the valid range for $\delta_1$. 
These $\mathcal{O}(1)$ constants directly impact the convergence rate of the algorithm, as $\alpha$ and $\beta$ are required to be sufficiently large values for \eqref{eq:suffcond_c} to hold. 
Observe that \eqref{eq:suffcond_c} implies setting $\alpha = \alpha'/{\| \brho_t \|^2}$, and $\beta = \beta' \| \brho_t \|^2$, where $\alpha', \beta' = \mathcal{O}(1)$, hence $\alpha \beta = \mathcal{O}(1)$. 
Since clearly $h < 2$, and $c^2 > 4$, and the regularity condition can at best hold when $\alpha' \beta' > 4$ by definition. 
Therefore, the RIP over rank-1, PSD matrices with RIC-$\delta_1 \leq 0.214$ is a sufficient condition for the exact recovery via GWF. 

\item Since $\| \brho_t \|$ is unknown a priori, a suitable approximate scaling in setting constants $\alpha, \beta$ is  $\| \brho_0 \|^2$. 
This scaling factor is merely the leading eigenvalue $\lambda_0$ of the spectral matrix in \eqref{eq:spectralInitGWF}. 
From Lemma \ref{lem:Lemma3}, for $\mathcal{F}$ satisfying RIP over rank-1, PSD matrices with the RIC-$\delta_1$, $\lambda_0$ is lower bounded by $1-\delta_1$, hence the condition in \eqref{eq:suffcond_c} can be enforced by setting $\alpha', \beta' = \mathcal{O}(1)$ such that
$$
\frac{1+\delta_1}{\alpha'} + \frac{c^2(\delta_1)}{(1-\delta_1) \beta'} \leq h(\delta_1).
$$
Similarly, picking $\mu \leq \frac{2}{\beta}$ to yield convergence rate of $\frac{2\mu}{\alpha} \leq \frac{4}{\alpha' \beta'}$, the iterations in \eqref{eq:altUpdtWF} must have a step size $\mu$ that is $\mathcal{O}({1}/{\| \brho_t \|^2})$. 
This is essentially where the normalization term in \eqref{eq:GWF_Updates} originates from, with $\| \brho_0 \|^2$ serving as an approximation to $\| \brho_t \|^2$. 
As a result, Lemma \ref{lem:Lemma1} for \eqref{eq:GWF_Updates} simply holds for $\mu_k \leq (1-\delta_1)\frac{2}{\beta'}$ due the effect of the mismatch in the scaling factors, in agreement to what is noted in \cite{candes2015phase}.  




\item Overall, Theorem \ref{thm:Theorem1} establishes that the convergence speed of GWF algorithm is controlled by the restricted isometry constant $\delta_1$.  
As $\delta_1$ approaches to the critical limit of $0.214$, $c$ and $\delta_2$ values increase super-linearly as shown in Figure \ref{fig:Figure_1}. 
This has strong implications on the convergence speed, as $\beta'$ is inversely proportional to the choice of step sizes $\mu_k$, and quadratically related to the magnitude of $c$. 

\item A consequence of our result is a universal upper bound on $\epsilon$ under our sufficient condition of RIP over rank-1, PSD matrices. As depicted in Figure \ref{fig:Figure_1}, Theorem \ref{thm:Theorem1} determines what sufficiently close means numerically. 

\end{enumerate}
\end{remark}


\begin{remark}
Despite solving the identical perturbed problem, GWF iterations provably converge, whereas the convergence guarantees of Uzawa's method vanish due to inclusion of a non-convex constraint. 
Similarly, the special structure of the constraint set and the GWF iterates suffice the RIP condition to be satisfied only over rank-1 matrices in the PSD cone, whereas the uniqueness condition of Uzawa's method requires RIP over the set of rank-2 matrices in the non-convex rank minimization problem. 
\end{remark}




\begin{figure}
\centering
\includegraphics[scale=0.325]{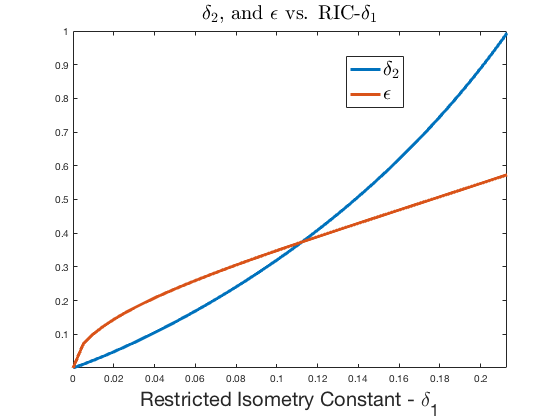}
\includegraphics[scale=0.325]{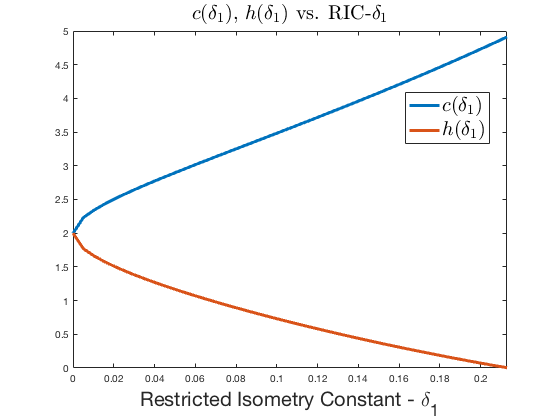}
\caption{\emph{Numerical evaluation of $\epsilon$, RIC-$\delta_2$ and $c$, $h$ values with respect to the RIC-$\delta_1 \leq 0.214$ of the restricted isometry property on the set of rank-1 PSD matrices.}}
\label{fig:Figure_1}
\end{figure}

%

\subsection{Restricted Isometry Property for Cross-Correlation of Gaussian Measurements}

In standard WF for phase retrieval with the i.i.d. complex Gaussian model, i.e., $\mathbf{L}_m \sim \mathcal{N}(0, \mathbf{I}/2) + \mathrm{i} \mathcal{N}(0, \mathbf{I}/2)$, the accuracy of the spectral estimate \eqref{eq:spectral} is established as
\begin{equation}\label{eq:WF_init}
\| \mathbf{Y}  - ( \brho_t \brho_t^H + \| \brho_t \|^2 \mathbf{I} ) \| \leq \delta \| \brho_t \|^2
\end{equation}
with probability $1-10e^{\gamma N} - 8/N^2$ where $\gamma$ is a fixed positive numerical constant. 
This result is derived from the concentration bound of the Hessian of the objective function around its expectation at a global minimizer $\brho_t$ such that
\begin{equation}
\| \nabla^2 \mathcal{J} ( \brho_t ) - \mathbb{E} [ \nabla^2 \mathcal{J} ( \brho_t ) ] \| \leq \delta \| \brho_t \|^2
\end{equation}
where $\delta$ is the concentration bound. 

In the problem of phase retrieval, plugging in the definition of phaseless measurements into \eqref{eq:spectral}, the auto-correlations yield the $4^{th}$ moments of the elements of the Gaussian measurement vectors.
This introduces a bias of $\| \brho_t \|^2 \mathbf{I}$ in the spectral estimate $\mathbf{Y}$ as can be seen in \eqref{eq:WF_init}. 
Moving from the auto-correlations to cross-correlations removes this bias component from the spectral matrix of GWF. 
Hence, for the Gaussian model, cross-correlated measurement map, i.e., $\mathcal{F}$ when $i \neq j$, satisfies the RIP over the set of rank-1, PSD matrices. 

Without loss of generality and following \cite{candes2015phase}, we present our result for the case $\| \brho_t \| = 1$.
\begin{theorem}{\emph{RIP over Rank-1, PSD Matrices for Cross-Correlated Gaussian Measurements.}}\label{thm:Theorem2}
Let the measurement vectors $\mathbf{L}_i^m, \mathbf{L}_j^m$ in \eqref{eq:interferom} follow the i.i.d. complex Gaussian model, i.e. $\mathbf{L}_i^m, \mathbf{L}_j^m \sim \mathcal{N}(0, \mathbf{I}/2) + \mathrm{i} \mathcal{N}(0, \mathbf{I}/2)$. Then, the lifted forward model $\frac{1}{\sqrt{M}}\mathcal{F}$ in \eqref{eq:LiftedLinMod} for cross-correlated measurements (when $i \neq j$), satisfies RIP defined in \eqref{eq:RIP1} for $r = 1$ over the PSD cone, with probability $1-8e^{-\gamma N}-5/N^{2}$ given $\mathcal{O}(N \log N)$ measurements, where $\gamma$ is a fixed positive numerical constant. Moreover, the spectral matrix $\hat{\mathbf{X}}$ defined in \eqref{eq:spectralInitGWF} satisfies
\begin{equation}
\| \hat{\mathbf{X}}  -  \brho_t \brho_t^H  \| \leq \delta_1,
\end{equation}
where $\brho_t$ is the ground truth signal with $\| \brho_t \| = 1$ and $\delta_1$ is the restricted isometry constant of $\frac{1}{\sqrt{M}}\mathcal{F}$. 
\end{theorem}
\begin{proof}
See Section \ref{sec:Proof_Thm2}. 
\end{proof}



\begin{remark}
Outcomes of Theorem \ref{thm:Theorem2} are explained with the following remarks.
\begin{enumerate}
\item Theorem \ref{thm:Theorem2} establishes the relationship between the concentration bound of the spectral matrix, and the RIP-1 condition for interferometric inversion. This indicates that the regularity condition of WF framework is \emph{redundant} for our problem if $\delta_1$ is picked properly, since by Theorem \ref{thm:Theorem1}, RIP$_+$ with RIC-$\delta_1 \leq 0.214$ directly implies the regularity condition. 

\item Note that the equivalent linear model in the lifted domain actually has $N^2$ unknowns. By our formulation of GWF in the lifted problem, having measurements of the order of $N \log N$, therefore, corresponds to an under-determined system of equations in which exact recovery guarantees of GWF hold. 

\item Note that our measurement complexity is identical to that of standard WF. This is an expected result, as the cross-correlations only impact the removal of the diagonal bias in the spectral matrix $\mathbf{Y}$, not the concentration of $\mathbf{Y}$ around its expectation. 

\item It should be re-iterated that when $i = j$, the backprojection estimate $\mathbf{Y}$ contains a diagonal bias of the form $\| \brho_t \|^2 \mathbf{I}$, which breaks our RIP condition over rank-1, PSD matrices. 
In such a case, the concentration bound then holds around the expectation term that includes this diagonal bias. 
Establishing a sufficient condition for the phaseless case under our framework is not in the scope of this paper, which we omit for future work. 
\end{enumerate}
\end{remark}

\begin{remark}
It should be noted that phaseless measurements in the Gaussian model is known to satisfy a restricted isometry property under the $\ell_1$ norm in the range of the lifted forward map. This is referred to as RIP-1 condition in literature, and is prominently featured in exact recovery theory of standard WF, and more recently in \cite{soltanolkotabi2019structured} for non-convex optimization via gradient descent, even in the presence of non-convex regularizers. 

A critical difference of our work is that such a RIP type condition isn't used as a tool to obtain local curvature and local smoothness conditions to assert the regularity condition with a high probability. 
Instead, under the RIP studied in this paper, a regime is derived in which the regularity condition is guaranteed to hold, deterministically. 
Notably, for the interferometric inversion problem with $i \neq j$, $\delta_1$ values tested in \cite{candes2015phase} are confidently within the range in which our GWF theory applies. 
\end{remark}

\section{Proofs of Theorems \ref{thm:Theorem1} and \ref{thm:Theorem2}}\label{sec:Proofs}

In this section, we present the proofs of Theorems \ref{thm:Theorem1} and \ref{thm:Theorem2}. 
We first present key lemmas, and next prove the theorems using these lemmas. 
We provide the detailed proofs of the lemmas in Appendix \ref{sec:AppThm1} and \ref{sec:AppThm2}.  

To prove Theorem \ref{thm:Theorem1}, we begin by showing that the RIC-$\delta_1$ of our RIP condition determines the distance $\epsilon$ of the spectral initialization in a one-to-one manner.
We then establish that the regularity condition is directly implied by RIP$_+$ with RIC-$\delta_1 \leq 0.214$. 
In achieving this result we first show that the structure of the rank-1, PSD set allows for restricted isometry property to hold locally for the difference of two rank-1 PSD matrices, i.e., a local RIP-2 condition similar to the one in \cite{li2018rapid}, with a RIC-$\delta_2$. 
The upper bound on $\delta_1$ ensures that RIC-$\delta_2$ of the local RIP-2 satisfies $\delta_2 < 1$. 
The local RIP-2 condition, in turn, ensures that restricted strong convexity holds in the $\epsilon$-neighborhood of the global solution set, which leads to exact recovery conditions of GWF. 

For Theorem \ref{thm:Theorem2}, we first show that the bias term in \eqref{eq:WF_init} resulting from the $4^{th}$ moments of the random Gaussian entries disappear when we have cross-correlations instead of auto-correlations of measurements. 
We then establish that the spectral matrix is concentrated around its expectation, using the machinery in \cite{candes2015phase}, adapted for cross-correlations. 
Finally, we use the definition of the spectral matrix to derive the RIP over rank-1, PSD matrices from the concentration bound, which yields the RIC-$\delta_1$. 

\subsection{Proof of Theorem \ref{thm:Theorem1}}\label{sec:Proof_Thm1}

Without the loss of generality, we assume $\brho_t$ is a solution with $\| \brho_t \| = 1$. 
In establishing the exact recovery guarantees for GWF, we take a two-step approach. 
For a lifted forward map $\mathcal{F}$ satisfying RIP$_+$ with RIC-$\delta_1$ , we first show that the initialization by spectral method yields an estimate that is in the set $E(\epsilon)$. 
We then establish the regularity condition \eqref{eq:RegCon} for the objective function \eqref{eq:GWF_opt} in the $\epsilon$-neighborhood defined by the initialization. 
These two results culminate into convergence to a global solution at a geometric rate as stated in Theorem \ref{thm:Theorem1}. 


\subsubsection{$\epsilon$-Neighborhood of Spectral Initialization}

Rather than the law of large numbers approach in \cite{candes2015phase}, we take the geometric point of view of \cite{wang2018solving} in establishing the $\epsilon$-neighborhood of the spectral initialization.  
We begin by evaluating the distance of the leading eigenvector $\mathbf{v}_0 \in \mathbb{C}^N$ of the spectral matrix in \eqref{eq:spectralInitGWF} to the global solution set \eqref{eq:SolutionSet}. 
Recall the definition of the distance metric 
\begin{equation}
\text{dist} (\v_0, \brho_t) = \| \v_0 - e^{\mathrm{i} \Phi(\v_0)} \brho_t \|
\end{equation}
which is essentially the Euclidean distance of $\mathbf{v}_0$ to the closest point in the solution set $\mathit{P}$ in \eqref{eq:SolutionSet}. 
Without loss of generality, we fix $\Phi(\v_0) = \Phi_0$ and incorporate it into $\brho_t$ such that $\hat{\brho}_t = e^{\mathrm{i} \Phi_0} \brho_t$ represents the closest solution to $\v_0$ in $\mathit{P}$.    
We breakdown the key arguments of our proof into the following three lemmas. 
 
\begin{lemma}\label{lem:Lemma2}
Let $\brho_t$ be a solution with $\| \brho_t \| = 1$ and $\hat{\brho}_t$ is the closest solution in $\mathit{P}$ to $\v_0$. Then, $\mathrm{Re} \langle \hat{\brho}_t, \v_0 \rangle = \langle \hat{\brho}_t, \v_0 \rangle$, and
\begin{equation}
\hat{\brho}_t = \cos(\theta) \v_0 + \sin(\theta) \v_0^{\perp}
\end{equation}
where $\| \v_0 \| = 1$, $\cos ( \theta )  = \mathrm{Re}  \langle \hat{\brho}_t, \v_0 \rangle$, and $\v_0^{\perp}$ is a unit vector lying in a plane whose normal is $\v_0$. 
Similarly, there exists a perpendicular unit vector, $\hat{\brho}_t^{\perp}$, to $\hat{\brho}_t$ such that
\begin{equation}
\hat{\brho}_t^{\perp} = - \sin(\theta) \v_0 + \cos(\theta) \v_0^{\perp}.
\end{equation} 
\end{lemma}

\begin{proof}
See Appendix \ref{sec:App2}.
\end{proof}

%

\begin{lemma}\label{lem:Lemma3}
Consider the spectral matrix $\hat{\mathbf{X}}$ given by \eqref{eq:spectral_init_Fin}, and denote the spectral matrix projected onto the positive semi-definite cone as $\hat{\mathbf{X}}_{PSD}$.  
Then, for a lifted mapping $\mathcal{F}$ satisfying RIP$_+$ with RIC-$\delta_1 = \delta$, $\hat{\mathbf{X}}$ and $\hat{\mathbf{X}}_{PSD}$ have the identical leading eigenvalue-eigenvector pair $\lambda_0$, $\v_0$ such that 
$$
1 - \delta \leq \lambda_0 \leq 1 + \delta.
$$
Furthermore, $\hat{\mathbf{X}}$ and $\hat{\mathbf{X}}_{PSD}$ generate identical spectral initialization, $\brho_0$. 
\end{lemma}
\begin{proof}
See Appendix \ref{sec:App3}.
\end{proof}

Lemma \ref{lem:Lemma3} allows us to analyze the distance of the initial estimate $\brho_0$ to the solution set by the convenience of either the positive semi-definite $\hat{\mathbf{X}}_{PSD}$ or the symmetric spectral estimate $\hat{\mathbf{X}}$, since they generate the same initial estimate $\brho_0$.  

Next, using the Lemmas \ref{lem:Lemma2} and \ref{lem:Lemma3} we reach the following key result.
\begin{lemma}\label{lem:Lemma4}
In the setup of Lemmas \ref{lem:Lemma2} and \ref{lem:Lemma3}, for the angle $\theta$ between the one-dimensional sub-spaces spanned by $\hat{\brho}_t$ and $\v_0$ we have
\begin{equation}
\sin^2(\theta) \leq \frac{\delta}{1-\delta}
\end{equation}
where $\delta$ is the RIC-$\delta_1$ of the lifted map $\mathcal{F}$ satisfying RIP$_+$ . 
\end{lemma}
\begin{proof}
See Appendix \ref{sec:App4}.
\end{proof}

From Lemma \ref{lem:Lemma4}, we can now lower bound the inner product of $\hat{\brho}_t$ and $\v_0$ such that
$$
(\text{Re} \langle \hat{\brho}_t, \v_0 \rangle )^2 = \cos^2 (\theta) = 1 - \sin^2(\theta) \geq 1 - \kappa,
$$
where $\kappa = \frac{\delta}{1-\delta}$. 
Writing the distance of the spectral initialization $\brho_0 = \sqrt{\lambda_0} \v_0$ to the solution set, we have
$$
\text{dist}^2 (\brho_0, \brho_t) =  \lambda_0 + 1 - 2 \text{Re} \langle e^{\mathrm{i} \Phi(\brho_0)} {\brho}_t, \sqrt{\lambda_0} \v_0 \rangle. 
$$
It is easy to see that $\text{Re} \langle e^{\mathrm{i} \Phi(\brho_0)} {\brho}_t, \sqrt{\lambda_0} \v_0 \rangle$ is maximized when $\Phi(\brho_0) = \Phi(\v_0)$, hence we get 
\begin{equation}\label{eq:distInit}
\text{dist}^2 (\brho_0, \brho_t) =  \lambda_0 + 1 - 2 \sqrt{\lambda_0} \text{Re} \langle  \hat{\brho}_t , \v_0 \rangle  \leq \lambda_0 + 1 - 2 \sqrt{\lambda_0} \sqrt{1 - \kappa}. 
\end{equation}
From Lemma \ref{lem:Lemma3}, we know that $1-\delta \leq \lambda_0 \leq 1 + \delta$. 
Moreover, the upper bound on the right hand side of \eqref{eq:distInit} is simply a quadratic term with respect to $\sqrt{\lambda_0}$ since $\sqrt{\lambda_0} (\sqrt{\lambda_0} - 2 \sqrt{1 - \kappa} ) + 1$, which is maximized at the boundary of the domain of values $\sqrt{\lambda_0}$ takes. 
Since the quadratic equation is minimized at $\sqrt{1 - \kappa}$ and we have ${1 - \kappa} \leq {1 - \delta}$ for the domain of possible values of $0 \leq \delta < 1$, $\lambda_0 = 1 + \delta$ is an upper bound for the right-hand-side of \eqref{eq:distInit}. 
Hence, we obtain
$$
\text{dist}^2 (\brho_0, \brho_t) \leq 2 + \delta  - 2 \sqrt{1+\delta} \sqrt{1 - \kappa}.
$$
Writing the Taylor series expansion of $\sqrt{1+\delta}$ around $0$, and discarding the components of order $\mathcal{O}(\delta^3)$ and higher, we have the final upper bound 
\begin{equation}
\text{dist}^2 (\brho_0, \brho_t) \leq (2 + \delta) (1 - \sqrt{1 - \kappa}) + \frac{\delta^2}{8},
\end{equation}
which sets the $\epsilon$-neighborhood as 
\begin{equation}
\epsilon^2 =  (2 + \delta) (1 - \sqrt{1 - \kappa}) + \frac{\delta^2}{8}.
\end{equation} 

\subsubsection{Proof of the Regularity Condition}

Recall that we seek a solution to the interferometric inversion problem by minimizing the following loss function
\begin{equation}\label{eq:objfun_repr}
\mathcal{J}(\brho) = \frac{1}{2M} \sum_{m = 1}^M | \left(\mathbf{L}_i^m\right)^H \brho \brho^H \mathbf{L}_j^m - d_{ij}^m |^2
\end{equation}
and address the optimization by forming the steepest descent iterates
\begin{equation}
\brho^{k+1} = \brho^{k} - \mu \nabla \mathcal{J}(\brho^k)
\end{equation}
where $\mu$ is the learning rate and $\nabla \mathcal{J}$ is the complex gradient defined by the {Wirtinger} derivatives. 

As shown in Section \ref{sec:Upd}, the gradient evaluated at a point $\brho$ can be expressed as
\begin{equation}\label{eq:GradRep}
\nabla \mathcal{J}(\brho)  =  \mathcal{Y}(\brho) \brho,
\end{equation}
where
\begin{equation}\label{eq:repsS6}
\mathcal{Y}(\brho) = \mathcal{P}_{S} \left( \mathcal{F}^H \mathcal{F} (\tilde{\brho} - \tilde{\brho_t}) \right)
\end{equation}
with $\tilde{\brho}$ and $\tilde{\brho}_t$ denoting the lifted variables $\tilde{\brho} = \brho \brho^H$ and $\tilde{\brho}_t = \brho_t \brho_t^H$, respectively. 
Invoking Lemma \ref{cor:Corr1} and 
the linearity of $\mathcal{F}^H \mathcal{F}$ and $\bdelta$ of \eqref{eq:perturbmod}, \eqref{eq:repsS6} can be represented as
$$
\mathcal{Y}(\brho) =  \mathcal{P}_{S} \left( \tilde{\brho} - \tilde{\brho_t} + \bdelta (\tilde{\mathbf{e}})\right)
$$
where $\tilde{\mathbf{e}} = \tilde{\brho} - \tilde{\brho_t}$ is the error in the lifted problem. 
Since the lifted variables are already symmetric we can take them out of the projection operator due its linearity. Hence, for the update term we obtain
\begin{eqnarray}\label{eq:updt_term}
\nabla \mathcal{J}(\brho) & = & \mathcal{Y}(\brho) \brho = \tilde{\brho} \brho - \tilde{\brho_t} \brho +  \mathcal{P}_{S}(\bdelta(\tilde{\mathbf{e}})) \brho, \\
 & = & \| \brho \|^2 \brho - (\brho_t^H \brho) \brho_t  +  \mathcal{P}_{S}(\bdelta(\tilde{\mathbf{e}})) \brho.
\end{eqnarray}
Reprising the regularity condition under consideration, we need to establish that there exists constants $\alpha$ and $\beta$, such that $\alpha \beta > 4$ for all $\brho \in E(\epsilon)$ and
\begin{equation}
\label{eq:RegularityCond}
\mathrm{Re} \left( \langle \nabla \mathcal{J}(\brho),  (\brho - {\brho}_t e^{\mathrm{i} \Phi(\brho)}) \rangle \right) \geq \frac{1}{\alpha} \text{dist}^2 (\brho, {\brho}_t) + \frac{1}{\beta} \|  \nabla \mathcal{J}(\brho) \|^2.
\end{equation}

To show the existence of constants $\alpha$ and $\beta$ that satisfy \eqref{eq:RegularityCond}, we upper bound the gradient term, which converts the regularity condition to a \emph{restricted strong convexity condition} \cite{chen2017solving, sun2018geometric}. 
We begin the proof by introducing the following key lemmas. 
\begin{lemma}\label{lem:Lemma5}
Let $\brho_t$ be the ground truth signal with $\| \brho_t \| = 1$, and $\hat{\brho}_t$ denote the global solution closest to $\brho$ such that $\hat{\brho}_t = e^{\mathrm{i} \Phi(\brho)} \brho_t$. Then, for any $ \brho \in E(\epsilon)$, we have
\begin{equation}
(\sqrt{(1-\epsilon)(2-\epsilon)}) \| \brho - \hat{\brho}_t \| \leq \| \brho \brho^H - \brho_t \brho_t^H \|_F \leq  (2+\epsilon) \| \brho - \hat{\brho}_t \|.
\end{equation}
\end{lemma}
\begin{proof}
See Appendix \ref{sec:App5}.
\end{proof}

\begin{lemma}\label{lem:Lemma6}
Let $\brho_t$ be the ground truth signal with $\| \brho_t \| = 1$, and let the linear map $\mathcal{F}$ satisfy RIP$_+$ with RIC-$\delta_1$. Then, for $\delta_1 \leq 0.214$ and any $\brho \in E(\epsilon)$, we have $\delta_2 < 1$ such that
\begin{equation}\label{eq:RIP2}
(1 - \delta_2 ) \| \brho \brho^H - \brho_t \brho_t^H \|_F^2 \leq \| \mathcal{F}(\brho \brho^H - \brho_t \brho_t^H) \|^2  \leq (1 + \delta_2) \| \brho \brho^H - \brho_t \brho_t^H \|_F^2,
\end{equation}
where $ \delta_2 = \frac{\sqrt{2}(2+\epsilon)}{\sqrt{(1-\epsilon)(2-\epsilon})} \delta_1$.
\end{lemma}
We refer to \eqref{eq:RIP2} as the {local RIP-2 condition} in the lifted domain with RIC-$\delta_2$ for the mapping $\mathcal{F}$.
The two lemmas culminate into the local Lipschitz continuity of $\nabla \mathcal{J}$. 
\begin{proof}
See Appendix \ref{sec:App6}.
\end{proof}

\begin{lemma}\label{lem:Lemma7}
In the setup of Lemmas \ref{lem:Lemma5} and \ref{lem:Lemma6}, for any $\brho \in E(\epsilon)$, the objective function $\mathcal{J}$ in \eqref{eq:objfun_repr} is Lipschitz differentiable with
\begin{equation}
\| \nabla \mathcal{J}(\brho) \| \leq c \cdot \mathrm{dist} (\brho, \brho_t) 
\end{equation}
where $ c = (1+\epsilon)(2+\epsilon) (1+\delta_1)$ is the Lipschitz constant. Furthermore, to establish the regularity condition for $\mathcal{J}$, it is sufficient to show that
\begin{equation}
\label{eq:RegCond_new}
\mathrm{Re} \left( \langle \nabla \mathcal{J}(\brho),  (\brho - {\brho}_t e^{\mathrm{i} \Phi(\brho)}) \rangle \right) \geq (\frac{1}{\alpha} +  \frac{c^2}{\beta}) \mathrm{dist}^2 (\brho, {\brho}_t)
\end{equation}
for any $\brho \in E(\epsilon)$. 
\end{lemma}
\begin{proof}
See Appendix \ref{sec:App7}.
\end{proof}

We finally utilize the following lemma to obtain an alternative form of the restricted strong convexity condition \cite{zhang2015restricted}.   
\begin{lemma}\label{lem:Lemma8}
For the objective function $\mathcal{J}$ in \eqref{eq:objfun_repr}, the condition in \eqref{eq:RegCond_new} is satisfied if
\begin{equation}\label{eq:RegConditionFinal}
\mathcal{J}(\brho) \geq \frac{\eta}{2}  \mathrm{dist}^2 (\brho, {\brho}_t)
\end{equation}
where $\eta = \frac{1}{\alpha} +  \frac{c^2}{\beta}$.
\end{lemma}
\begin{proof}
See Appendix \ref{sec:App8}.
\end{proof}

Writing the objective function explicitly in terms of the lifted terms, and applying the lower bound from the local RIP-2 condition of the lifted forward model, we can express the regularity condition simply as
$$
\frac{1}{2} \| \mathcal{F} [\brho \brho^H - \brho_t \brho^H_t] \|^2 \geq \frac{(1-\delta_2)}{2} \| \brho \brho^H - \brho_t \brho^H_t \|_F^2 \geq \frac{\eta}{2}  \mathrm{dist}^2 (\brho, {\brho}_t),
$$
where $\delta_2$ is as defined in Lemma \ref{lem:Lemma6}. 
From Lemma \ref{lem:Lemma5}, the regularity condition is then satisfied by identifying $\alpha, \beta$ with $\alpha \beta > 4$ such that
$$
(1 - \delta_2) \| \brho \brho^H - \brho_t \brho^H_t \|_F^2 \geq {(1-\delta_2)}  (1- \epsilon) (2- \epsilon ) \text{dist}^2 (\brho, {\brho}_t) \geq  ( \frac{1}{\alpha} +  \frac{c^2}{\beta}) \mathrm{dist}^2 (\brho, {\brho}_t),
$$
where, from Lemma \ref{lem:Lemma7}, $c = (1+\epsilon)(2+\epsilon) (1+\delta_1)$. 

\subsection{Proof of Theorem \ref{thm:Theorem2}}\label{sec:Proof_Thm2}

\begin{lemma}{Expectation of Spectral Matrix.}\label{lem:Lemma9}
Let measurement vectors $\mathbf{L}_{i}^m, \mathbf{L}_{j}^m$ be statistically independent and distributed according to the complex Gaussian model as $\mathbf{L}_{i}, \mathbf{L}_{j}^m \sim \mathcal{N}(0, \frac{1}{2} \mathbf{I} ) + \mathrm{i} \mathcal{N}(0, \frac{1}{2} \mathbf{I})$. 
Let $\brho_t$ be independent of the measurement vectors, and $\mathbf{Y}$ denote the backprojection estimate of the lifted signal generated by the spectral method given as 
$$
\mathbf{Y}= \frac{1}{M} \mathcal{F}^H \mathcal{F} ( \brho_t \brho_t^H ),
$$
where $\mathcal{F}$ is the lifted forward map in \eqref{eq:LiftedLinMod}. Then,
$$
\mathbb{E}[ \mathbf{Y} ] = \brho_t \brho_t^H. 
$$
\end{lemma}
\begin{proof}
See Appendix \ref{sec:App9}.
\end{proof}

\begin{lemma}{Concentration Around Expectation.}\label{lem:Lemma10}
In the setup of Lemma \ref{lem:Lemma9}, assume that the number of measurements is $M = C(\delta) \cdot  N \log N$, where $C$ is a constant that depends on $\delta$. Then,
\begin{equation}\label{eq:lem10eqn}
\| \mathbf{Y} - \mathbb{E}[\mathbf{Y}] \| \leq \delta
\end{equation}
holds with probability at least $p = 1 - 8e^{-\gamma N} - 5N^{-2}$ where $\gamma$ is a fixed positive constant. 
\end{lemma}
\begin{proof}
See Appendix \ref{sec:App10}.
\end{proof}

Plugging in the expectation from Lemma \ref{lem:Lemma9} into \eqref{eq:lem10eqn}, and using the definition of $\bdelta$ from Lemma \ref{cor:Corr1}, the Lemmas \ref{lem:Lemma9} and \ref{lem:Lemma10} culminate to
\begin{equation}\label{eq:repPert}
\| \bdelta ( \brho_t \brho_t^H) \| \leq \delta
\end{equation} 
for \emph{any} $\brho_t$ with $\| \brho_t \| = 1$ with probability at least $p$. 
From the definition of the spectral norm on $\mathbb{C}^{N \times N}$, we can write the left-hand-side of \eqref{eq:repPert} equivalently as
\begin{equation}
\underset{\brho \in C^N, \| \brho \| = 1}{\text{max}} | \brho^H \bdelta(\brho_t \brho_t^H) \brho |  \leq \delta.
\end{equation}
Since the spectral norm corresponds to the maximum over the unit sphere in $\mathbb{C}^N$, we have
\begin{equation}
| \brho_t^H \bdelta(\brho_t \brho_t^H) \brho_t | \leq \| \bdelta ( \brho_t \brho_t^H) \|  \leq  \delta.
\end{equation}
Observe that the left-hand-side can equivalently be represented as a Frobenius inner product via lifting as 
\begin{equation}
| \brho_t^H \bdelta(\brho_t \brho_t^H) \brho_t |  = | \langle \bdelta(\brho_t \brho_t^H), \brho_t \brho_t^H \rangle_F |. 
\end{equation}
Having Lemmas \ref{lem:Lemma9} and \ref{lem:Lemma10} hold for any element $\brho_t \in \mathbb{C}^N$ with $\| \brho_t \| = 1$ via unitary invariance, for any $\brho \in \mathbb{C}^N$ we obtain
\begin{equation}
 \frac{| \langle (\frac{1}{M}\mathcal{F}^H \mathcal{F} - \mathbf{I}) (\brho \brho^H), {\brho}\brho^H \rangle_F |}{\| {\brho}\brho^H \|_F^2} \leq \delta,
\end{equation}
which yields
\begin{equation}
(1-\delta){\| {\brho}\brho^H \|_F^2}  \leq \frac{1}{M} \| \mathcal{F}(\brho \brho^H) \|^2 \leq (1 + \delta) {\| {\brho}\brho^H \|_F^2}
\end{equation}
for any $\brho \in \mathbb{C}^N$. 
Therefore, RIP$_+$ with RIC-$\delta$ is established with probability at least $p$ for mapping $\frac{1}{\sqrt{M}}\mathcal{F}$ with $\mathbf{L}_{i}^m, \mathbf{L}_{j}^m \sim \mathcal{N}(0, \frac{1}{2} \mathbf{I} ) + \mathrm{i} \mathcal{N}(0, \frac{1}{2} \mathbf{I})$ where $M = C(\delta) \cdot  N \log N$. 

Furthermore, we know that the true lifted unknown $\brho_t \brho_t^H$ lies in the positive semi-definite cone, which is a convex set. 
Since the spectral matrix $\hat{\mathbf{X}}$ is the projection of $\frac{1}{M}\mathcal{F}^H \mathcal{F}(\brho_t \brho_t^H)$ onto the set of Hermetian symmetric matrices, from the non-expansiveness property of projections onto convex sets we have
$$
\| \hat{\mathbf{X}} - \brho_t \brho_t^H \| \leq \| \frac{1}{M}\mathcal{F}^H \mathcal{F}(\brho_t \brho_t^H) - \brho_t \brho_t^H \| \leq \delta,
$$
which completes the proof of Theorem \ref{thm:Theorem2}.

\section{Numerical Simulations}\label{sec:NumSim}
\subsection{Signal Recovery From Random Gaussian Measurements}\label{sec:Nsim1}
We begin by considering recovery of random signals from cross-correlations of complex random Gaussian measurements, $\mathbf{L}^m_i, \mathbf{L}^m_j \sim \mathcal{N}(\mathbf{0}, \frac{1}{2} \mathbf{I}) + \mathrm{i} \mathcal{N}(\mathbf{0}, \frac{1}{2} \mathbf{I})$. 
For our numerical evaluations of the Gaussian model, we conduct an experiment similar to that of \cite{candes2015phase}.
We set $N = 128$, and run $100$ instances of interferometric inversion by GWF with independently sampled Gaussian measurement vectors on two types of signals: random low-pass signals, $\brho^{LP}$, and random Gaussian signals, $\brho^G$.
The entries of the signals are generated independently of the measurement vectors at each instance by  
\begin{equation}
\rho^{LP}_l = \sum_{p = - \frac{P}{2} +1}^{\frac{P}{2}} ( X_l + \mathrm{i} Y_l ) e^{\frac{2\pi \mathrm{i}(p-1)(l-1)}{N}}, \quad \rho^G_l = \sum_{p = - \frac{N}{2} +1}^{\frac{N}{2}} \frac{1}{\sqrt{8}}( X_l + \mathrm{i} Y_l ) e^{\frac{2\pi \mathrm{i}(p-1)(l-1)}{N}},
\end{equation}
where $P = N/8$, and $X_l$ and $Y_l$ are i.i.d. $\mathcal{N}(0,1)$. 
%
As described in \cite{candes2015phase}, random low-pass signal corresponds to a bandlimited version of this random model and variances are adjusted so that the expected signal power is the same. 

\begin{figure}
\centering
\includegraphics[scale=0.315]{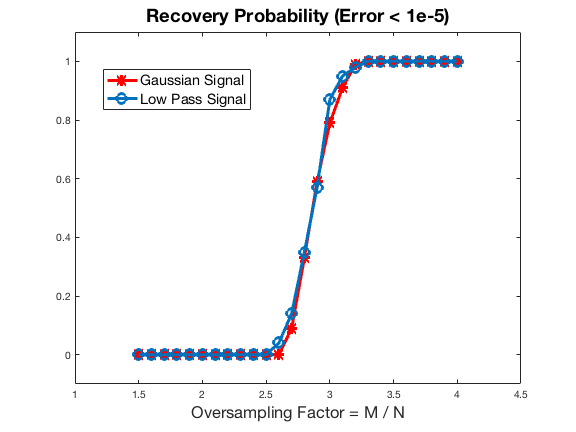}
\includegraphics[scale=0.32]{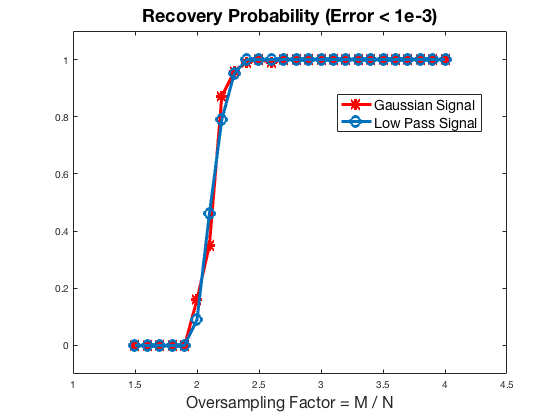}
\caption{\emph{Empirical recovery probabilities based on $100$ random trials vs. oversampling factor of the number of measurements $M/N$.} The red curves correspond to empirical recovery probability of the Gaussian signal, whereas the blue curve corresponds to that of realization of the random low-pass signal model. The two figures vary with respect to the values of success criterion assumed for successful recovery, with a relative error of $10^{-5}$ and $10^{-3}$, respectively.}
\label{fig:Figure2}
\end{figure}

We implement the GWF algorithm with the learning rate heuristic of WF in \cite{candes2015phase} such that the descent algorithm takes smaller steps initially due to higher inaccuracy of the iterates. The step size is gradually increased such that $\mu_k  =  \text{min}( 1 - \mathrm{e}^{k/\tau_0}, \mu_{max})$, where $\tau_0 = 33000$, and $\mu_{max} = 0.2$. For $2500$ iterations, the learning parameter corresponds to a nearly linear regime and attains the maximum value of $0.073$. 

In the experimentation, we compute the empirical probability of success after $2500$ iterations by counting the exact recovery instances of GWF recovery from different realizations of Gaussian measurements for $\{ \mathbf{L}_{i}^m, \mathbf{L}_{j}^m\}_{m=1}^M$.  
We evaluate the exact recovery by the relative normalized error of the final estimate, $\brho_{GWF}$, such that $\text{dist}(\brho_{GWF} , \brho_t) / \| \brho_t \| \leq  err = 10^{-5}$. In addition, we evaluate the probability of moderately precise recovery by setting $err = 10^{-3}$.
As shown in Figure \ref{fig:Figure2}, our experimentation indicates that beginning with $3N$ interferometric Gaussian measurements, GWF achieves exact recovery with high probabilities. 
Furthermore, the method provides robust recovery with as low as $2.3N$ interferometric Gaussian measurements, with a relative error of $10^{-3}$ and below at over $95$ percent. 

\subsection{Multistatic Passive Radar Imaging}

An interferometric inversion problem of great interest is multistatic passive radar imaging. 
We consider an imaging set-up in which several static, terrestrial receivers are placed in a circle of radius around the scene of interest, which is illuminated by a transmitter of opportunity.
An exemplary multistatic imaging geometry is illustrated in Figure \ref{fig:Figure3}.
At receiver $i$, the back-scattered signal is collected at a fixed location by a linear map $\mathbf{L}_i$ parameterized by the temporal frequency variable $\omega$. 
The linear measurements collected at two receivers $i$ and $j$ are then pairwise correlated in time to yield the cross-correlation model \eqref{eq:interferom} defined in the temporal frequency domain. 

The key advantage of the interferometric model is the elimination of the dependence of measurements on the transmitter location and phase of the transmitted waveform, both of which are unknown in the passive scenario. 
In prior studies, the interferometric wave-based imaging was approached by low rank recovery methods \cite{son2015passive, Mason2015, yonel2018generalized}. 
We postulate that GWF framework provides a computationally and memory-wise efficient alternative to LRMR based passive radar imaging. 

\subsubsection{Received Signal Model}
Let $\a_i^r\in\mathbb{R}^3$ denote the spatial locations of the receivers, and assume $S$ number of receivers such that $i = 1, 2, \cdots S$.
We assume that scattered signals are due to a single source of opportunity located at $\a^t$.
The location on the surface of the earth is denoted by $\x = (\bi x,\psi(\bi x)) \in \mathbb{R}^3$, where $\bi x = (x_1,x_2) \in \Rb^2$ and $\psi: \mathbb{R}^2\rightarrow \mathbb{R}$ is a known ground topography; and ${\rho} : \mathbb{R}^2 \rightarrow \mathbb{R}$ denotes the ground reflectivity.

\begin{figure}
\centering
\includegraphics[width=2.5in]{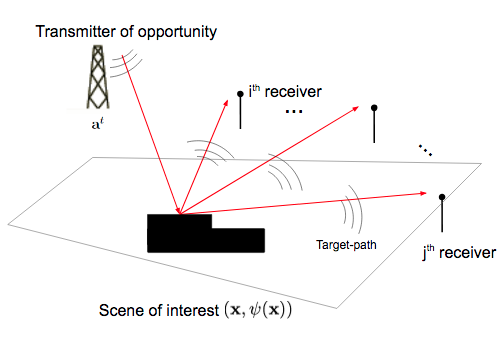}
\includegraphics[width=2.5in]{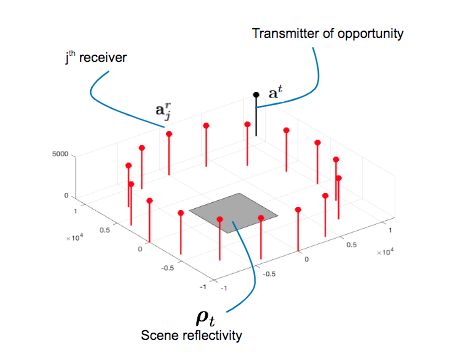}
\caption{\emph{An illustration of a multistatic imaging configuration. A scene is illuminated by a stationary illuminator of opportunity, located at $\mathbf{a}^t$. The backscattered signal is measured by a collection of stationary receivers, encircling the scene of interest at locations $\mathbf{a}^r_j$.}}
\label{fig:Figure3}
\end{figure}

Under the flat topography assumption and Born approximation, and assuming waves propagate in free space, the fast-time temporal Fourier transform of the received signal at the $i$-th receiver can be modeled as \cite{Yarman08}
\begin{equation}
f_i(\omega) \approx \Lc_i[{\rho}](\omega,s) := p(\o, s) \int_D 
  \mathrm{e}^{-\mathrm{i}\omega\frac{\phi_i(\x)}{c_0}}\alpha_i(\bi {x}, \a^t)\rho(\bi {x}) d \bi{x}, 
\label{eq:received_signal}
\end{equation}
where $\o$ is the temporal frequency variable,$c_0$ is the speed-of-light in free space, $p(\o, s)$ is the transmitted waveform, $\alpha_i(\bi x, \a^t)$ is the azimuth beam pattern, and
\begin{equation}
\phi_{i}(\bi x) = |\x-\a_i^r| + |\x-\a^t |,
\label{eq:range}
\end{equation}
is the bi-static delay term.

\begin{figure}
\centering
\includegraphics[scale=0.325]{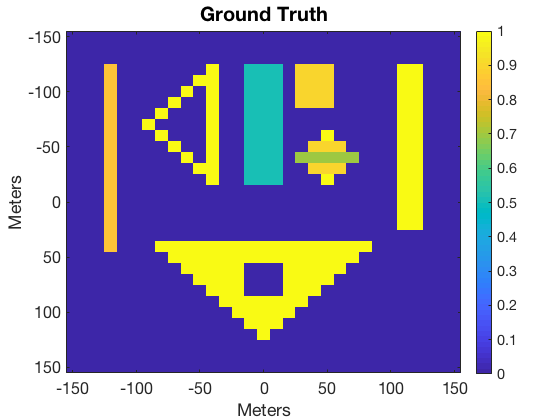}
\includegraphics[scale=0.325]{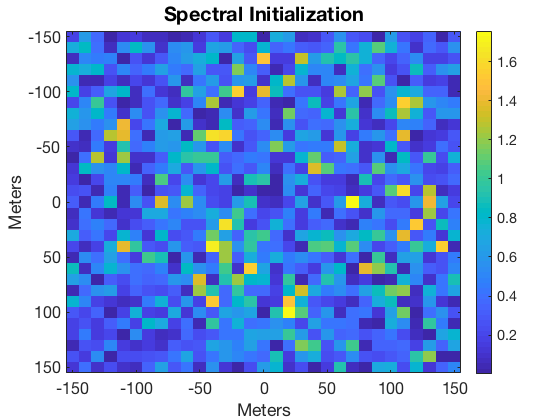}
\caption{\emph{The ground truth image and the initial estimate by the spectral method of GWF by Equations \eqref{eq:spectralInitGWF} and \eqref{eq:GWF_initp}}.}
\label{fig:Figure4}
\end{figure}

Following the derivations in \cite{Mason2015} and discretizing the domain $D$ of $\rho$ into $N$ samples, the cross-correlated data model for the multi-static scenario is obtained as
\begin{equation}
d_{ij}(\o) =  \sum_{k = 1}^N  e^{- \mathrm{i} \o  (|\x_k - \a_i^r | + \hat{\a}^t \cdot \x_k) / c_0}  \brho_k \sum_{k' = 1}^N e^{\mathrm{i} \o (|{\x}_{k'} -  \a_j^r| + \hat{\a}^t \cdot \x_{k'})}  \overline{\brho_{k'}}
\label{eq:discrete_datamod}
\end{equation}
where $\brho_k =\rho(\x_k)$ is the $k^{th}$ entry of the discretized scene reflectivity vector $\brho \in \mathbb{C}^N$, and $\hat{\a}^t$ it the unit vector in the direction of the $\a^t$. 
We next discretize the temporal frequency domain $\Omega$ into $M'$ samples and define the measurement vectors with entries
\begin{equation}
[ \mathbf{L}^{m'}_{i}]_k =  e^{\mathrm{i} \o_{m'} (|\x_k  - \a_{i}^r | - \hat{\a}^t \cdot \x_k) / c_0} \quad k = 1,...,N.
\label{eq:discreteL_Def2}
\end{equation}
Rewriting (\ref{eq:discrete_datamod}) using (\ref{eq:discreteL_Def2}), we obtain the interferometric measurement model as:
\begin{equation} \label{eq:discreteL_Def3}
d^{ij}(\o_{m'}) = \langle \mathbf{L}^{m'}_i, \ \brho \rangle \overline{\langle \mathbf{L}^{m'}_j, \ \brho \rangle}  \quad m' = 1,...,M', \ i = 1,...,S, \ j \neq i,
\end{equation}
which corresponds to total of $M = M' \binom{S}{2}$ cross-correlated measurements\footnote{Note that since we have multiple receivers to cross-correlate, the $m = 1, \cdots M$ summation of the objective function in \eqref{eq:GWF_opt} now consists of two sums, one over all unique cross-correlations $i \neq j$, and one over the temporal frequency index $m'$.} 
In \cite{son2019exact}, we show that the model in \eqref{eq:discrete_datamod} satisfies the sufficient condition of Theorem \ref{thm:Theorem1}, and hence GWF can provide exact image reconstruction for multi-static passive radar.  


\subsubsection{Simulation Setup and Results}
We assume isotropic transmit and receive antennas, and simulate a transmitted signal with $20$MHz bandwidth and $1$GHz center frequency.
We assume isotropic transmit and receive antennas, and simulate a transmitted signal with $20$MHz bandwidth and $1$GHz center frequency.
We place the center of the scene at the origin of the coordinate system, and generate a phantom, which consists of multiple point and extended targets as depicted in Figure \ref{fig:Figure4}. 
The transmitter is fixed and located at coordinates $\a^t = [11.5, 11.5, 0.5]$ km. 
We simulate a flat spectrum waveform, and sample the temporal frequency domain into $32$ samples. 

\emph{i) GWF Reconstruction:} To evaluate the performance of GWF for interferometric multistatic radar imaging, we simulate a $300\times 300 \mathrm{m}^2$ scene, and discretize it by $10$m, which corresponds to $31\times 31$ pixels, hence, $N = 961$. 
We use $16$ receiver antennas that are placed in a circle of radius $10$km around the scene at height of $0.5$km, which corresponds to $120$ unique cross-correlations at each temporal frequency sample. 
We generate the back-scattered signals at each receiver by the linear measurement map of \eqref{eq:received_signal}, and correlate linear measurements of the each pairwise combination of receivers to generate interferometric data.   
In reconstruction, we deploy the approximate measurement vectors in which the transmitter distance is removed as defined in \eqref{eq:discrete_datamod} and \eqref{eq:discreteL_Def2}, hence only the transmitter look-direction is used at recovery by GWF. 


\begin{figure}
\centering
\includegraphics[scale=0.325]{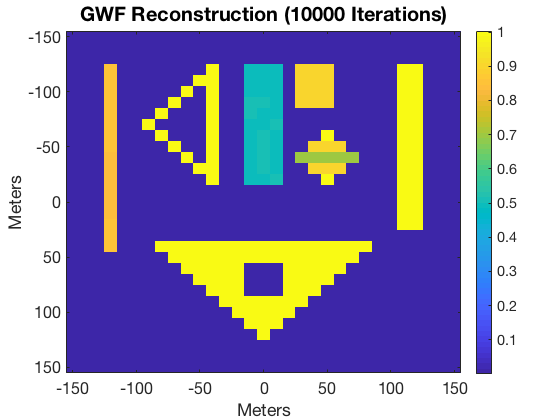}
\includegraphics[scale=0.325]{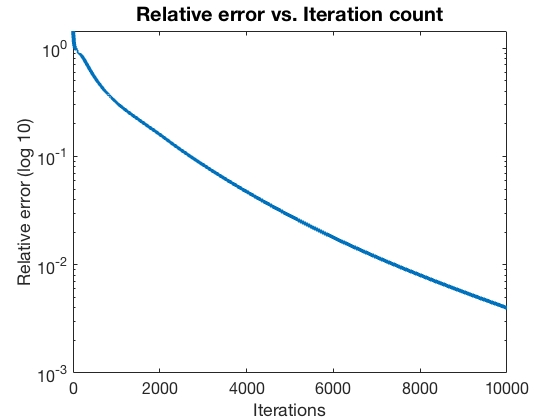}
\caption{\emph{Reconstructed image by GWF after 10000 iterations of \eqref{eq:GWF_Updates} and the relative error with respect to the ground truth in $\log$-scale vs. iterations.}}
\label{fig:Figure5}
\end{figure}

In Figure \ref{fig:Figure5}, we demonstrate the reconstruction obtained by GWF after $10000$ iterations, using the gradually increasing step size heuristic described in Section \ref{sec:Nsim1}. 
In addition, we plot the relative error of GWF iterates, which displays a geometric rate of convergence as stated in Theorem \ref{thm:Theorem1}. 
The numerical results on our phantom image indicate GWF has the capacity to form accurate imagery of complex extended scenes. 

\emph{ii) Comparison to LRMR by Uzawa's Method:} We next perform smaller sized experiments to compare with Uzawa's method for LRMR by lifting. Namely, we compare GWF to three different formulations that yield Uzawa's iterations of the form in \eqref{eq:SVT}. These are the convex-relaxed trace regularization problem solved by \cite{Mason2015}, the rank-1 constrained non-convex version with projections defined in \eqref{eq:Uzaw_proj}, and the convex formulation obtained by dropping the rank constraint fully, i.e., a projected gradient descent analogue of \cite{Demanet14}\footnote{Note that the original proposed method in \cite{Demanet14} follows the Douglas-Rachford splitting framework to solve the non-trace regularized convex program in the lifted domain, which has different exact recovery arguments to Uzawa's method. To avoid the computational burden required to solve a linear system in the lifted domain, we use the projected gradient approach under the identical problem formulation.}. Notably, we assess the reconstruction performance with respect to amount of computations, and run the algorithms for identical number of flops.
For our problem size, we run GWF for $8000$ iterations, which roughly corresponds to $56$ iterations of Uzawa's method. 
As complexity is only known upto an order, to avoid from positive bias towards GWF we run the variations of the Uzawa's method for $110$ iterations, i.e., double of what the mismatch in computation from lifting the unknown indicates. 

The simulated scene corresponds to $144 \times 144$m, discritized uniformly by $12$m, to yield a $12 \times 12$ unknown, hence $N = 144$. 
We use $12$ receiver antennas that are placed in a circle of radius $10$km around the scene at height of $0.5$km, yielding $66$ unique cross-correlations at each temporal frequency sample. 
The transmitted waveform bandwidth is set at $10$MHz, around $1.9$GHz central frequency, which corresponds to imaging below the range resolution limit of the system \cite{son2019exact}. 
We provide mean squared error with respect to number of flops in computation over the lifted and the signal domains, and the final images reconstructed by each algorithm in Figures \ref{fig:Figure6} and \ref{fig:Figure7}, respectively. 
Overall, our results indicate that GWF achieves exact reconstruction considerably faster than LRMR by Uzawa's method when compared at identical number of flops. 
Furthermore, the convergent behavior of the convex LRMR methods under consideration supports our observation that their exact recovery guarantees are more stringent than those of GWF. 
Thereby, our numerical results demonstrate its superior performance at identical number of flops to lifting based approaches, and promote GWF as a favorable alternative to state of the art interferometric wave-based imaging methods. 

\begin{figure}[!t]
\centering
\includegraphics[scale=0.325]{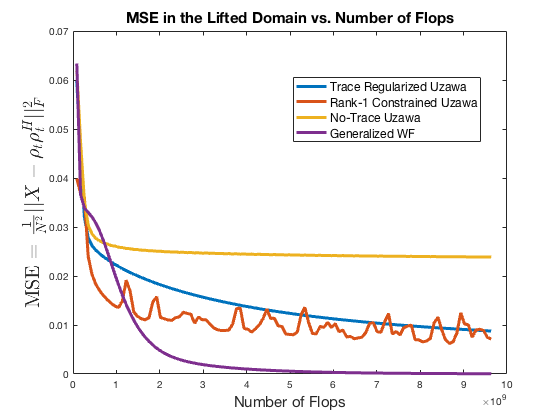}
\includegraphics[scale=0.325]{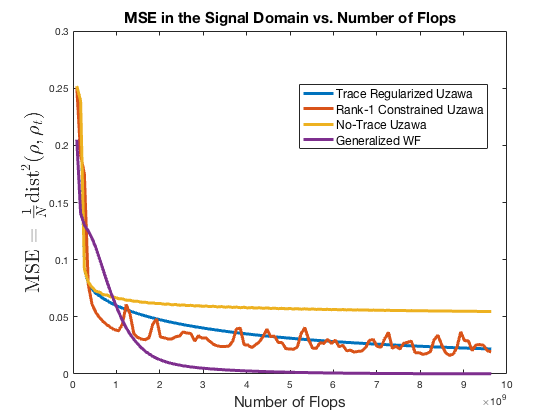}
\caption{Mean squared error curves in the lifted and signal domain vs. number of flops, respectively.}
\label{fig:Figure6}
\end{figure}

\begin{figure}[!t]
\centering
\subfloat[Ground truth scene.]{\includegraphics[scale=0.2]{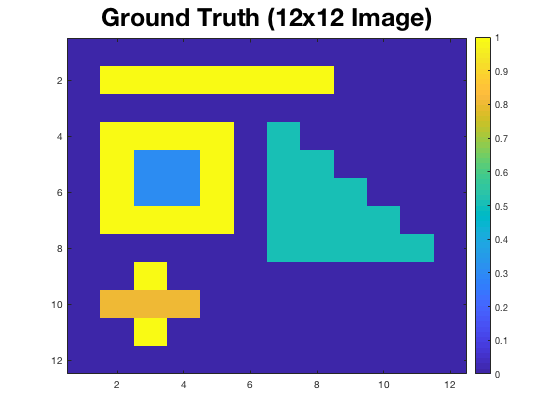}}\hfil
\subfloat[Spectral initialization.]{\includegraphics[scale=0.2]{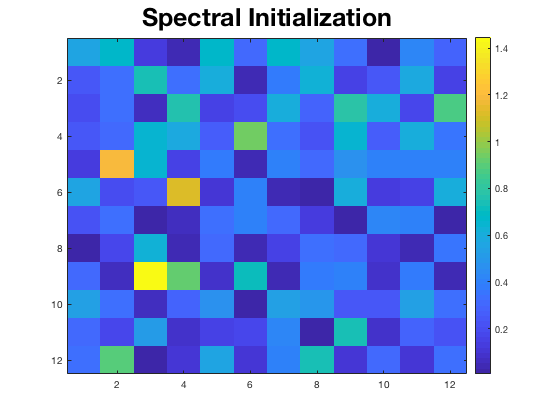}}\hfil 
\subfloat[GWF Reconstruction.]{\includegraphics[scale=0.2]{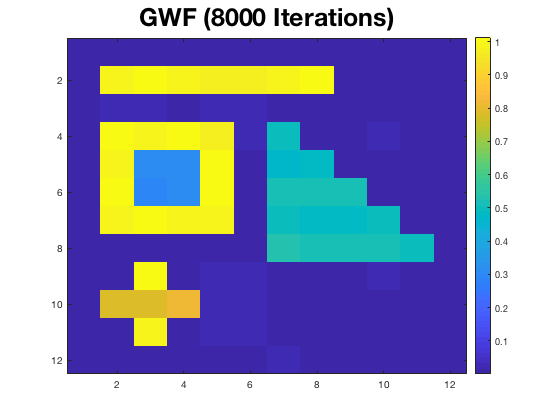}} 

\subfloat[Uzawa's Method with trace regularization.]{\includegraphics[scale=0.2]{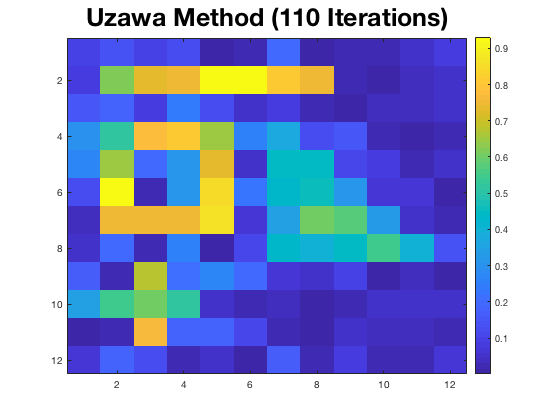}}\hfil   
\subfloat[Uzawa's Method with rank-1 constraint.]{\includegraphics[scale=0.2]{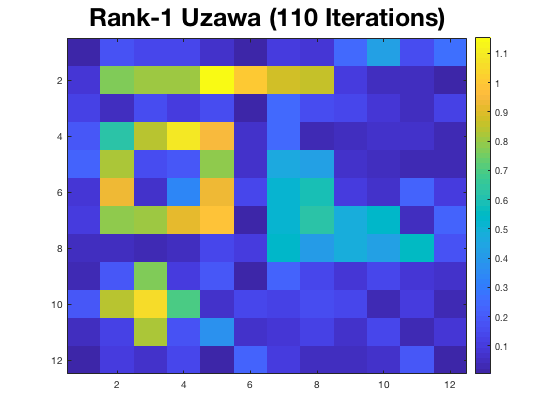}}\hfil
\subfloat[Uzawa's Method with only PSD constraint.]{\includegraphics[scale=0.2]{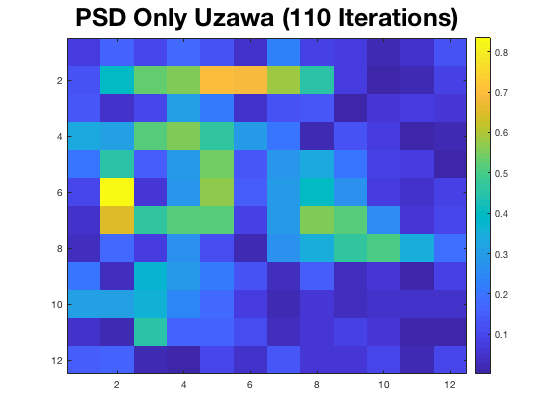}}
\caption{Reconstruction results of GWF and variations of Uzawa's method at identical number of flops ($8000$ iterations for GWF, $110$ iterations for Uzawa's methods.}
\label{fig:Figure7}
\end{figure}





\section{Conclusion}\label{sec:Conc}

In this paper, we present a novel framework for exact interferometric inversion. 
We approach the interferometric inversion problem from the perspective of phase retrieval techniques. 
We examine two of the most prominent phase retrieval methods, namely LRMR based PhaseLift, and non-convex optimization based WF, and bridge the theory between the two frameworks. 
We then generalize WF and formulate the GWF framework for interferometric inversion, and extend the exact recovery guarantees to arbitrary measurement maps with properties that are characterized in the equivalent lifted problem. 
Thereby, we establish exact recovery conditions for a larger class of problems than that of standard WF. 

We identify the sufficient conditions for exact interferometric inversion on the lifted forward model as the restricted isometry property on rank-1, PSD matrices, with a restricted isometry constant of $\delta \leq 0.214$. 
In developing our theory, we use the special structure of the rank-1, PSD set of matrices to show that the RIP directly implies the regularity condition of WF. 
Furthermore, we show
that the concentration bound of the spectral matrix directly implies the RIP over rank-1, PSD matrices for cross-correlations of the complex Gaussian model. 
Hence, in generalizing the theory of WF for interferometric inversion in the complex Gaussian case, we demonstrate that the regularity condition becomes redundant. 
We illustrate that the empirical probability of exact interferometric inversion by GWF require smaller oversampling factors than that of phase retrieval in the Gaussian model. 
Finally, we demonstrate the applicability of GWF in a deterministic, passive multi-static radar imaging problem using realistic imaging parameters. 
In conclusion, our paper shows that the computational and theoretical advantages promote GWF as a practical technique in real-world imaging applications. 


\section*{Acknowledgement}
This work was supported by the Air Force Office of Scientific Research (AFOSR) under the agreement FA9550-16-1-0234, Office of Naval Research (ONR) under the agreement N0001418-1-2068 and by the National Science Foundation (NSF) under Grant No ECCS-1809234.

\bibliographystyle{IEEEtran}
\bibliography{ref1}

\begin{thebibliography}{10}
\providecommand{\url}[1]{#1}
\csname url@samestyle\endcsname
\providecommand{\newblock}{\relax}
\providecommand{\bibinfo}[2]{#2}
\providecommand{\BIBentrySTDinterwordspacing}{\spaceskip=0pt\relax}
\providecommand{\BIBentryALTinterwordstretchfactor}{4}
\providecommand{\BIBentryALTinterwordspacing}{\spaceskip=\fontdimen2\font plus
\BIBentryALTinterwordstretchfactor\fontdimen3\font minus
  \fontdimen4\font\relax}
\providecommand{\BIBforeignlanguage}[2]{{%
\expandafter\ifx\csname l@#1\endcsname\relax
\typeout{** WARNING: IEEEtran.bst: No hyphenation pattern has been}%
\typeout{** loaded for the language `#1'. Using the pattern for}%
\typeout{** the default language instead.}%
\else
\language=\csname l@#1\endcsname
\fi
#2}}
\providecommand{\BIBdecl}{\relax}
\BIBdecl

\bibitem{candes2015phase}
E.~J. Candes, X.~Li, and M.~Soltanolkotabi, ``Phase retrieval via wirtinger
  flow: Theory and algorithms,'' \emph{IEEE Transactions on Information
  Theory}, vol.~61, no.~4, pp. 1985--2007, 2015.

\bibitem{goldstein1987interferometric}
R.~M. Goldstein and H.~Zebker, ``Interferometric radar measurement of ocean
  surface currents,'' \emph{Nature}, vol. 328, no. 6132, p. 707, 1987.

\bibitem{saebo2010seafloor}
T.~O. Saebo, ``Seafloor depth estimation by means of interferometric synthetic
  aperture sonar,'' 2010.

\bibitem{bamler1998synthetic}
R.~Bamler and P.~Hartl, ``Synthetic aperture radar interferometry,''
  \emph{Inverse problems}, vol.~14, no.~4, p.~R1, 1998.

\bibitem{Yarman08}
C.~Yarman and B.~Yazici, ``Synthetic aperture hitchhiker imaging,'' \emph{IEEE
  Trans. Image Processing}, pp. 2156--2173, 2008.

\bibitem{Wang11}
L.~Wang, C.~Yarman, and B.~Yazici, ``{Doppler-Hitchhiker}: A novel passive
  synthetic aperture radar using ultranarrowband sources of opportunity,''
  \emph{Geoscience and Remote Sensing, IEEE Trans.}, vol.~49, no.~10, pp.
  3521--3537, Oct 2011.

\bibitem{Wang10}
L.~Wang, I.-Y. Son, and B.~Yazici, ``Passive imaging using distributed
  apertures in multiple-scattering environments,'' \emph{Inverse Problems},
  vol.~26, no.~6, p. 065002, 2010.

\bibitem{yarman2008bistatic}
C.~E. Yarman, B.~Yazici, and M.~Cheney, ``Bistatic synthetic aperture radar
  imaging for arbitrary flight trajectories,'' \emph{IEEE Transactions on Image
  Processing}, vol.~17, no.~1, pp. 84--93, 2008.

\bibitem{Yarman10}
\BIBentryALTinterwordspacing
C.~Yarman, L.~Wang, and B.~Yazici, ``Doppler synthetic aperture hitchhiker
  imaging,'' \emph{Inverse Problems}, vol.~26, no.~6, p. 065006, 2010.
  [Online]. Available: \url{http://stacks.iop.org/0266-5611/26/i=6/a=065006}
\BIBentrySTDinterwordspacing

\bibitem{LWang12_b}
L.~Wang and B.~Yazici, ``Bistatic synthetic aperture radar imaging using
  narrowband continuous waveforms,'' \emph{Image Processing, IEEE Transactions
  on}, vol.~21, no.~8, pp. 3673--3686, Aug 2012.

\bibitem{Mason2015}
E.~Mason, I.~Y. Son, and B.~Yaz{\i}c{\i}, ``Passive synthetic aperture radar
  imaging using low-rank matrix recovery methods,'' \emph{IEEE Journal of
  Selected Topics in Signal Processing}, vol.~9, no.~8, pp. 1570--1582, Dec
  2015.

\bibitem{son2015passive}
I.-Y. Son and B.~Yazici, ``Passive imaging with multistatic polarimetric
  radar,'' in \emph{Radar Conference (RadarCon), 2015 IEEE}.\hskip 1em plus
  0.5em minus 0.4em\relax IEEE, 2015, pp. 1584--1589.

\bibitem{Wang13_3}
L.~Wang and B.~Yazici, ``Ground moving target imaging using ultranarrowband
  continuous wave synthetic aperture radar,'' \emph{Geoscience and Remote
  Sensing, IEEE Trans.}, vol.~51, no.~9, pp. 4893--4910, Sept 2013.

\bibitem{Wang14}
------, ``Bistatic synthetic aperture radar imaging of moving targets using
  ultra-narrowband continuous waveforms,'' \emph{SIAM Journal on Imaging
  Sciences}, vol.~7, no.~2, pp. 824--866, 2014.

\bibitem{LWang12}
------, ``Passive imaging of moving targets using sparse distributed
  apertures,'' \emph{SIAM Journal on Imaging Sciences}, vol.~5, no.~3, pp.
  769--808, 2012.

\bibitem{wang12}
------, ``Passive imaging of moving targets exploiting multiple scattering
  using sparse distributed apertures,'' \emph{Inverse Problems}, vol.~28,
  no.~12, p. 125009, 2012.

\bibitem{son2017passive}
I.-Y. Son and B.~Yazici, ``Passive polarimetric multistatic radar detection of
  moving targets,'' \emph{arXiv preprint arXiv:1706.09369}, 2017.

\bibitem{Wacks14}
S.~Wacks and B.~Yazici, ``Passive synthetic aperture hitchhiker imaging of
  ground moving targets - part 1: Image formation and velocity estimation,''
  \emph{IEEE Trans. Image Processing}, vol.~23, no.~6, pp. 2487--2500, June
  2014.

\bibitem{Wacks14_2}
------, ``Passive synthetic aperture hitchhiker imaging of ground moving
  targets part 2: Performance analysis,'' \emph{IEEE Trans. Image Processing},
  vol.~23, no.~9, pp. 4126--4138, Sept 2014.

\bibitem{ammari2013passive}
H.~Ammari, J.~Garnier, and W.~Jing, ``Passive array correlation-based imaging
  in a random waveguide,'' \emph{Multiscale Modeling \& Simulation}, vol.~11,
  no.~2, pp. 656--681, 2013.

\bibitem{ralston2007interferometric}
T.~S. Ralston, D.~L. Marks, P.~S. Carney, and S.~A. Boppart, ``Interferometric
  synthetic aperture microscopy,'' \emph{Nature Physics}, vol.~3, no.~2, p.
  129, 2007.

\bibitem{stoica2007probing}
P.~Stoica, J.~Li, and Y.~Xie, ``On probing signal design for mimo radar,''
  \emph{IEEE Transactions on Signal Processing}, vol.~55, no.~8, pp.
  4151--4161, 2007.

\bibitem{patwari2005locating}
N.~Patwari, J.~N. Ash, S.~Kyperountas, A.~O. Hero, R.~L. Moses, and N.~S.
  Correal, ``Locating the nodes: cooperative localization in wireless sensor
  networks,'' \emph{IEEE Signal processing magazine}, vol.~22, no.~4, pp.
  54--69, 2005.

\bibitem{gezici2005localization}
S.~Gezici, Z.~Tian, G.~B. Giannakis, H.~Kobayashi, A.~F. Molisch, H.~V. Poor,
  and Z.~Sahinoglu, ``Localization via ultra-wideband radios: a look at
  positioning aspects for future sensor networks,'' \emph{IEEE signal
  processing magazine}, vol.~22, no.~4, pp. 70--84, 2005.

\bibitem{guivant2001optimization}
J.~E. Guivant and E.~M. Nebot, ``Optimization of the simultaneous localization
  and map-building algorithm for real-time implementation,'' \emph{IEEE
  transactions on robotics and automation}, vol.~17, no.~3, pp. 242--257, 2001.

\bibitem{garnier2005imaging}
J.~Garnier, ``Imaging in randomly layered media by cross-correlating noisy
  signals,'' \emph{Multiscale Modeling \& Simulation}, vol.~4, no.~2, pp.
  610--640, 2005.

\bibitem{lobkis2001emergence}
O.~I. Lobkis and R.~L. Weaver, ``On the emergence of the green’s function in
  the correlations of a diffuse field,'' \emph{The Journal of the Acoustical
  Society of America}, vol. 110, no.~6, pp. 3011--3017, 2001.

\bibitem{blomgren2002super}
P.~Blomgren, G.~Papanicolaou, and H.~Zhao, ``Super-resolution in time-reversal
  acoustics,'' \emph{The Journal of the Acoustical Society of America}, vol.
  111, no.~1, pp. 230--248, 2002.

\bibitem{gough2004displaced}
P.~Gough and M.~Miller, ``Displaced ping imaging autofocus for a
  multi-hydrophone sas,'' \emph{IEE Proceedings-Radar, Sonar and Navigation},
  vol. 151, no.~3, pp. 163--170, 2004.

\bibitem{flax1988phase}
S.~Flax and M.~O'Donnell, ``Phase-aberration correction using signals from
  point reflectors and diffuse scatterers: Basic principles,'' \emph{IEEE
  transactions on ultrasonics, ferroelectrics, and frequency control}, vol.~35,
  no.~6, pp. 758--767, 1988.

\bibitem{mason2016robustness}
E.~Mason and B.~Yazici, ``Robustness of lrmr based passive radar imaging to
  phase errors,'' in \emph{EUSAR 2016: 11th European Conference on Synthetic
  Aperture Radar, Proceedings of}.\hskip 1em plus 0.5em minus 0.4em\relax VDE,
  2016, pp. 1--4.

\bibitem{jain2003cross}
B.~Jain and A.~Taylor, ``Cross-correlation tomography: measuring dark energy
  evolution with weak lensing,'' \emph{Physical Review Letters}, vol.~91,
  no.~14, p. 141302, 2003.

\bibitem{copeland2006dynamics}
E.~J. Copeland, M.~Sami, and S.~Tsujikawa, ``Dynamics of dark energy,''
  \emph{International Journal of Modern Physics D}, vol.~15, no.~11, pp.
  1753--1935, 2006.

\bibitem{schotland2010quantum}
J.~C. Schotland, ``Quantum imaging and inverse scattering,'' \emph{Optics
  letters}, vol.~35, no.~20, pp. 3309--3311, 2010.

\bibitem{Candes13a}
E.~Candes, Y.~Eldar, T.~Strohmer, and V.~Voroninski, ``Phase retrieval via
  matrix completion,'' \emph{SIAM Journal on Imaging Sciences}, vol.~6, no.~1,
  pp. 199--225, 2013.

\bibitem{Candes13b}
E.~Candes and T.~Strohmer, ``Phaselift: Exact and stable recovery from
  magnitude measurements via convex programming,'' \emph{Comm Pure Appl Math},
  vol.~66, pp. 1241--1274, 2013.

\bibitem{Waldspurger2015}
\BIBentryALTinterwordspacing
I.~Waldspurger, A.~d'Aspremont, and S.~Mallat, ``Phase recovery, maxcut and
  complex semidefinite programming,'' \emph{Mathematical Programming}, vol.
  149, no.~1, pp. 47--81, 2015. [Online]. Available:
  \url{http://dx.doi.org/10.1007/s10107-013-0738-9}
\BIBentrySTDinterwordspacing

\bibitem{chen2015solving}
Y.~Chen and E.~Candes, ``Solving random quadratic systems of equations is
  nearly as easy as solving linear systems,'' in \emph{Advances in Neural
  Information Processing Systems}, 2015, pp. 739--747.

\bibitem{chen2017}
J.~Chen, L.~Wang, X.~Zhang, and Q.~Gu, ``Robust wirtinger flow for phase
  retrieval with arbitrary corruption,'' \emph{arXiv preprint
  arXiv:1704.06256}, 2017.

\bibitem{Zhang2017a}
H.~Zhang, Y.~Chi, and Y.~Liang, ``Median-truncated nonconvex approach for phase
  retrieval with outliers,'' 2017, preprint.

\bibitem{zhang2016reshaped}
H.~Zhang and Y.~Liang, ``Reshaped wirtinger flow for solving quadratic systems
  of equations,'' \emph{stat}, vol. 1050, p.~25, 2016.

\bibitem{Zhang2017b}
H.~Zhang, Y.~Zhou, Y.~Liang, and Y.~Chi, ``Reshaped wirtinger flow and
  incremental algorithms for solving quadratic systems of equations,'' 2017,
  preprint.

\bibitem{wang2018solving}
G.~Wang, G.~B. Giannakis, and Y.~C. Eldar, ``Solving systems of random
  quadratic equations via truncated amplitude flow,'' \emph{IEEE Transactions
  on Information Theory}, vol.~64, no.~2, pp. 773--794, 2018.

\bibitem{bendory2018non}
T.~Bendory, Y.~C. Eldar, and N.~Boumal, ``Non-convex phase retrieval from stft
  measurements,'' \emph{IEEE Transactions on Information Theory}, vol.~64,
  no.~1, pp. 467--484, 2018.

\bibitem{soltanolkotabi2019structured}
M.~Soltanolkotabi, ``Structured signal recovery from quadratic measurements:
  Breaking sample complexity barriers via nonconvex optimization,'' \emph{IEEE
  Transactions on Information Theory}, vol.~65, no.~4, pp. 2374--2400, 2019.

\bibitem{wacks2018doppler}
S.~Wacks \emph{et~al.}, ``Doppler-dpca and doppler-ati: novel sar modalities
  for imaging of moving targets using ultra-narrowband waveforms,'' \emph{IEEE
  Transactions on Computational Imaging}, vol.~4, no.~1, pp. 125--136, 2018.

\bibitem{Demanet13}
L.~Demanet and V.~Jugnon, ``Convex recovery from interferometric
  measurements,'' \emph{IEEE Transactions on Computational Imaging}, vol.~3,
  no.~2, pp. 282--295, 2017.

\bibitem{Demanet14}
\BIBentryALTinterwordspacing
L.~Demanet and P.~Hand, ``\BIBforeignlanguage{English}{Stable optimizationless
  recovery from phaseless linear measurements},''
  \emph{\BIBforeignlanguage{English}{Journal of Fourier Analysis and
  Applications}}, vol.~20, no.~1, pp. 199--221, 2014. [Online]. Available:
  \url{http://dx.doi.org/10.1007/s00041-013-9305-2}
\BIBentrySTDinterwordspacing

\bibitem{cai2010singular}
J.-F. Cai, E.~J. Cand{\`e}s, and Z.~Shen, ``A singular value thresholding
  algorithm for matrix completion,'' \emph{SIAM Journal on Optimization},
  vol.~20, no.~4, pp. 1956--1982, 2010.

\bibitem{recht2010guaranteed}
B.~Recht, M.~Fazel, and P.~A. Parrilo, ``Guaranteed minimum-rank solutions of
  linear matrix equations via nuclear norm minimization,'' \emph{SIAM review},
  vol.~52, no.~3, pp. 471--501, 2010.

\bibitem{Chai11}
\BIBentryALTinterwordspacing
A.~Chai, M.~Moscoso, and G.~Papanicolaou, ``Array imaging using intensity-only
  measurements,'' \emph{Inverse Problems}, vol.~27, no.~1, p. 015005, 2011.
  [Online]. Available: \url{http://stacks.iop.org/0266-5611/27/i=1/a=015005}
\BIBentrySTDinterwordspacing

\bibitem{son2019exact}
I.-Y. Son, B.~Yonel, and B.~Yazici, ``Exact recovery of extended targets using
  multistatic interferometric measurements via generalized wirtinger flow,''
  \emph{arXiv preprint arXiv:1905.10459}, 2019.

\bibitem{bandeira2014saving}
A.~S. Bandeira, J.~Cahill, D.~G. Mixon, and A.~A. Nelson, ``Saving phase:
  Injectivity and stability for phase retrieval,'' \emph{Applied and
  Computational Harmonic Analysis}, vol.~37, no.~1, pp. 106--125, 2014.

\bibitem{candes2015phase2}
E.~J. Candes, X.~Li, and M.~Soltanolkotabi, ``Phase retrieval from coded
  diffraction patterns,'' \emph{Applied and Computational Harmonic Analysis},
  vol.~39, no.~2, pp. 277--299, 2015.

\bibitem{donoho2006compressed}
D.~L. Donoho, ``Compressed sensing,'' \emph{IEEE Transactions on information
  theory}, vol.~52, no.~4, pp. 1289--1306, 2006.

\bibitem{candes2006stable}
E.~J. Candes, J.~K. Romberg, and T.~Tao, ``Stable signal recovery from
  incomplete and inaccurate measurements,'' \emph{Communications on Pure and
  Applied Mathematics: A Journal Issued by the Courant Institute of
  Mathematical Sciences}, vol.~59, no.~8, pp. 1207--1223, 2006.

\bibitem{recht2008necessary}
B.~Recht, W.~Xu, and B.~Hassibi, ``Necessary and sufficient conditions for
  success of the nuclear norm heuristic for rank minimization,'' in
  \emph{Decision and Control, 2008. CDC 2008. 47th IEEE Conference on}.\hskip
  1em plus 0.5em minus 0.4em\relax IEEE, 2008, pp. 3065--3070.

\bibitem{recht2011null}
------, ``Null space conditions and thresholds for rank minimization,''
  \emph{Mathematical programming}, vol. 127, no.~1, pp. 175--202, 2011.

\bibitem{oymak2011simplified}
S.~Oymak, K.~Mohan, M.~Fazel, and B.~Hassibi, ``A simplified approach to
  recovery conditions for low rank matrices,'' in \emph{Information Theory
  Proceedings (ISIT), 2011 IEEE International Symposium on}.\hskip 1em plus
  0.5em minus 0.4em\relax IEEE, 2011, pp. 2318--2322.

\bibitem{cai2013sharp}
T.~T. Cai, ``Sharp rip bound for sparse signal and low-rank matrix recovery,''
  \emph{Appl. Comput. Harmon. Anal}, vol.~35, pp. 74--93, 2013.

\bibitem{bhojanapalli2016global}
S.~Bhojanapalli, B.~Neyshabur, and N.~Srebro, ``Global optimality of local
  search for low rank matrix recovery,'' in \emph{Advances in Neural
  Information Processing Systems}, 2016, pp. 3873--3881.

\bibitem{zheng2015convergent}
Q.~Zheng and J.~Lafferty, ``A convergent gradient descent algorithm for rank
  minimization and semidefinite programming from random linear measurements,''
  in \emph{Advances in Neural Information Processing Systems}, 2015, pp.
  109--117.

\bibitem{tu2015low}
S.~Tu, R.~Boczar, M.~Simchowitz, M.~Soltanolkotabi, and B.~Recht, ``Low-rank
  solutions of linear matrix equations via procrustes flow,'' \emph{arXiv
  preprint arXiv:1507.03566}, 2015.

\bibitem{wang2016unified}
L.~Wang, X.~Zhang, and Q.~Gu, ``A unified computational and statistical
  framework for nonconvex low-rank matrix estimation,'' \emph{arXiv preprint
  arXiv:1610.05275}, 2016.

\bibitem{chi2018nonconvex}
Y.~Chi, Y.~M. Lu, and Y.~Chen, ``Nonconvex optimization meets low-rank matrix
  factorization: An overview,'' \emph{arXiv preprint arXiv:1809.09573}, 2018.

\bibitem{li2018rapid}
X.~Li, S.~Ling, T.~Strohmer, and K.~Wei, ``Rapid, robust, and reliable blind
  deconvolution via nonconvex optimization,'' \emph{Applied and Computational
  Harmonic Analysis}, 2018.

\bibitem{chen2017solving}
Y.~Chen and E.~J. Candes, ``Solving random quadratic systems of equations is
  nearly as easy as solving linear systems,'' \emph{Communications on Pure and
  Applied Mathematics}, vol.~70, pp. 0822--0883, 2017.

\bibitem{sun2018geometric}
J.~Sun, Q.~Qu, and J.~Wright, ``A geometric analysis of phase retrieval,''
  \emph{Foundations of Computational Mathematics}, vol.~18, no.~5, pp.
  1131--1198, 2018.

\bibitem{zhang2015restricted}
H.~Zhang and L.~Cheng, ``Restricted strong convexity and its applications to
  convergence analysis of gradient-type methods in convex optimization,''
  \emph{Optimization Letters}, vol.~9, no.~5, pp. 961--979, 2015.

\bibitem{yonel2018generalized}
B.~Yonel, I.-Y. Son, and B.~Yazici, ``Generalized wirtinger flow for passive
  polarimetric reconstruction of extended dipole targets,'' in \emph{2018 IEEE
  Radar Conference (RadarConf18)}.\hskip 1em plus 0.5em minus 0.4em\relax IEEE,
  2018, pp. 1353--1358.

\bibitem{foucart2013mathematical}
S.~Foucart and H.~Rauhut, \emph{A mathematical introduction to compressive
  sensing}.\hskip 1em plus 0.5em minus 0.4em\relax Birkh{\"a}user Basel, 2013,
  vol.~1, no.~3.

\bibitem{vershynin2010introduction}
R.~Vershynin, ``Introduction to the non-asymptotic analysis of random
  matrices,'' \emph{arXiv preprint arXiv:1011.3027}, 2010.

\end{thebibliography}

\appendix
\section{Derivations}\label{sec:AppDerivs}

\subsection{Derivation of $\nabla \mathcal{J}$}\label{sec:App0}

Recall the objective function $\mathcal{J}$ in \eqref{eq:GWF_opt}
\begin{equation} \label{eq:ObjFunct}
\mathcal{J}(\brho) = \frac{1}{2M} \sum_{m = 1}^M | \left(\mathbf{L}_i^m\right)^H \brho \brho^H \mathbf{L}_j^m - d^{ij}_m |^2.
\end{equation}
Letting $e^m = \left(\mathbf{L}_i^m \right)^H \brho \brho^H \mathbf{L}_j^m - d_{ij}^m$, we write 
\begin{equation}\label{eq:gradTerm}
\frac{\partial \mathcal{J}}{\partial \brho} = \frac{1}{2M} \sum_{ m = 1}^M \frac{\partial}{\partial \brho} (e^m \bar{e}^m) = \frac{1}{2M} \sum_{ m = 1}^M \frac{\partial e^m}{\partial \brho} \bar{e}^m +\frac{\partial \bar{e}^m}{\partial \brho} e^m.
\end{equation}
Having $\brho$ and $\brho^H$ independent by properties of Wirtinger derivatives, we compute the partial derivatives in \eqref{eq:gradTerm} as
\begin{equation}\label{eq:GaussUpdt}
\frac{\partial \mathcal{J}}{\partial \brho} = \frac{1}{2M} \sum_{ m = 1}^M \bar{e}^m (\brho^H \mathbf{L}_j^m \left(\mathbf{L}_i^m \right)^H) + e^m (\brho^H \mathbf{L}_i^m \left(\mathbf{L}_j^m \right)^H).
\end{equation}
Using the definition of complex gradient provided in \eqref{eq:compGrad}, we finally get
\begin{equation}
\nabla \mathcal{J}  = \frac{1}{2M} \sum_{ m = 1}^M \bar{e}^m \mathbf{L}_j^m \left(\mathbf{L}_i^m \right)^H \brho + e^m \mathbf{L}_i^m \left(\mathbf{L}_j^m \right)^H \brho.
\end{equation}

\subsection{Proof Of Lemma \ref{cor:Corr1}}\label{sec:App1}
Assuming the RIP on $\mathcal{F}$ over rank-1, PSD matrices with RIC-$\delta$, we write
\begin{equation}\label{eq:App1_1}
(1-\delta) \| {\brho} \brho^H \|_F^2 \leq \| \mathcal{F}({\brho}\brho^H ) \|^2 \leq (1+\delta) \| {\brho} \brho^H \|_F^2.
\end{equation}
Equivalently, from the definition of the Frobenius inner product and the adjoint operator $\mathcal{F}^H$, we re-express \eqref{eq:App1_1} as
\begin{eqnarray}\label{eq:App1_2}
(1-\delta) \langle {\brho} \brho^H, {\brho}\brho^H \rangle_F \leq \langle \mathcal{F}^H \mathcal{F} ({\brho}\brho^H ), {\brho}\brho^H \rangle_F & \leq & (1+\delta) \langle {\brho}\brho^H, {\brho}\brho^H \rangle_F, \\ 
| \langle (\mathcal{F}^H \mathcal{F} - \mathbf{I})({\brho}\brho^H ) , {\brho}\brho^H \rangle_F | &  \leq  & \delta \| {\brho}\brho^H \|_F^2,
\end{eqnarray}
hence for any $\brho \in \mathbb{C}^N$ we have 
\begin{equation}
\frac{| \langle (\mathcal{F}^H \mathcal{F} - \mathbf{I})(\brho \brho^H ) , \brho \brho^H \rangle_F | }{ \| \brho{\brho}^H \|_F^2} \leq \delta.
\end{equation}
Now consider the definition of the spectral norm, with
\begin{equation}\label{eq:App1_3}
\| \bdelta(\brho \brho^H) \| = \underset{\| \v \| = 1}{\text{max}} | \v^H \bdelta(\brho \brho^H) \v | = \underset{\| \v \| = 1}{\text{max}} | \langle \bdelta(\brho \brho^H ) , \v \v^H \rangle_F |. 
\end{equation}
From the Hermitian property of $\bdelta$, we have
\begin{equation}\label{eq:App1_3_1}
\| \bdelta(\brho \brho^H) \|  = \underset{\| \v \| = 1}{\text{max}} | \langle \sqrt{\bdelta}(\brho \brho^H ) , \sqrt{\bdelta}^H (\v \v^H) \rangle_F | \leq \underset{\| \v \| = 1}{\text{max}}  \|  \sqrt{\bdelta}(\brho \brho^H )  \|_F \| \sqrt{\bdelta} (\v \v^H) \|_F,
\end{equation}
where $\sqrt{\bdelta}: \mathbb{C}^{N \times N} \rightarrow \mathbb{C}^{N \times N}$ with $\sqrt{\bdelta} \sqrt{\bdelta} = \bdelta$, and the inequality follows from matrix Cauchy-Schwartz property, and the fact that Frobenius norm is unaltered by conjugation. 
Observe that using the Hermitian property of $\bdelta$ via Definition 6.1 in \cite{foucart2013mathematical}, for any $\brho \in \mathbb{C}^N$, we have $ \| \sqrt{\bdelta}(\brho \brho^H) \|_F^2/ \| \brho \brho \|_F^2 \leq \delta$, hence
\begin{equation}
\|  \sqrt{\bdelta}(\brho \brho^H )  \|_F \underset{\| \v \| = 1}{\text{max}} \| \sqrt{\bdelta} (\v \v^H) \|_F \leq  \sqrt{\delta} \|  \sqrt{\bdelta}(\brho \brho^H )  \|_F \leq \delta \| \brho \brho^H \|_F. 
\end{equation}
Using the definition of the spectral norm, we simply obtain
\begin{equation}\label{eq:App1_4}
\frac{\| \bdelta( \brho \brho^H ) \|}{\| \brho \brho^H \|} \leq \delta,
\end{equation}
for any $\brho \in \mathbb{C}^N$ if \eqref{eq:App1_2} holds. 
Furthermore since $\delta$ is the smallest constant such that \eqref{eq:App1_2} is satisfied, let
\begin{equation}
\delta := \underset{\mathbf{\brho \in \mathbb{C}^N \setminus \{0 \}}}{\text{max}} \frac{| \langle \bdelta(\brho \brho^H ) , \brho \brho^H \rangle_F |}{ \| \brho \brho^H \|_F^2} =  \underset{\mathbf{\brho \in \mathbb{C}^N \setminus \{0 \}}}{\text{max}} \frac{\| \sqrt{\bdelta}(\brho \brho^H) \|_F^2}{\| \brho \brho^H \|_F^2},
\end{equation} 
Revisiting \eqref{eq:App1_3_1}, since from \eqref{eq:App1_4} we know for any $\brho$, $\| \bdelta (\brho \brho^H) \| \leq \delta \| \brho \brho^H \|_F$, the maximal value of $\delta$ is reached at $\v = \brho/\| \brho \|$, as
\begin{equation}
\begin{split}
\underset{\mathbf{\brho \in \mathbb{C}^N \setminus \{0 \}}}{\text{max}} \| \bdelta( \brho \brho^H ) \| 
& \leq 
\underset{\mathbf{\brho \in \mathbb{C}^N \setminus \{0 \}}}{\text{max}}
 \ \underset{\| \v \| = 1}{\text{max}}  \|  \sqrt{\bdelta}(\brho \brho^H )  \|_F \| \sqrt{\bdelta} (\v \v^H) \|_F \\
& = \underset{\mathbf{\brho \in \mathbb{C}^N \setminus \{0 \}}}{\text{max}} \frac{\|  \sqrt{\bdelta}(\brho \brho^H )\|_F^2}{\| \brho \brho^H \|_F} = \delta \| \brho \brho^H \|_F,
\end{split}
\end{equation}
and by definition
\begin{equation}
\begin{split}
\underset{\mathbf{\brho \in \mathbb{C}^N \setminus \{0 \}}}{\text{max}} \frac{\| \bdelta( \brho \brho^H ) \|}{\| \brho \brho^H \|_F } & \geq \underset{\mathbf{\brho \in \mathbb{C}^N \setminus \{0 \}}}{\text{max}}  \underset{\| \v \| = 1}{\text{max}} \frac{| \langle \bdelta(\brho \brho^H ) , \v \v^H \rangle_F |}{\| \brho \brho^H \|_F} \\
& \geq \underset{\mathbf{\brho \in \mathbb{C}^N \setminus \{0 \}}}{\text{max}} \frac{| \langle \bdelta(\brho \brho^H ) , \brho \brho^H \rangle_F |}{ \| \brho \brho^H \|_F^2} = \delta.
\end{split}
\end{equation}
Hence the proof is complete. 

%


\section{Lemmas for Theorem \ref{thm:Theorem1}}\label{sec:AppThm1}
\subsection{Proof of Lemma \ref{lem:Lemma2}}\label{sec:App2}
Since $\hat{\brho}_t  = e^{\mathrm{i} \Phi(\v_0)} \brho_t$, we have
\begin{equation}
\text{dist}^2 (\v_0, \brho_t) = \| \v_0 - \hat{\brho}_t \|^2 = \| \v_0 \|^2 + \| \hat{\brho}_t \|^2 - 2 \text{Re} \langle \hat{\brho}_t, \v_0 \rangle,
\end{equation}
Knowing that $\Phi(\v_0)$ achieves $\text{Re}\langle \hat{\brho}_t, \v_0 \rangle = | \langle {\brho}_t, \v_0 \rangle |$, we have
\begin{equation}
\text{Re}\langle \hat{\brho}_t, \v_0 \rangle  =  \langle \hat{\brho}_t, \v_0 \rangle = | \langle e^{ \mathrm{i} \Phi(\v_0)} {\brho}_t, \v_0 \rangle |. 
\end{equation}
Since $\hat{\brho}_t$ and $\v_0$ are unit norm, the geometric angle between them can be written as 
\begin{equation}
\cos(\theta)  =  \langle \hat{\brho}_t, \v_0 \rangle.
\end{equation}
Invoking the representation theorem in Hilbert spaces, there exists a unit vector $\v_0^{\perp}$ that lies in the plane whose normal is $\v_0$ such that
\begin{eqnarray}
\langle \hat{\brho}_t, \v_0 \rangle = \cos(\theta) \langle \v_0, \v_0 \rangle + \sin(\theta) \langle  \v_0^{\perp} ,\v_0 \rangle, \\
\langle \hat{\brho}_t - (\cos(\theta)\v_0 +  \sin(\theta) \v_0^{\perp}) , \v_0 \rangle = 0. 
\end{eqnarray}
The inner product is zero only when: 1) $\hat{\brho}_t = \cos(\theta)\v_0 +  \sin(\theta) \v_0^{\perp}$, 2) $\hat{\brho}_t - (\cos(\theta)\v_0 +  \sin(\theta) \v_0^{\perp}) $ is perpendicular to $\v_0$. The latter case occurs iff $\hat{\brho}_t - \cos(\theta)\v_0 = 0$, which is true only for $\theta = 0$, which indicates the identical solution as the former hence we have the unique representation
\begin{equation}
\hat{\brho}_t = \cos(\theta)\v_0 +  \sin(\theta) \v_0^{\perp}.
\end{equation}
Using the same representation, it is straightforward to see that the unit length element $\hat{\brho}_t^{\perp} =  -\sin(\theta)\v_0 +  \cos(\theta) \v_0^{\perp}$ satisfies
\begin{equation}
\langle  \hat{\brho}_t^{\perp}, \hat{\brho}_t,  \rangle = \langle  -\sin(\theta)\v_0 +  \cos(\theta) \v_0^{\perp},  \cos(\theta)\v_0 +  \sin(\theta) \v_0^{\perp}  \rangle = 0
\end{equation}
and hence lies, in the plane whose normal is $\hat{\brho}_t$. 

\subsection{Proof of Lemma \ref{lem:Lemma3}}\label{sec:App3}
Recalling the representation of the spectral estimate in the lifted problem we have
\begin{equation}
\hat{\mathbf{X}} = \mathcal{P}_{S} ( \mathcal{F}^H \mathcal{F} (\brho_t \brho_t^H ) ) 
\end{equation}
where $\mathcal{P}_{S}$ is the projection onto the set of symmetric matrices. 
$\hat{\mathbf{X}}_{PSD}$ is similarly written as
\begin{equation}
\hat{\mathbf{X}}_{PSD} = \mathcal{P}_{PSD} (\hat{\mathbf{X}}) = \mathcal{P}_{PSD}( \mathcal{F}^H \mathcal{F} (\brho_t \brho_t^H ) ) 
\end{equation}
since the $PSD$-cone is a subset of the set of symmetric matrices $S$, and $\mathcal{P}_{PSD}$ projection consists of projection onto $S$ by $\mathcal{P}_{S}$, followed by suppression of negative eigenvalues.  
From Lemma \ref{cor:Corr1}, denoting $\tilde{\brho}_t = \brho_t \brho_t^H$ for the solution $\brho_t$ obeying $\| \brho_t \| = 1$ we have
\begin{eqnarray}
\| \hat{\mathbf{X}} - \tilde{\brho}_t \| & \leq & \| \frac{1}{2} \mathcal{F}^H \mathcal{F}(\tilde{\brho}_t) + \frac{1}{2} (\mathcal{F}^H \mathcal{F}(\tilde{\brho}_t))^H - \tilde{\brho}_t\| \\
& \leq & \frac{1}{2} \| \mathcal{F}^H \mathcal{F}(\tilde{\brho}_t) -\tilde{\brho}_t \| +  \frac{1}{2} \| (\mathcal{F}^H \mathcal{F}(\tilde{\brho}_t))^H -\tilde{\brho}_t \|  \\ \nonumber
& \leq & \frac{1}{2} \delta + \frac{1}{2} \| ( (\tilde{\brho}_t) + \bdelta(\tilde{\brho_t}) )^H -\tilde{\brho}_t \|. \nonumber
\end{eqnarray}
And since $\tilde{\brho}_t$ is Hermitian symmetric and that the spectral norm is unaffected by adjoint operation, we have
\begin{equation}
\frac{1}{2} \| ( (\tilde{\brho}_t) + \bdelta(\tilde{\brho_t}) )^H -\tilde{\brho}_t \| =  \frac{1}{2} \| (\bdelta(\tilde{\brho_t}) )^H  \|  =  \frac{1}{2} \| \bdelta(\tilde{\brho_t})   \| \leq \frac{1}{2} \delta.
\end{equation}
Hence the spectral norm of the lifted error is upper bounded as follows: 
\begin{equation}
\| \hat{\mathbf{X}} - \tilde{\brho}_t \| \leq \delta.
\end{equation}
For the maximal eigenvalue $\lambda_0$ of $\hat{\mathbf{X}}$, we can write
\begin{equation}
\lambda_0 \geq \brho_t^H \hat{\mathbf{X}} \brho_t = \brho_t^H (\hat{\mathbf{X}} - \tilde{\brho}_t) \brho_t  +  1 \geq  1 - \delta.
\end{equation}
On the other hand, we have
\begin{eqnarray}
| \v_0^H (\hat{\mathbf{X}} - \tilde{\brho}_t) \v_0 | & \leq & \| \hat{\mathbf{X}} - \tilde{\brho}_t \| \leq \delta \ \rightarrow \ | \lambda_0 - | \v_0^H \brho_t |^2 | \leq  \delta, \\
\lambda_0 \leq \delta + | \v_0^H \brho_t |^2 | & \leq & 1 + \delta. \nonumber
\end{eqnarray}
Hence we obtain
\begin{equation}
1 - \delta \leq \lambda_0 \leq 1 + \delta.
\end{equation}
Since the positive semi-definite estimate $\hat{\mathbf{X}}_{PSD}$ only differs from $\hat{\mathbf{X}}$ by the suppression of negative eigenvalues, and since spectral initialization only preserves the leading eigenvalue-eigenvector pair $\lambda_0$, $\v_0$ and $\lambda_0 \geq 1-\delta$ where $\delta > 0$, we have the identical $\brho_0 = \sqrt{\lambda_0} \v_0$. 

\subsection{Proof of Lemma \ref{lem:Lemma4}}\label{sec:App4}
We consider the case where the spectral estimate \eqref{eq:spectralInitGWF} is projected onto the positive semi-definite cone as in LRMR. 
In this case, the estimate obtained in spectral initialization is projected onto the positive semi-definite cone to obtain $\hat{\mathbf{X}}_{PSD}$, which yields the identical initial point $\brho_0$ per Lemma \ref{lem:Lemma3}. 
Since its a symmetric, PSD matrix, we can decompose $\hat{\mathbf{X}}_{PSD}$ such that
$$
\hat{\mathbf{X}}_{PSD} = \mathbf{S}_0 \mathbf{S}_0,
$$
where $\mathbf{S}_0$ is a positive semi-definite matrix with its eigenvalues as $\sqrt{\lambda_{\hat{\mathbf{X}}_{PSD}}}$. 
Since
\begin{equation}
\v_0 := \underset{\| \v \| = 1}{\text{argmax}} \ \v^H \hat{\mathbf{X}}_{PSD} \v, 
\end{equation}
where $\v^H \hat{\mathbf{X}}_{PSD} \v = \v^H \mathbf{S}_0 \mathbf{S}_0 \v = \| \mathbf{S}_0 \v \|^2$. 
Then using the definitions of $\hat{\brho}_t$ and $\hat{\brho}_t^{\perp}$ from Lemma \ref{lem:Lemma2}, we have
\begin{eqnarray}
 \mathbf{S}_0 \hat{\brho}_t &=& \cos(\theta) \mathbf{S}_0 \v_0  + \sin(\theta)  \mathbf{S}_0 \v_0^{\perp},\\
 \mathbf{S}_0 \hat{\brho}_t^{\perp} &=& - \sin(\theta) \mathbf{S}_0 \v_0 + \cos(\theta)  \mathbf{S}_0, \v_0^{\perp}
\end{eqnarray}
and from Pythogorian theorem, since we have orthogonal components, we get
\begin{eqnarray}
 \| \mathbf{S}_0 \hat{\brho}_t \|^2 &=& \cos^2(\theta) \| \mathbf{S}_0 \v_0 \|^2  + \sin^2(\theta) \| \mathbf{S}_0 \v_0^{\perp} \|^2, \\
\| \mathbf{S}_0 \hat{\brho}_t^{\perp} \|^2 &=& \sin^2(\theta) \| \mathbf{S}_0 \v_0 \|^2 + \cos^2(\theta) \| \mathbf{S}_0 \v_0^{\perp} \|^2.
\end{eqnarray}
Consider the expression $f = \sin^2(\theta)  \| \mathbf{S}_0 \hat{\brho}_t \|^2  - \| \mathbf{S}_0 \hat{\brho}_t^{\perp} \|^2$. Following the algebra in \cite{wang2018solving}, we have
\begin{eqnarray}
 f & = & \sin^2(\theta) \left( \cos^2(\theta) \| \mathbf{S}_0 \v_0 \|^2  + \sin^2(\theta) \| \mathbf{S}_0 \v_0^{\perp} \|^2 \right) - \left(\sin^2(\theta) \| \mathbf{S}_0 \v_0 \|^2 + \cos^2(\theta) \| \mathbf{S}_0 \v_0^{\perp} \|^2 \right) \nonumber \\
& = & \sin^2(\theta) \left( \cos^2(\theta) \| \mathbf{S}_0 \v_0 \|^2  - \| \mathbf{S}_0 \v_0 \|^2 + \sin^2(\theta) \| \mathbf{S}_0 \v_0^{\perp} \|^2 \right) - \cos^2(\theta) \| \mathbf{S}_0 \v_0^{\perp} \|^2, \nonumber
\end{eqnarray}
which finally yields
\begin{equation}\label{eq:Lem4_fin1}
f = \sin^4(\theta) \left(  \| \mathbf{S}_0 \v_0^{\perp} \|^2 - \| \mathbf{S}_0 \v_0 \|^2  \right) - \cos^2(\theta) \| \mathbf{S}_0 \v_0^{\perp} \|^2. 
\end{equation}
Since $\v_0$ is the unit vector that maximizes $\| \mathbf{S}_0 \v \|^2$, with the fact that $\v_0^{\perp}$ is also a unit vector, this expression is always non-positive, hence
$$
 \sin^2(\theta)  \| \mathbf{S}_0 \hat{\brho}_t \|^2  - \| \mathbf{S}_0 \hat{\brho}_t^{\perp} \|^2 \leq 0,
$$
\begin{equation}
 \sin^2(\theta)  \leq \frac{\| \mathbf{S}_0 \hat{\brho}_t^{\perp} \|^2 }{\| \mathbf{S}_0 \hat{\brho}_t \|^2 }.
\end{equation}
Equivalently, expressing the upper bound with the spectral estimate, we obtain
\begin{equation}\label{eq:uppBound}
\sin^2(\theta)  \leq \frac{(\hat{\brho}_t^{\perp})^H \hat{\mathbf{X}}_{PSD} \hat{\brho}_t^{\perp}}{ \hat{\brho}_t^H \hat{\mathbf{X}}_{PSD} \hat{\brho}_t}.
\end{equation}

Now we consider the spectral estimate $\hat{\mathbf{X}}_{PSD}$. We know that $\hat{\mathbf{X}}_{PSD}$ is obtained by projecting the intermediate estimate $\mathbf{Y} = \mathcal{F}^H \mathcal{F} (\tilde{\brho}_t)$ onto the feasible set of PSD matrices as defined in Lemma \ref{lem:Lemma3}. 
From Lemma \ref{cor:Corr1}, we know that $\mathcal{F}^H \mathcal{F}$ is approximately an identity on the domain of rank-1, PSD matrices, hence we can write the perturbation model
\begin{equation}
\hat{\mathbf{X}}_{PSD} = \mathcal{P}_{PSD} \left( \tilde{\brho}_t + \bdelta [\tilde{\brho}_t] \right)
\end{equation}
Since projection onto the PSD cone is non-expansive under the spectral norm, we have the following
\begin{equation}
\| \hat{\mathbf{X}}_{PSD} - \tilde{\brho}_t \| \leq \| \mathbf{Y} - \tilde{\brho}_t \|.
\end{equation}
Setting $\hat{\mathbf{X}}_{PSD} = \tilde{\brho}_t  + \tilde{\e}$, the upper bound in (\ref{eq:uppBound}) can be written as
\begin{equation}
\sin^2(\theta)  \leq \frac{(\hat{\brho}_t^{\perp})^H \tilde{\brho}_t \hat{\brho}_t^{\perp} + (\hat{\brho}_t^{\perp})^H \tilde{\e} \hat{\brho}_t^{\perp} }{ \hat{\brho}_t^H \tilde{\brho}_t \hat{\brho}_t + \hat{\brho}_t^H \tilde{\e}\hat{\brho}_t }.
\end{equation}
Since we have that $\tilde{\brho}_t = \brho_t \brho_t^H$, the bound reduces to
\begin{equation}
\sin^2(\theta)  \leq  \frac{(\hat{\brho}_t^{\perp})^H \tilde{\e} \hat{\brho}_t^{\perp} }{ 1 + \hat{\brho}_t^H \tilde{\e} \hat{\brho}_t }
\end{equation}
as $\hat{\brho}_t^{\perp}$ is orthogonal to $\hat{\brho}_t$ and the global phase component in $\hat{\brho}_t = e^{j \Phi(v_0)} \brho_t$ vanishes in the quadratic form. Moreover, we know that $(\hat{\brho}_t^{\perp})^H \tilde{\e} \hat{\brho}_t^{\perp}$ is non-negative from positive semi-definitivity of $\hat{\mathbf{X}}_{PSD}$.  
Hence we further upper bound the numerator by the spectral norm of $\tilde{\e}$, and follow with the series of bounds
\begin{equation}
\sin^2(\theta)  \leq  \frac{ \| \tilde{\e} \|}{ 1 + \hat{\brho}_t^H \tilde{\e} \hat{\brho}_t } \leq \frac{ \| \bdelta(\tilde{\brho}_t) \|}{ 1 + \hat{\brho}_t^H \tilde{\e} \hat{\brho}_t } \leq \frac{ \delta \| \tilde{\brho}_t \|_F}{ 1 + \hat{\brho}_t^H \tilde{\e} \hat{\brho}_t } 
\end{equation}
where the last two inequalities follow from the non-expansiveness property of the projection operator onto the PSD cone, and the definition of $\bdelta$ from Lemma \ref{cor:Corr1}, such that the spectral norm of $\bdelta$ is upper bounded by the RIC-$\delta$ of $\mathcal{F}$.  
For the term in the denominator, since we established that $\| \tilde{\e} \|_2 \leq \delta $, we have that  $ - \delta \leq \hat{\brho}_t^H \tilde{\e} \hat{\brho}_t \leq \delta$ to finally obtain
\begin{equation}
\sin^2(\theta) \leq \frac{\delta}{ 1 + \hat{\brho}_t^H \tilde{\e} \hat{\brho}_t } \leq \frac{\delta}{ 1 - \delta }.
\end{equation}

\subsection{Proof of Lemma \ref{lem:Lemma5}}\label{sec:App5}
\subsubsection{The Upper Bound}
Noting that $\hat{\brho}_t$ is the closest solution in $\mathit{P}$ to an estimate $\brho$, we define $\brho_e = \brho - \hat{\brho}_t$. Since $\brho \in E(\epsilon)$, $\text{dist}(\brho, \brho_t ) = \| \brho_e \| \leq \epsilon$, and from reverse triangle inequality we have
$$
| \| \brho \| - \| \hat{\brho}_t \| | \leq \| \brho - \hat{\brho}_t \| \leq \epsilon
$$
Since $\| \hat{\brho}_t \| = \|{\brho}_t \| = 1$, we have that $ 1-\epsilon \leq \| {\brho} \| \leq 1+\epsilon$.
Having $\hat{\brho}_t \hat{\brho}_t^H = {\brho}_t {\brho}_t^H$, we let $\tilde{\mathbf{e}} = \brho \brho^H - \hat{\brho}_t \hat{\brho}_t^H$ denote the error in the lifted problem. Then, writing the {lifted} error $\tilde
{\brho}_e = \brho_e \brho_e^H$ as
\begin{eqnarray}
\brho_e \brho_e^H & = & (\brho - \hat{\brho}_t)(\brho - \hat{\brho}_t)^H = \brho \brho^H - \brho \hat{\brho}_t^H - \hat{\brho}_t \brho^H + \hat{\brho}_t \hat{\brho}_t^H + (2\hat{\brho}_t \hat{\brho}_t^H - 2\hat{\brho}_t \hat{\brho}_t^H) \nonumber \\
& = &\tilde{\e} - (\brho - \hat{\brho}_t) \hat{\brho}_t^H  - \hat{\brho}_t (\brho^H - \hat{\brho}_t^H),
\end{eqnarray}
finally yields the expression
\begin{equation}\label{eq:lifted_Error_fin_app5}
\tilde{\e} = \brho_e \brho_e^H + \brho_e \hat{\brho}_t^H + \hat{\brho}_t \brho_e^H
\end{equation}
for the error in the lifted domain. 
We can then write the upper bound for $\| \tilde{\e} \|_F$ as 
\begin{equation}
\| \tilde{\e} \|_F \leq \| \brho_e \brho_e^H  \|_F + \| \brho_e \hat{\brho}_t^H \|_F + \| \hat{\brho}_t \brho_e^H \|_F.
\end{equation}
Since all the arguments of the Frobenius norm in the right-hand-side have rank-1, we have $\| \cdot \|_2 = \| \cdot \|_F$. Knowing that $\| \brho_t \| = 1$ we get
\begin{equation}
\| \tilde{\e} \|_F \leq 2 \|  \hat{\brho}_t \| \| \brho_e \| + \| \brho_e \|^2 \leq (2 + \epsilon) \| \brho_e \|.
\end{equation}
Hence we obtain the upper bound
\begin{equation}
\| \brho \brho^H - {\brho}_t {\brho}_t^H \|_F \leq (2+\epsilon) \text{dist} (\brho , {\brho}_t ).
\end{equation}

\subsubsection{The Lower Bound}

Expanding the error in the lifted domain, we get
\begin{equation}\label{eq:lifterr_app5_1}
\| \tilde{\e} \|_F^2 = \| \brho \brho^H \|_F^2 + \| \brho_t \brho_t^H \|_F^2 - 2 \text{Re} \langle \brho \brho^H, \brho_t \brho_t^H \rangle_F.
\end{equation}
Since we have the rank-1 lifted signals, the Frobenius norms and the inner product reduce to
\begin{eqnarray}
 \| \tilde{\e} \|_F^2 & = & \| \brho \|^4 + \| \brho_t \|^4 - 2 | \langle \brho , \brho_t \rangle |^2 = (\| \brho \|^4 -  | \langle \brho , \brho_t \rangle |^2 ) +  (\| \brho_t \|^4 - | \langle \brho , \brho_t \rangle |^2 ) \nonumber \\
& = & (\| \brho \|^2 +  | \langle \brho , \brho_t \rangle |)(\| \brho \|^2 -  | \langle \brho , \brho_t \rangle |) + (\| \brho_t \|^2 + | \langle \brho , \brho_t \rangle |)(\| \brho_t \|^2 - | \langle \brho , \brho_t \rangle |). \label{eq:lifterr_ap5}
\end{eqnarray}
Having $\text{dist}^2(\brho, \brho_t) = \| \brho \|^2 + \| \brho_t \|^2 - 2 | \langle \brho, \brho_t \rangle | = \| \brho - \hat{\brho}_t \| \geq 0$, we can lower bound \eqref{eq:lifterr_app5_1} using \eqref{eq:lifterr_ap5} as
\begin{equation}\label{eq:lifterr_lbound}
\| \tilde{\e} \|_F^2 \geq \text{min} \left( (\| \brho \|^2 +  | \langle \brho , \brho_t \rangle |), (\| \brho_t \|^2 + | \langle \brho , \brho_t \rangle |)\right) \left( \| \brho \|^2 + \| \brho_t \|^2 - 2 | \langle \brho, \brho_t \rangle | \right).
\end{equation}

Knowing that $\text{dist}^2(\brho, \brho_t) \leq \epsilon^2$ and the result from the reverse triangle inequality on $\| \brho \|$, the terms within the minimization are further lower bounded using
\begin{eqnarray}
 2 | \langle \brho, \brho_t \rangle |  & \geq & \| \brho \|^2 + \| \brho_t \|^2 - \epsilon^2 \nonumber \\
 | \langle \brho, \brho_t \rangle |  & \geq & (1 - \epsilon). 
\end{eqnarray}
We then get the bound on the scalar multiplying $\text{dist}^2(\brho, \brho_t)$ as
\begin{equation}
\text{min} \left( (\| \brho \|^2 +  | \langle \brho , \brho_t \rangle |), (\| \brho_t \|^2 + | \langle \brho , \brho_t \rangle |)\right) \geq  (1 - \epsilon)^2 + (1-\epsilon),
\end{equation}
which yields the final lower bound as
\begin{equation}
 \| \brho \brho^H - {\brho}_t {\brho}_t^H \|_F \geq \sqrt{(1-\epsilon)(2-\epsilon)} \ \text{dist} (\brho , {\brho}_t ).
\end{equation}

\subsection{Proof of Lemma \ref{lem:Lemma6}}\label{sec:App6}
From the adjoint property of the inner product we have
$$ \| \mathcal{F} ( \brho \brho^H - \brho_t \brho_t^H ) \|^2 = \langle \brho \brho^H - \brho_t \brho_t^H , \mathcal{F}^H \mathcal{F} ( \brho \brho^H - \brho_t \brho_t^H ) \rangle_F.$$ 
Splitting the linear operator $\mathcal{F}^H \mathcal{F}$ over the rank-1 inputs, we can use the perturbation model from Lemma  \ref{cor:Corr1} such that
\begin{eqnarray}
& = & \langle \brho \brho^H - \brho_t \brho_t^H , \mathcal{F}^H \mathcal{F} (\brho \brho^H ) \rangle_F -  \langle \brho \brho^H - \brho_t \brho_t^H , \mathcal{F}^H \mathcal{F} (\brho_t \brho_t^H ) \rangle_F \nonumber \\
& = & \| \brho \brho^H - \brho_t \brho_t^H \|_F^2 + \langle \brho \brho^H - \brho_t \brho_t^H, \bdelta(\brho \brho^H) - \bdelta( \brho_t \brho_t^H) \rangle_F. \label{eq:App6_eq1}
\end{eqnarray}

Using the representation of $\tilde{\mathbf{e}} = \brho \brho^H - \brho_t \brho_t^H$ in \eqref{eq:lifted_Error_fin_app5}, and the linearity of $\bdelta$ we get 
\begin{equation}\label{eq:App6_eq2_v1}
\bdelta(\brho \brho^H) - \bdelta(\brho_t \brho_t^H) = \bdelta(\brho_e \brho_e^H) + \bdelta(\brho_e \hat{\brho}_t^H + \hat{\brho}_t \brho_e^H). 
\end{equation}
Notably, the domain of $\bdelta$ is by definition the set of rank-1, PSD matrices. 
Having $\brho_e \hat{\brho}_t^H + \hat{\brho}_t \brho_e^H$, a symmetric matrix of at most rank-2 in the argument of $\bdelta$ on the right-hand-side, we can represent \eqref{eq:App6_eq2_v1} by an eigenvalue decomposition and use the linearity of $\bdelta$ to obtain
\begin{equation}\label{eq:App6_eq2}
\bdelta(\brho \brho^H) - \bdelta(\brho_t \brho_t^H)  = \bdelta(\brho_e \brho_e^H) + \lambda_1 \bdelta(\v_1 \v_1^H) + \lambda_2 \bdelta(\v_2 \v_2^H),
\end{equation}
where $\brho_e \hat{\brho}_t^H + \hat{\brho}_t \brho_e^H = \sum_{i =1}^2 \lambda_i \v_i \v_i^H$, with $\| \v_i \| = 1$. 
Plugging \eqref{eq:App6_eq2} into \eqref{eq:App6_eq1} and applying the triangle inequality, we have
\begin{equation}\label{eq:App6_eq3}
\begin{split}
 | \langle \tilde{\mathbf{e}}, \bdelta(\brho \brho^H) - \bdelta( \brho_t \brho_t^H) \rangle_F |  \leq \ & | \langle \tilde{\mathbf{e}}, \bdelta(\brho_e \brho_e^H) \rangle_F | \\
& + |\lambda_1| | \langle \tilde{\mathbf{e}}, \bdelta(\v_1 \v_1^H) \rangle_F | + |\lambda_2| | \langle \tilde{\mathbf{e}}, \bdelta(\v_2 \v_2^H) \rangle_F |. 
\end{split}
\end{equation}
Furthermore, knowing that $\tilde{\mathbf{e}}$ is symmetric and at most rank-2, let $\tilde{\mathbf{e}} = \sum_{i = 1}^2 \sigma_i \mathbf{u}_i \mathbf{u}_i^H$, where $\| \mathbf{u}_i \| = 1$. Then, using the triangle inequality on the right-hand-side terms in \eqref{eq:App6_eq3}, and the outcome of Lemma \ref{cor:Corr1} to RIP over rank-1, PSD matrices with RIC-$\delta_1$, we have
\begin{equation}\label{eq:App6_eq3v2}
 | \langle \tilde{\mathbf{e}}, \bdelta(\brho \brho^H) - \bdelta( \brho_t \brho_t^H) \rangle_F |  \leq  (\sum_{i = 1}^2 | \sigma_i | ) \delta_1 (\| \brho_e \brho_e^H \|_F +  |\lambda_1| + |\lambda_2| ). 
\end{equation}
Furthermore, since $\tilde{\mathbf{e}}$ is rank-2 by definition, and $\sum_{i = 1}^2 | \sigma_i | = \| \tilde{\mathbf{e}} \|_{*} \leq \sqrt{2} \| \tilde{\mathbf{e}} \|_F$, we obtain
\begin{eqnarray}
 | \langle \tilde{\mathbf{e}}, \bdelta(\brho \brho^H) - \bdelta( \brho_t \brho_t^H) \rangle_F |  & \leq &\delta_1 \sqrt{2} \| \brho \brho^H - \brho_t \brho_t^H \|_F \left(\| \brho_e \brho_e^H \|_F +  |\lambda_1| + |\lambda_2| \right) \nonumber \\
& \leq &  \delta_1 \sqrt{2} \| \brho \brho^H - \brho_t \brho_t^H \|_F \left(\| \brho_e \brho_e^H \|_F +  \|  \brho_e \hat{\brho}_t^H \|_* + \| \hat{\brho}_t \brho_e^H  \|_* \right), \label{eq:App6_eq4}
\end{eqnarray}
where again we use the fact that $\| \brho_e \hat{\brho}_t^H   + \hat{\brho}_t \brho_e^H \|_* =  | \lambda_1 | + | \lambda_2 |$ by definition, from on which the triangle inequality follows.  
Invoking the rank-1 property on the terms inside the parenthesis in \eqref{eq:App6_eq4}, the Frobenius norm and the nuclear norms can be computed by the spectral norm, since $\| \cdot \| = \| \cdot \|_F = \| \cdot \|_*$ for a rank-1 argument. 
Then, having $ \| \brho_t \| = 1$, for the right-hand-side we obtain
\begin{eqnarray}
| \langle \brho \brho^H - \brho_t \brho_t^H, \bdelta(\brho \brho^H) - \bdelta( \brho_t \brho_t^H) \rangle_F | &\leq & \delta_1 \sqrt{2} \| \brho \brho^H - \brho_t \brho_t^H \|_F \left(2 \| \brho_t \| \| \brho_e \| + \| \brho_e \|^2 \right) \nonumber \\
& \leq & \delta_1 \sqrt{2} (2 + \epsilon)  \| \brho \brho^H - \brho_t \brho_t^H \|_F \ \text{dist}( \brho, \brho_t ). \label{eq:App6_eq5}
\end{eqnarray}
Using the bound 
\begin{equation}
\sqrt{(1-\epsilon)(2-\epsilon)} \text{dist}(\brho, \brho_t) \leq \| \brho \brho^H - \brho_t \brho_t^H \|_F,
\end{equation}
from Lemma \ref{lem:Lemma5}, we finally obtain the upper bound on the perturbation on $ \| \brho \brho^H - \brho_t \brho_t^H \|_F^2 $ as
\begin{equation}
| \langle \brho \brho^H - \brho_t \brho_t^H, \bdelta(\brho \brho^H) - \bdelta( \brho_t \brho_t^H) \rangle_F |  \leq \delta_1 \frac{(2+\epsilon)\sqrt{2}}{\sqrt{(1-\epsilon)(2-\epsilon)}} \| \brho \brho^H - \brho_t \brho_t^H \|_F^2.
\end{equation}
Hence, setting $\delta_2 = \delta_1 \frac{(2+\epsilon)\sqrt{2}}{\sqrt{(1-\epsilon)(2-\epsilon)}}$, we have the local RIP-2 condition satisfied with RIC-$\delta_2$.

\subsection{Proof of Lemma \ref{lem:Lemma7}}\label{sec:App7}
Having $\tilde{\brho}_t = \brho_t \brho_t^H = \hat{\brho}_t \hat{\brho}_t^H$, we re-write the gradient term in \eqref{eq:updt_term} as
$$
\nabla \mathcal{J}(\brho) = \| \brho \|^2 \brho - (\hat{\brho}_t^H \brho) \hat{\brho_t}  +  \mathcal{P}_{S}(\bdelta(\tilde{\e})) \brho
$$
where $\tilde{\e} = \brho \brho^H - \brho_t \brho_t^H$. 
Then, simply from the triangle inequality and the definition of projection operator on the set of symmetric matrices, we get
\begin{eqnarray}
\| \nabla \mathcal{J}(\brho) \| &\leq & \| \| \brho \|^2 \brho - (\hat{\brho_t}^H \brho) \hat{\brho}_t \| + \frac{1}{2} \| \bdelta(\tilde{\mathbf{e}}) \|_2 \| \brho \| + \frac{1}{2} \| \bdelta(\tilde{\mathbf{e}})^H \|_2 \| \brho \| \nonumber \\
\| \nabla \mathcal{J}(\brho) \| & \leq & \| \| \brho \|^2 \brho - (\hat{\brho_t}^H \brho) \hat{\brho}_t \| + \| \bdelta(\tilde{\mathbf{e}}) \|_2 \| \brho \|, 
\end{eqnarray}
since the spectral norm is unchanged by Hermitian transpose operation. 
For the spectral norm of $\bdelta(\tilde{\mathbf{e}})$, we use the representation in \eqref{eq:App6_eq2}, and apply the triangle inequality such that
\begin{equation}
\| \bdelta(\tilde{\mathbf{e}}) \| \leq \| \bdelta(\brho_e \brho_e^H) \|  + |\lambda_1| \| \bdelta(\v_1 \v_1^H) \| + |\lambda_2|  \| \bdelta(\v_2 \v_2^H)\|,
\end{equation}
where again, $\brho_e \hat{\brho}_t^H + \hat{\brho}_t \brho_e^H = \sum_{i =1}^2 \lambda_i \v_i \v_i^H$, with $\| \v_i \| = 1$. We obtain the identical form in \eqref{eq:App6_eq5} such that
\begin{equation}
\| \bdelta(\tilde{\mathbf{e}}) \| \leq \delta_1 (\| \brho_e \|^2 + 2 \| \brho_t \| \| \brho_e \| ) \leq \delta_1 (2 + \epsilon) \| \brho_e \|,
\end{equation}
which overall yields
\begin{eqnarray}
\| \nabla \mathcal{J}(\brho) \| & \leq & \| \| \brho \|^2 \brho - (\hat{\brho_t}^H \brho) \hat{\brho}_t \| +  \delta_1 (2 + \epsilon) \| \brho_e \| \| \brho \| \\
& \leq & \| \brho \| \left( (\| \brho \| + \| \brho_t \|) \| \brho - \hat{\brho}_t \|  + \delta_1 (2 + \epsilon) \| \brho_e \| \right). 
\end{eqnarray}
Using the fact that $\brho \in E(\epsilon)$, and $\hat{\brho}_t = e^{\mathrm{i} \phi(\brho)} \brho_t$,  we know that $ 1- \epsilon \leq \| {\brho} \| \leq 1+\epsilon$ from triangle and reverse triangle inequalities. 
Hence,
\begin{eqnarray}
 \| \nabla \mathcal{J}(\brho) \| & \leq & (1 + \epsilon) \left( (2 + \epsilon) \| \brho_e \|  + \delta_1 (2 + \epsilon) \| \brho_e \|  \right) \\
 & \leq & (1 + \epsilon) (1 + \delta_1) (2 + \epsilon) \| \brho - \hat{\brho}_t \|.
\end{eqnarray}
Thereby, setting $c =  (1 + \epsilon) (1 + \delta_1) (2 + \epsilon)$ the right-hand-side of the regularity condition is upper bounded as
\begin{equation}
\frac{1}{\alpha} \| \brho - \hat{\brho}_t \|^2 + \frac{1}{\beta} \|  \nabla \mathcal{J}(\brho) \|^2 \leq  (\frac{1}{\alpha} + \frac{c^2}{\beta}) \| \brho - \hat{\brho}_t \|^2 
\end{equation}
and the regularity condition \eqref{eq:RegularityCond} is established if the following condition holds:
\begin{equation}
\text{Re} \left( \langle \nabla \mathcal{J}(\brho),  (\brho - e^{\mathrm{i} \Phi(\brho)} \brho_t) \rangle \right) \geq  (\frac{1}{\alpha} +  \frac{c^2}{\beta}) \ \text{dist}^2( \brho, \brho_t ). 
\end{equation}

\subsection{Proof of Lemma \ref{lem:Lemma8}}\label{sec:App8}
Indeed, the form that the regularity condition is nothing but the \emph{restricted strong convexity condition.} Since $\nabla \mathcal{J}(\brho_t e^{\mathrm{i} \Phi(\brho)}) = 0 $ for any $\Phi(\brho) = \Phi_0 \in [0, 2\pi)$, by definition, one can equivalently write (\ref{eq:RegCond_new}) as
\begin{equation}
\text{Re} \left( \langle \left(\nabla \mathcal{J}(\brho) - \nabla \mathcal{J} (\brho_t e^{\mathrm{i} \Phi(\brho) }) \right),  (\brho - {\brho}_t e^{\mathrm{i} \Phi(\brho)}) \rangle \right) \geq \eta \| \brho -  {\brho}_t e^{\mathrm{i} \Phi(\brho)} \|^2,
\end{equation} 
where $\eta =  (\frac{1}{\alpha} +  \frac{c^2}{\beta})$. 
Reorganizing the terms, we have
\begin{equation}\label{eq:App8_eqn1}
\text{Re} \left( \langle \left(\nabla \mathcal{J}(\brho) - \nabla \mathcal{J} (\brho_t e^{\mathrm{i} \Phi(\brho) }) \right) - \eta (\brho -  {\brho}_t e^{\mathrm{i} \Phi(\brho)}),  (\brho - {\brho}_t e^{\mathrm{i} \Phi(\brho)}) \rangle \right) \geq 0. 
\end{equation}
Letting $ g(\brho) = \mathcal{J}(\brho) - \frac{\eta}{2} \| \brho \|^2$, we can write \eqref{eq:App8_eqn1} as:
\begin{equation}\label{eq:App8_eqn2}
\text{Re} \left( \langle \left(\nabla g(\brho) - \nabla g(\brho_t e^{\mathrm{i} \Phi(\brho) }) \right),  (\brho - {\brho}_t e^{\mathrm{i} \Phi(\brho)}) \rangle \right) \geq 0.
\end{equation}
For any $\brho$ in the $\epsilon$-ball of $\brho_t e^{\mathrm{i} \Phi(\brho)}$, \eqref{eq:App8_eqn2}
is merely the local convexity condition for $g$ at point $\hat{\brho}_t = \brho_t e^{\mathrm{i} \Phi(\brho) }$. Since the $\epsilon$-ball around $\hat{\brho}_t$ is a convex set, we can use the equivalent condition
\begin{equation}
g(\brho) \geq g( \brho_t e^{\mathrm{i} \Phi(\brho) } ) + \text{Re} \left( \nabla g(\brho_t e^{\mathrm{i} \Phi(\brho) })^H ( \brho - \brho_t e^{\mathrm{i} \Phi(\brho) } ) \right).
\end{equation}
Plugging in the definition for $g$, we obtain
\begin{equation}
\mathcal{J}(\brho) \geq \mathcal{J}(\brho_t e^{\mathrm{i} \Phi(\brho) } ) + \text{Re} \left( \nabla \mathcal{J}(\brho_t e^{\mathrm{i} \Phi(\brho) })^H ( \brho - \brho_t e^{\mathrm{i} \Phi(\brho) } ) \right) + \frac{\eta}{2} \|  \brho - \brho_t e^{\mathrm{i} \Phi(\brho) } \|^2.
\end{equation}
Since $\brho_t e^{\mathrm{i} \Phi(\brho) }$ is a global minimizer of $\mathcal{J}$, it satisfies the first order optimality condition with $\nabla \mathcal{J}(\brho_t e^{\mathrm{i} \Phi(\brho)}) = 0$, and the minimum it attains is $0$. 
Hence the condition reduces to
\begin{equation}
\mathcal{J}(\brho) \geq \frac{\eta}{2}  \text{dist}^2 (\brho, {\brho}_t).
\end{equation}


\section{Lemmas for Theorem \ref{thm:Theorem2}}\label{sec:AppThm2}

\subsection{Proof of Lemma \ref{lem:Lemma9}}\label{sec:App9}
For the intermediate stage$, \mathbf{Y}$, of the spectral estimate, we write
\begin{equation}\label{eq:App9_eqn1}
\mathbf{Y} = \frac{1}{M}\mathcal{F}^H \mathcal{F} (\brho_t \brho_t) = \frac{1}{M} \sum_{m = 1}^M ( (\mathbf{L}_{i}^m )^H \brho_t \brho_t^H \mathbf{L}_j^m ) \mathbf{L}_{i}^m (\mathbf{L}_j^m)^H. 
\end{equation}
Reorganizing the terms in \eqref{eq:App9_eqn1} and taking the expectation, we have
\begin{equation}\label{eq:App9_eqn2}
\mathbb{E}[\mathbf{Y}] =  \frac{1}{M} \sum_{m = 1}^M \mathbb{E} [ (\mathbf{L}_{i}^m (\mathbf{L}_{i}^m )^H \brho_t ) (\mathbf{L}_j^m (\mathbf{L}_j^m)^H \brho_t)^H].
\end{equation}
For fixed $m$, and having $\brho_t$ is independent of sampling vectors, the $N \times N$ matrix inside the summation has the entries of the form,
\begin{equation}\label{eq:App9_eqn5}
\sum_{n = 1}^N \sum_{{n'} = 1}^N \mathbb{E}[(\mathbf{L}_{i}^m)_k (\overline{\mathbf{L}_{i}^m})_n  (\overline{\mathbf{L}_{j}^m})_l (\mathbf{L}_{j}^m)_{n'}] {\brho_t}_n \overline{\brho_t}_{n'},
\end{equation}
where $k, l$ denote the row and column indexes, respectfully. 
Since $\mathbf{L}_i^m$ and $\mathbf{L}_j^m$ are independent of each other and have i.i.d. entries, the $4^{th}$ moments of Gaussian entries are removed as $\mathbb{E}[(\mathbf{L}_{i}^m)_k (\overline{\mathbf{L}_{i}^m})_n  (\overline{\mathbf{L}_{j}^m})_l (\mathbf{L}_{j}^m)_{n'}] = \mathbb{E}[(\mathbf{L}_{i}^m)_k (\overline{\mathbf{L}_{i}^m})_n] \mathbb{E}[(\overline{\mathbf{L}_{j}^m})_l (\mathbf{L}_{j}^m)_{n'}]$, in which the the expectations are only non-zero for $n = k$, $n' = l$, yielding 
\begin{equation}
= \mathbb{E}[|(\mathbf{L}_{i}^m)_k|^2] \mathbb{E}[| ({\mathbf{L}_{j}^m})_l |^2] {\brho_t}_k \overline{\brho_t}_{l}, 
\end{equation}
where $\mathbf{L}_{i}^m, \mathbf{L}_{j}^m \sim \mathcal{N}(0, \frac{1}{2} \mathbf{I} ) + \mathrm{i} \mathcal{N}(0, \frac{1}{2} \mathbf{I})$ have unit variance. 
Hence,
\begin{equation}
\mathbb{E} [ \mathbf{Y}  ]_{k,l} = {\brho_t}_k \overline{\brho_t}_{l},
\end{equation}
which is precisely equal to $\mathbb{E} [ \mathbf{Y} ] = \brho_t \brho_t^H$. 

\subsection{Proof of Lemma \ref{lem:Lemma10}}\label{sec:App10}
We use the machinery in the proof of the concentration bound of the Hessian in \cite{candes2015phase}.\footnote{Corresponds to Lemma 7.4 in the source material.}
By unitary invariance, we take $\brho_t = \mathbf{e}_1$, where $\mathbf{e}_1$ is the first standard basis vector. We want to establish that
\begin{equation}\label{eq:App10_eqn1}
\| \frac{1}{M} \sum_{m = 1}^M \overline{({\mathbf{L}_i^m})_1} ({\mathbf{L}_j^m})_1 \mathbf{L}_i^m (\mathbf{L}_j^m)^H - \e_1 \e_1^T \| \leq \delta.
\end{equation}
The inequality in \eqref{eq:App10_eqn1} is equivalent to
\begin{eqnarray}
I_0(\y) & := & \big| \y^H \left(\frac{1}{M} \sum_{m = 1}^M \overline{({\mathbf{L}_i^m})_1} ({\mathbf{L}_j^m})_1 \mathbf{L}_i^m (\mathbf{L}_j^m)^H - \e_1 \e_1^T  \right) \y \big| \\
& = & \big| \frac{1}{M} \sum_{m = 1}^M \overline{({\mathbf{L}_i^m})_1} ({\mathbf{L}_j^m})_1 (\y^H \mathbf{L}_i^m  (\mathbf{L}_j^m)^H \y) - | \y_1 |^2  \big| \leq \delta,
\end{eqnarray}
for any $\y \in \mathbb{C}^N$ obeying $\| \y \|  = 1$. 
Letting $\y = [\y_1, \tilde{\y}]$ where $\tilde{\y} \in \mathbb{C}^{N-1}$, and similarly partitioning the sampling vectors as $\mathbf{L}_{i}^m = [(\mathbf{L}_{i}^m)_1, \tilde{\mathbf{L}}_{i}^m]$ we have
\begin{equation}
\begin{split}
\y^H \mathbf{L}_i^m  (\mathbf{L}_j^m)^H \y & = (\mathbf{L}_i^m )_1 \overline{\y_1} \overline{(\mathbf{L}_j^m )_1} \y_1 +  (\mathbf{L}_i^m )_1 \overline{\y_1} (\tilde{\mathbf{L}}_{j}^m)^H \tilde{\y} \\
& + \tilde{\y}^H \tilde{\mathbf{L}}_{i}^m \overline{(\mathbf{L}_j^m )_1} \y_1 +  \tilde{\y}^H \tilde{\mathbf{L}}_{i}^m (\tilde{\mathbf{L}}_{j}^m)^H \tilde{\y}.
\end{split}
\end{equation}
This yields 
\begin{eqnarray}
I_0(\y) & = & \big| \frac{1}{M} \sum_{m = 1}^M ( |(\mathbf{L}_{i}^m)_1|^2 | (\mathbf{L}_{j}^m)_1|^2 -1 ) | \y_1 |^2 + |(\mathbf{L}_{i}^m)_1|^2 (\mathbf{L}_{j}^m)_1 \overline{\y_1} (\tilde{\mathbf{L}}_{j}^m)^H \tilde{\y} \cdots \nonumber \\
\cdots & + &\  |(\mathbf{L}_{j}^m)_1|^2 \overline{(\mathbf{L}_i^m )_1} \y_1 \tilde{\y}^H \tilde{\mathbf{L}}_{i}^m + \overline{({\mathbf{L}_i^m})_1} ({\mathbf{L}_j^m})_1 \tilde{\y}^H \tilde{\mathbf{L}}_{i}^m (\tilde{\mathbf{L}}_{j}^m)^H \tilde{\y} \big| \nonumber \\
& \leq & \big| \frac{1}{M} \sum_{m = 1}^M ( |(\mathbf{L}_{i}^m)_1|^2 | (\mathbf{L}_{j}^m)_1|^2 -1 ) | \y_1 |^2 \big| + \big| \frac{1}{M} \sum_{m = 1}^M |(\mathbf{L}_{i}^m)_1|^2 (\mathbf{L}_{j}^m)_1 \overline{\y_1} (\tilde{\mathbf{L}}_{j}^m)^H \tilde{\y} \big| \nonumber \\
& + &\big| \frac{1}{M} \sum_{m = 1}^M |(\mathbf{L}_{j}^m)_1|^2 \overline{(\mathbf{L}_i^m )_1} \y_1 \tilde{\y}^H \tilde{\mathbf{L}}_{i}^m \big| + \big| \frac{1}{M} \sum_{m = 1}^M \overline{({\mathbf{L}_i^m})_1} ({\mathbf{L}_j^m})_1 \tilde{\y}^H \tilde{\mathbf{L}}_{i}^m (\tilde{\mathbf{L}}_{j}^m)^H \tilde{\y} \big|. 
\end{eqnarray}
Due to the independence $\mathbf{L}_i^m$ and $\mathbf{L}_j^m$, and from the fact that they are zero mean, unit variance i.i.d. Gaussian entry vectors, all four terms on the right-hand-side are measures of distance to expected values. 

For the second and third terms, we have 
\begin{equation}
\begin{split}
\big| \frac{1}{M} \sum_{m = 1}^M |(\mathbf{L}_{i}^m)_1|^2 (\mathbf{L}_{j}^m)_1 \overline{\y_1} (\tilde{\mathbf{L}}_{j}^m)^H \tilde{\y} \big|  \leq \frac{1}{M} \sqrt{\sum_{m = 1}^M \big| |(\mathbf{L}_{i}^m)_1|^2 (\mathbf{L}_{j}^m)_1 \overline{\y_1} \big|^2 } \sqrt{\sum_{m = 1}^M \big| (\tilde{\mathbf{L}}_{j}^m)^H \tilde{\y} \big|^2}  \\
\leq | {\y_1} |  \frac{1}{M}  \sqrt{N \sum_{m = 1}^M  |(\mathbf{L}_{i}^m)_1|^4 |(\mathbf{L}_{j}^m)_1|^2 } \left( \frac{1}{\sqrt{N}}\sum_{m = 1}^M \big| (\tilde{\mathbf{L}}_{j}^m)^H \tilde{\y} \big| \right)
\end{split}
\end{equation}
where the inequalities follow from Cauchy-Schwartz and $\ell_2 \leq \ell_1$. We then invoke Hoeffding's inequality from Proposition 10 in \cite{vershynin2010introduction}. For any $\delta_0$ and $\gamma$, there exists a constant $C(\delta_0, \gamma)$ such that $M \geq C(\delta_0, \gamma) \sqrt{N \sum_{m = 1}^M  |(\mathbf{L}_{i}^m)_1|^4 |(\mathbf{L}_{j}^m)_1|^2 }$ where
\begin{equation}
\big| \frac{1}{M} \sum_{m = 1}^M |(\mathbf{L}_{i}^m)_1|^2 (\mathbf{L}_{j}^m)_1 \overline{\y_1} (\tilde{\mathbf{L}}_{j}^m)^H \tilde{\y} \big|   \leq \delta_0 \| \tilde{\y} \| | {\y_1} | \leq \delta_0,
\end{equation}
with probability $1 - 3 e^{- 2 \gamma N}$. For the final term, we invoke the Bernstein-type inequality of Proposition 16 in \cite{vershynin2010introduction} per \cite{candes2015phase}, such that for any positive constants $\delta_0$, $\gamma$, there exists the constant $C(\delta_0, \gamma)$ with
\begin{equation}
M \geq C(\delta_0, \gamma) \left( \sqrt{N \sum_{m = 1}^M  |(\mathbf{L}_{i}^m)_1|^2 |(\mathbf{L}_{j}^m)_1|^2 } + N \underset{m = 1 \cdots M}{max} \ |(\mathbf{L}_{i}^m)_1| |(\mathbf{L}_{j}^m)_1| \right)
\end{equation}
where
\begin{equation}
\big| \frac{1}{M} \sum_{m = 1}^M \overline{({\mathbf{L}_i^m})_1} ({\mathbf{L}_j^m})_1 \tilde{\y}^H \tilde{\mathbf{L}}_{i}^m (\tilde{\mathbf{L}}_{j}^m)^H \tilde{\y} \big| \leq \delta_0 \| \tilde{\y} \|^2 \leq \delta_0
\end{equation}  
with probability $1 - 2 e^{-2 \gamma N}$. 

To control the remaining terms, we use Chebyshev's inequality, 
per \cite{candes2015phase}. 
For any $\epsilon_0 > 0$ there exists a constant $C$ with $M \geq C \cdot N$ such that the following hold
\begin{equation}
\big| \frac{1}{M} \sum_{m = 1}^M ( |(\mathbf{L}_{i}^m)_1|^2 | (\mathbf{L}_{j}^m)_1|^2 -1 ) | \y_1 |^2 \big| \leq  \epsilon_0 | \y_1 |^2,
\end{equation}
\begin{equation}
\big| \frac{1}{M} \sum_{m = 1}^M ( |(\mathbf{L}_{i}^m)_1|^4 | (\mathbf{L}_{j}^m)_1|^2 -1 )\big| \leq \epsilon_0,  \ \big| \frac{1}{M} \sum_{m = 1}^M ( |(\mathbf{L}_{i}^m)_1|^2 | (\mathbf{L}_{j}^m)_1|^4 -1 )\big| \leq \epsilon_0 
\end{equation}
with probability at least $1 - 3N^{-2}$. Moreover, from union bound we have 
\begin{equation}
\underset{m = 1 \cdots M}{\text{max}} \ |(\mathbf{L}_{i,j}^m)_1| \leq \sqrt{10 \log M}
\end{equation}
with probability at least $1 - 2N^{-2}$. 
As in \cite{candes2015phase}, we denote the event that the results from Chebyshev's inequality hold by $E_0$. Then, on event $E_0$, combining all the terms, the inequality
\begin{equation}
I_0(\y)  \leq \epsilon_0 | \y_1 |^2 + \delta_0 | y_1|  \| \tilde{\y} \| + \delta_0 \| \tilde{\y} \|^2 \leq \epsilon_0 + 2 \delta_0,
\end{equation}
holds with probability at least $1- 8e^{-2 \gamma N}$. We then follow by the $\epsilon$-net argument of \cite{candes2015phase} via Lemma 5.4 in \cite{vershynin2010introduction} to bound the operator norm such that
\begin{equation}
\underset{\y \in \mathcal{S}_{\mathbb{C}^N}}{\text{max}} I_0(\y) \leq \underset{\y \in \mathcal{N}}{\text{max}} I_0(\y) \leq  2 \epsilon_0 + 4 \delta_0,
\end{equation}
where $\mathcal{S}_{\mathbb{C}^N}$ is the unit sphere in $\mathbb{C}^N$ and $\mathcal{N}$ is an $1/4$-net of $\mathcal{S}_{\mathbb{C}^N}$. 
Then, choosing appropriate $\epsilon_0$, $\delta_0$ and $\gamma$, and applying the union bound we have
\begin{equation}
\| \frac{1}{M} \sum_{m = 1}^M \overline{({\mathbf{L}_i^m})_1} ({\mathbf{L}_j^m})_1 \mathbf{L}_i^m (\mathbf{L}_j^m)^H - \e_1 \e_1^T \| \leq \delta
\end{equation}
with probability $1 - 8e^{- \gamma N}$ for $\delta = 2 \epsilon_0 + 4 \delta_0$, and 
\begin{equation}
\begin{split}
M \geq C' \Biggl( \sqrt{N \sum_{m = 1}^M  |(\mathbf{L}_{i}^m)_1|^4 |(\mathbf{L}_{j}^m)_1|^2 } & + \sqrt{N \sum_{m = 1}^M  |(\mathbf{L}_{i}^m)_1|^2 |(\mathbf{L}_{j}^m)_1|^2 } \\
& + N \underset{m = 1 \cdots M}{max} \ |(\mathbf{L}_{i}^m)_1| |(\mathbf{L}_{j}^m)_1| \Biggr) .
\end{split}
\end{equation}
From $E_0$, we have $M \geq C \cdot N$ which gives $M = \mathcal{O}(N \log N)$, where overall event holds with probability at least $1 - 8e^{- \gamma N} - 5N^{-2}$, hence, the proof is complete. 




\end{document}